\documentclass[12pt]{article}
\usepackage{xspace}
\usepackage[colorlinks=true,linkcolor=red,citecolor=blue]{hyperref}
\usepackage{amssymb,amsmath,amsthm,graphicx,color,bbm}
\usepackage{thmtools, thm-restate}
\usepackage[margin=1.in]{geometry}
\usepackage{multirow,array}
\allowdisplaybreaks

\usepackage[capitalise,nameinlink]{cleveref}

\usepackage{algorithm}
\usepackage[noend]{algpseudocode}

\makeatletter
\newcommand*\rel@kern[1]{\kern#1\dimexpr\macc@kerna}
\newcommand*\widebar[1]{%
  \begingroup
  \def\mathaccent##1##2{%
    \rel@kern{0.8}%
    \overline{\rel@kern{-0.8}\macc@nucleus\rel@kern{0.2}}%
    \rel@kern{-0.2}%
  }%
  \macc@depth\@ne
  \let\math@bgroup\@empty \let\math@egroup\macc@set@skewchar
  \mathsurround\z@ \frozen@everymath{\mathgroup\macc@group\relax}%
  \macc@set@skewchar\relax
  \let\mathaccentV\macc@nested@a
  \macc@nested@a\relax111{#1}%
  \endgroup
}
\makeatother

\theoremstyle{plain}
\newtheorem{theorem}{Theorem}[section]
\newtheorem{corollary}[theorem]{Corollary}

\newtheorem{lemma}[theorem]{Lemma}
\newtheorem{claim}[theorem]{Claim}
\newtheorem{fact}[theorem]{Fact}

\newtheorem*{claimNoNum}{Claim}

\newtheorem{definition}[theorem]{Definition}

\newtheorem{thm}{Theorem}

\theoremstyle{remark}

\theoremstyle{plain}

\def\R{{\mathbb{R}}}
\def\F{{\mathbb{F}}}

\def\N{{\mathbb{N}}}

\renewcommand{\Pr}{\mathop{\bf Pr\/}}
\newcommand{\E}{\mathop{\bf E\/}}

\def\var{{\mathop{\bf Var\/}}}
\def\cov{{\mathop{\bf Cov\/}}}

\newcommand{\Var}{\mathop{\bf Var\/}}

\newcommand{\abs}[1]{\left|#1\right|}

\newcommand{\one}{{\mathbbm{1}}}
\def\B{{\{0,1\}}}
\def\pmone{{\{\pm1\}}}

\def\poly{{\mathrm{poly}}}

\newcommand{\remove}[1]{}

\newcommand{\mathify}[1]{\ifmmode{#1}\else\mbox{$#1$}\fi}

\newcommand{\Rp}{\mathcal{R}_p}

\newcommand{\eps}{\varepsilon}

\newcommand{\tildeO}{{\widetilde O}}

\newcommand{\cE}{\mathcal{E}}

\newcommand{\Dtagx}{{\mathcal{D}'_x}}
\newcommand{\Dx}{{\mathcal{D}_x}}

\newcommand{\D}{{\mathcal{D}}}
\newcommand{\error}{{\mathsf{err}}}

\newcommand{\Vsmall}{V_{\mathsf{small}}}

\renewcommand{\bar}{\widebar}
\renewcommand{\hat}{\widehat}
\renewcommand{\tilde}{\widetilde}

\newcommand{\Bias}{\mathsf{Bias}}

\newcommand{\ignore}[1]{}

\begin{document}

\title{Pseudorandom Generators for Width-3 Branching Programs}

\author{
Raghu Meka\thanks{{\tt raghum@cs.ucla.edu}. Supported by NSF grant CCF-1553605.}\\
\small UCLA\\
\and 
Omer Reingold\thanks{{\tt reingold@stanford.edu}. Supported in part by NSF grant CCF-1763299.
}\\
\small Stanford University\\
\and 
Avishay Tal\thanks{\texttt{avishay.tal@gmail.com}. Supported by a Motwani Postdoctoral Fellowship and by NSF grant CCF-1763299.
}\\
\small Stanford University\\
}

\maketitle

\begin{abstract}
We construct pseudorandom generators of seed length $\tilde{O}(\log(n)\cdot \log(1/\eps))$ that $\eps$-fool \emph{ordered} read-once branching programs (ROBPs) of width $3$ and length $n$.
For \emph{unordered} ROBPs, we construct pseudorandom generators with seed length $\tilde{O}(\log(n) \cdot \poly(1/\eps))$.
This is the first improvement for pseudorandom generators fooling width $3$ ROBPs since the work of Nisan~\cite{Nisan92}.

Our constructions are based on the ``iterated milder restrictions'' approach of \cite{GopalanMRTV12} (which further extends the Ajtai-Wigderson framework~\cite{AW85}), combined with the INW-generator \cite{INW94} at the last step (as analyzed by \cite{BravermanRRY14}). For the unordered case we combine iterated milder restrictions with the  generator of \cite{CHHL18}.

Two conceptual ideas that play an important role in our analysis are:
\begin{enumerate} 
\item A \emph{relabeling technique} allowing us to analyze a relabeled version of the given branching program, which turns out to be much easier.
 \item 	Treating the number of colliding layers in a branching program as a progress measure and showing that it reduces significantly under pseudorandom restrictions.
 \end{enumerate}

In addition, we achieve nearly optimal seed-length $\tilde{O}(\log(n/\eps))$ for the classes of: (1) read-once polynomials on $n$ variables, (2) locally-monotone ROBPs of length $n$ and width $3$ (generalizing read-once CNFs and DNFs), and (3) constant-width ROBPs of length $n$ having a layer of width $2$ in every consecutive $\poly\log(n)$ layers.
\end{abstract}

%
%
%

\thispagestyle{empty}
\newpage
{\small
\tableofcontents
}
\thispagestyle{empty}
\newpage

\clearpage 
\setcounter{page}{1}


\section{Introduction}
A central challenge in complexity theory is to understand the trade-off between \emph{space} and \emph{randomness} as resources and in particular, whether $\mathsf{BPL} = \mathsf{L}$. One of the main techniques we have for approaching this question is to design pseudorandom generators that fool tests computable in small space. The latter question can be elegantly captured in the language of designing pseudorandom generators for {\sf read-once branching programs}; we define these objects next.

\begin{definition}
For $w,n \in \N$, a {\sf read-once branching program} {\sf (ROBP)} of {\sf width $w$} and {\sf length $n$} is a layered directed graph $B$ with $n+1$ layers where all but the first layer have at most $w$ nodes, the first layer has a single vertex designated the {\sf start} vertex, and the vertices in the last layer are either labeled {\sf accept} or {\sf reject}. Each vertex in the first $n$ layers has exactly two outgoing edges to vertices in the next layer with one labeled $1$ and the other labeled $-1$. 

Given a ROBP as above, it defines a function $B:\pmone^n \to \pmone$ naturally where on input $x \in \pmone^n$ starting from the start vertex, you follow the edges labeled by $x_i$ for $1 \leq i \leq n$ and output $-1$ if the last vertex reached is accepting and $1$ otherwise. 
\end{definition}

Derandomizing space-bounded computations is fundamentally related to designing pseudorandom generators (and \emph{hitting set generators}) for ROBPs as above. 

\begin{definition}
Given a class of functions $\mathcal{F} = \{f: \pmone^n \rightarrow \mathbb{R}\}$, a function $G: \pmone^r \rightarrow \pmone^n$ is a {\sf pseudorandom generator (PRG)} with {\sf error} $\eps$ (or $\eps$-{\sf fools}) $\mathcal{F}$ if for every $f \in \mathcal{F}$, 
$$\left|\Pr_{x \in_u \pmone^n}[f(x)] - \Pr_{y \in_u \pmone^r}[f(G(y))] \right| \leq \eps.$$

We say the generator is {\sf log-space explicit} if $G$ can be computed in space logarithmic in the output length $n$ and refer to $r$ as the {\sf seed-length} of the generator. 
\end{definition}

It is well-known by now that if there exists a log-space explicit PRG (or even a hitting set generator) with constant error that fools ROBPs of width $n$ and length $n$ with seed-length $O(\log n)$, then $\mathsf{BPL} = \mathsf{L}$. In this vein, a seminal result of Nisan \cite{Nisan92} gave a log-space explicit PRG that $\eps$-fools ROBPs of width $w$ and length $n$ with seed-length $r = O((\log n) \cdot \log(w n/\eps))$. Despite significant attention, improving Nisan's PRG has been a fundamental bottleneck in pseudorandomness. For width $w = 2$, it is known that {\sf small-bias spaces} fool width two ROBPs (\cite{SaksZuckerman95,BogdanovDVY13}), leading to a PRG with seed-length $O(\log(n/\eps))$. However, even for the case of $\eps$ a constant and width $w = 3$, the best provable PRG had seed-length $O(\log^2 n)$--no better than what Nisan's PRG gives for polynomial width ROBPs. Nearly optimal hitting-sets generators for width-3 ROBPs were given in \cite{SimaZ11,GopalanMRTV12} while \cite{BravermanRRY14, KouckyNP11, De11} obtained PRGs with nearly optimal seed-length for special-classes of constant-width ROBPs. 
In this work, we obtain the first improvement over Nisan's PRG for width-$3$ ROBPs:

\begin{restatable}[Main Theorem]{thm}{thmMainOrdered}
\label{thm:main-3ROBP-ordered}
For any $\eps > 0$, there exists a log-space explicit PRG that $\eps$-fools width-$3$ ROBPs with seed-length%
\footnote{Henceforth, $\tilde{O}(t)$ is used to denote $O(t\cdot \poly\log(t))$.} 
$\tilde{O}(\log(n/\eps)) + O(\log(1/\eps) \cdot \log (n))$. 
\end{restatable}

We in fact also obtain PRG's with nearly optimal dependence for constant error for the bigger class of {\sf unordered} width-$3$ ROBPs, which are functions computable by ROBPs under some unknown permutation (see Section~\ref{prelim:ROBP} for the formal definition). In this regime, we improve the results of \cite{SteinkeVW17} that gave a PRG with seed-length $\tilde{O}(\log^3 n)$.

\begin{restatable}{thm}{thmMainUnordered}\label{thm:main-3ROBP-unordered}
For any $\eps > 0$, there exists a log-space explicit PRG that $\eps$-fools  unordered width-$3$ ROBPs with seed-length $\tilde{O}(\log(n/\eps)) + O( \poly(1/\eps) \cdot \log (n))$.
\end{restatable}

A special class of unordered width-$4$ ROBPs that have received recent attention are {\sf read-once polynomials} (see \cite{Tre10,LeeViola17}) for which we give a PRG with nearly optimal seed-length both in terms of the error and input length (up to $\poly(\log \log)$ factors):

\begin{restatable}{thm}{thmReadOncePoly}\label{thm:read-once poly}
There exists a log-space explicit $\eps$-PRG for the class of read-once polynomials on $n$ variables with seed-length $\tilde{O}(\log(n/\eps))$.
\end{restatable}

In comparison, the best previous PRG for read-once polynomials had seed-length $\tilde{O}(\log(n/\eps)) \cdot \log(1/\eps)$, thus in particular needed $\tilde{O}(\log^2 n)$ seed-length to fool read-once polynomials with polynomially small error. 

Our results rely on several new conceptual ideas as well as technical ingredients, including PRGs fooling other interesting intermediate classes of ROBPs, that we believe could be useful for other applications especially in the context of obtaining PRGs for constant-width ROBPs. Our results rely on the framework of \emph{iterative mild random restrictions} introduced in \cite{GopalanMRTV12} and further developed in \cite{ReingoldSV13, SteinkeVW17, GopalanY14, GKM15, HLV17, LeeViola17}, the latter two also present an elegant alternate view of the technique as \emph{bounded-independence plus noise}. We describe this framework, our proof techniques next.


\subsection{The Ajtai-Wigderson framework}

The Ajtai-Wigderson \cite{AW85} framework, that was revived and refined for ROBPs in the work of Gopalan,  Meka, Reingold, Trevisan, Vadhan~\cite{GopalanMRTV12}, provides a ``recipe'' for constructing PRGs for classes of functions that simplify under (pseudo)random restrictions. 
Roughly speaking, in order to fool a class of functions $\mathcal{C}$ it suffices to fool $\mathcal{C}$ under pseudo-random restrictions keeping each variable alive with probability $p$. 
Equivalently, it suffices to pseudorandomly assign $p$-fraction of the coordinates while approximately preserving the acceptance probability (on average) of every function $f\in \mathcal{C}$.
Suppose we have such a pseudorandom partial assignment, and assume that the class of functions $\mathcal{C}$ is closed under restrictions. Then, iteratively applying a pseudorandom partial assignment on the remaining coordinates until we assigned all of them gives us a pseudorandom generator for $\mathcal{C}$.
 We expect to assign all the coordinates after $O(p^{-1} \cdot \log n)$ iterations, thus if each iteration requires at most $s$ random bits, we get a PRG with seed-length $O(s\cdot p^{-1} \cdot \log n)$. Naively, it seems impossible to achieve nearly-logarithmic seed length  using this approach, however this was obtained in the work of \cite{GopalanMRTV12} as explained next.\looseness=-1

\paragraph{Achieving Near-Logarithmic Seed-Length.}
In the work of \cite{GopalanMRTV12} the Ajtai-Wigderson approach was used to construct $\eps$-PRGs for read-once CNFs (and read-once DNFs) with seed length $\tilde{O}(\log(n/\eps))$. 
In order to achieve nearly-logarithmic seed-length \cite{GopalanMRTV12} showed that one can assign a constant fraction of the coordinates while preserving the acceptance probability up to error $\poly(\eps/n)$ using only $s = \tilde{O}(\log(n/\eps))$ bits of randomness. Plugging into the estimates above would give naively seed-length $\tilde{O}(\log(n/\eps)\cdot \log(n))$. In order to avoid the  additional factor of $\log(n)$, they prove that after pseudorandomly assigning all but $1/\poly\log(n)$ of the coordinates, the function simplifies significantly so that it can be fooled using additional $O(\log (n/\eps))$-random bits.

We describe the approach more precisely. A $p$-pseudorandom restriction against a class of functions $\mathcal{C}$ specifies a set $T \subseteq [n]$ of roughly $p \cdot n$ of the coordinates, and an assignment $x\in \pmone^{T}$ to these coordinates, such that for any $f\in \mathcal{C}$:
$$\E_{T,x}\E_{y\in \pmone^{[n]\setminus T}}[f(x\circ y)] = \E_{z\in \pmone^n}[f(z)] \pm \eps$$
where $(x \circ y)$ denotes the string whose $T$-coordinates are taken from $x$ and other coordinates are taken from $y$. 
The main observation of \cite{GopalanMRTV12} is that given $T$, it suffices that $x$ would fool the {\sf Bias-function}, defined as
\[
\Bias_T f(x) \triangleq \E_{y\in \pmone^{[n]\setminus T}} [f(x\circ y)].
\]
This is due to the fact that 
\[
\Big|\E_{z\in \pmone^n}[f(z)] - \E_{x} \E_{y\in\pmone^{[n]\setminus T}} [f(x\circ y)]\Big|
= \Big|\E_{z\in \pmone^T} [\Bias_T f(z)] - \E_{x}[\Bias_T f(x)]\Big|.
\]
The observation that it suffices to fool the bias-function instead of just fooling the restricted functions, enabled \cite{GopalanMRTV12} to use ``mild'' restrictions with $p = \Omega(1)$ for the class of CNFs/DNFs. They show that in this case, the average of the restricted functions (i.e., the bias-function) is much easier to fool than a typical restricted function.

\section{Proof Overview}
Similarly to \cite{GopalanMRTV12}, in order to achieve a PRG with nearly-logarithmic seed-length fooling width-3 ROBPs, we show that:
\begin{enumerate}
	\item We can pseudorandomly assign half the input coordinates while preserving the acceptance probability (on average) of every width-3 ROBP up to error $\eps$, using seed-length $\tilde{O}(\log(n/\eps))$.
	\item After pseudorandomly assigning all but $1/\poly\log(n)$ of the coordinates any width-3 ROBP simplifies enough so that it can be fooled using additional $\tilde{O}(\log(n)\log(1/\eps))$ random bits.
\end{enumerate}
Both steps are involved and explained in greater detail in the next two sections.

\subsection{Pseudorandomly assigning half of the coordinates}
In Sections~\ref{sec:pseudorestriction} and \ref{sec:XOR_of_short} we prove the following theorem showing that we can pseudorandomly assign $1/\poly\log\log(n/\eps)$ of the coordinates while changing the acceptance probability by at most $\eps$.
\begin{restatable}{thm}{twosteps}
\label{thm:main_two_steps}
Let $n\in \N, \eps>0$. There exists a log-space explicit pseudorandom restriction assigning $p = 1/O(\log \log(n/\eps))^{6}$ fraction of the variables using $O(\log (n/\eps) \log\log(n/\eps))$ random bits, that maintains the acceptance probability of any unordered width-$3$ length-$n$ ROBP up to error $\eps$.
\end{restatable}

Given Theorem~\ref{thm:main_two_steps}, we can assign half of the coordinates by iteratively applying the pseudorandom restriction $O(1/p)$ times. This ultimately uses $O(\log (n/\eps) (\log\log(n/\eps))^7) = \tilde{O}(\log(n/\eps))$ random bits to assign half of the coordinates, as promised.

We describe the techniques that go into the proof of Theorem~\ref{thm:main_two_steps}.
The proof proceeds in two steps. The first step (described in Section~\ref{sec:pseudorestriction}) reduces the task of generating a pseudorandom restriction for width-3 ROBPs to the task of generating a pseudorandom restriction for the XOR of short (logarithmic-length) width-3 ROBPs.
The second step (described in Section~\ref{sec:XOR_of_short}) is a pseudorandom restriction for the latter class of Boolean functions.

\subsubsection{Reducing width-3 ROBPs to the XOR of short width-3 ROBPs}
Next, we explain how we reduce fooling width-3 ROBPs to fooling the XOR of short width-3 ROBPs. 
Let $B$ be a ROBP of length-$n$ and width-$3$.
We pick a set $T_0 \subseteq[n]$ of size $\approx n/2$ using an almost $O(\log(n/\eps))$-wise independent  distribution.
We wish to show that for most choices of $T_0$, we can pseudorandomly assign $pn$ of the coordinates in $T_0$, while fooling the Bias-function $\Bias_{T_0} B$.
Our \emph{main observation}	 is that for most choices for $T_0$, the bias-function $\Bias_{T_0} B$ is the average of simpler width-3 ROBPs.

Recall that every layer of edges in a ROBP contains two sets of edges, one corresponding to the transition made when the input bit equals $1$ and similarly one corresponding to the input bit equaling $-1$.
Observe that if the two sets of edges are the same, then the layer is {\sf redundant} and the value of the input bit  does not affect whether the ROBP accept or not.
We thus assume without loss of generality that there are no  redundant layers.
We say that a layer of edges is a {\sf colliding layer} if there are two edges marked by the same label (i.e. both labeled $1$ or both labeled $-1$) that enter the same vertex in the next layer.

First, suppose (ideally) that all layers in a width-3 ROBP are colliding. 
Then, under the pseudorandom restriction, with high probability, in every $O(\log (n/\eps))$ consecutive layers we will have a layer of edges whose corresponding variable is fixed to a value for which the edges in the layer collide, leaving at most $2$ vertices reachable in the next layer of vertices. 
Using a result of Bogdanov, Dvir, Verbin, Yehudayoff~\cite{BogdanovDVY13} such restricted ROBPs can be written as linear combinations of functions of the following form: XOR of width-3 ROBPs of length $O(\log(n/\eps))$ defined over disjoint sets of variables. It thus suffices to fool this XOR of short width-3 ROBPs in order to fool the restricted ROBP, as we do in Section~\ref{sec:XOR_of_short}.

The assumption that all layers in a width-3 ROBP are colliding is not necessarily true. In fact, it can be the case that in every layer of edges both the $1$-edges and the $(-1)$-edges form a permutation on the state space with no collisions. Indeed, such ROBPs are known in the literature as permutation-ROBPs. (For example, the $\text{MOD}_3(x_1, \ldots,x_n)$ function indicating whether $(\sum_{i}{x_i} \equiv 0 \mod 3)$ can be computed by width-3 permutation ROBP.)
Nonetheless, as mentioned earlier, it suffices to fool the bias-function and this task is easier than fooling each restricted function.

\paragraph{Relabeling Under The Bias Function:} 
In the following, we consider relabeling of a ROBP. Recall that in a ROBP every vertex has a pair of outgoing edges: one labeled $1$ and the other labeled $-1$. A {\sf relabeling} of a ROBP $B$ is any ROBP $B'$ that can be achieved from $B$ by swapping the labels for some of these pairs of edges.

Our \emph{key observation} is that the bias function $\Bias_T B$ of a program $B$ does not depend on the labels of the edges associated with the variables outside $T$. 
This is due to the fact that the value of $\Bias_T B(x)$ on a given partial input $x \in \pmone^T$ is the probability of acceptance of $B$ on a random assignment to the variables in $[n]\setminus T$, and this value remains the same under any relabeling of the edges associated with the variables in $[n] \setminus T$.
Moreover, a simple fact shows that any {\sf non-redundant} layer of edges can be relabeled so that it is {\sf colliding}.
Thus, for any ROBP $B$ and any fixed $T$, we can relabel the edges associated with variables with $[n]\setminus T$ so that they are colliding, yielding another width-3 ROBP, denoted $B^T$.
We get that $\Bias_T B = \Bias_T B^T$, and $B^T$ is a ROBP in which all layers in $[n]\setminus T$ are colliding. We can thus apply the previous argument and conclude that $\Bias_T B^T$ is the average of width-3 ROBPs whose vast majority have a layer of vertices of width-$2$ in every $O(\log (n/\eps))$ consecutive layers. These ROBPs are then fooled by the pseudorandom partial assignment described in the next section.

To sum up, since the bias-function is the average over all restricted functions of $B$, it also equals the average over all restricted functions of $B^T$, and these restricted functions are simple enough for us to fool.

Relabeling was previously used in~\cite{BrodyV10,Steinberger13,CGR14} to show that the best ROBPs distinguishing between certain distributions and the uniform distribution must be ``locally-monotone'' (see Section~\ref{prelim:ROBP} for the formal definition). 
In general, it is unclear how to argue locally monotone programs are the hardest ROBPs to fool. Nevertheless, in \cite{CHRT17}, relabeling helped bounding the sum of absolute values of Fourier coefficients of small width ROBPs.
In comparison, we use a relabeling technique to note that in the iterated random restrictions framework (when trying to fool the bias-function), one might as well treat the restricted layers as if they were locally monotone.

\subsubsection{Pseudorandom restrictions for the XOR of short width-3 ROBPs}
Our main result in Section~\ref{sec:XOR_of_short} is the following:
\begin{restatable}{thm}{mainintro}\label{thm:pseudorandom-restriction-XOR-short}
Let $n,w,b\in \N$, $\eps>0$. There exists a log-space explicit pseudorandom restriction assigning $p = 1/O(\log(b \cdot \log(n/\eps)))^{2w}$ fraction of  $n$ variables using $O(w \cdot \log (n/\eps) \cdot (\log\log(n/\eps) + \log(b)))$ random bits, that maintains the acceptance probability of any XOR of ROBPs of width-$w$ and length-$b$ (defined on disjoint sets of variables) up to error $\eps$.
\end{restatable}
Recall that in the previous section, we reduced the case of width-3 ROBPS to this case with $w = 3$ and $b = O(\log(n/\eps))$.
Our proof for Theorem~\ref{thm:pseudorandom-restriction-XOR-short} follows previous strategies by \cite{GopalanMRTV12, GopalanY14, GKM15, LeeViola17}. 
Indeed, the functions we are trying to fool are a special case of {\sf product-functions} that were recently studied in \cite{HLV17, LeeViola17}.
Product-functions are functions of the form $f(x) = f_1(x) \cdot f_2(x)   \cdots  f_m(x)$ where each $f_i$ depends on a set $B_i$ of at most $b$ variables, and $\{B_1, \ldots, B_m\}$ are pairwise-disjoint.

PRGs for product-functions were constructed in previous work, however none achieve the parameters we need. Haramaty, Lee and Viola \cite{HLV17} and Lee and Viola \cite{LeeViola17} constructed PRGs with seed length $\tilde{O}(b + \sqrt{m b\log(1/\eps)})$ and $\tilde{O}((b + \log(m/\eps))\cdot \log(1/\eps))$ respectively for such functions. While the latter is nearly optimal for constant $\eps$, we require $\eps$ to be smaller than $1/m$, since the reduction in the previous section from \cite{BogdanovDVY13} incurs a multiplicative factor of $m$ on the error.
Gopalan, Meka and Kane \cite{GKM15} achieve nearly optimal seed-length $\tilde{O}(\log(n/\eps))$ but only for the case where the blocks $B_1, \ldots, B_m$ are known.

The main reason we are able to achieve better seed-length is due to the fact that we further assume that the functions $f_1, \ldots, f_m$ are computed by constant-width ROBPs. We rely on the previous work of Chattopadhyay, Hatami, Reingold, Tal~\cite{CHRT17}. They constructed PRGs for constant-width length-$n$ ROBPs with seed-length $\poly\log(n)$. We observe that under an unusual setting of parameters, namely when applying this result to constant-width ROBPs of length $\poly\log(n)$, one gets seed-length $\tilde{O}(\log(n/\eps))$. This enables us to fool the XOR of any subset of $\poly\log(n)$ of the functions $f_1, \ldots, f_m$ using nearly-logarithmic seed-length. Relying on the proof strategy laid by Gopalan and Yehudayoff \cite{GopalanY14}, we bootstrap this into a pseudorandom restriction fooling the XOR of $f_1, \ldots, f_m$.

\subsection{Simplification under pseudorandom restrictions}
Recall that our proof strategy is similar to that of \cite{GopalanMRTV12}: 
\begin{enumerate}
 	\item For $i=0, \ldots, O(\log\log n)$: assign half of the remaining coordinates pseudorandomly using $\tilde{O}(\log(n/\eps))$ random bits, while changing the acceptance probability by at most~$\eps$.
\item   Pseudorandomly assign the remaining coordinates using $\tilde{O}(\log(n/\eps))$ random bits.
 \end{enumerate}

The first step was overviewed in the previous section. 
In order to carry on the second step, we wish to find some progress measure, that would decrease in each iteration of the first step.
For the case of CNFs the {\sf CNF-width} (i.e., the maximal number of literals in a clause) was a good progress measure for \cite{GopalanMRTV12}. They showed that without loss of generality the CNF-width is $O(\log(n/\eps))$ initially, and that it decreases by a constant-factor in each iteration of Step~1.

Our analogous progress measure is the number of colliding layers.
We recall that in a ROBP, some layers of edges form permutations on the state space, while others are colliding. 

We show that after the first application of step~1, with high probability the restricted ROBP can be written as a composition of $m$ subprograms $D_1, \ldots, D_m$ where each $D_i$ has at most $2$ vertices in the first and last layers and at most $\ell_0 = O(\log(n/\eps))$ colliding layers. 
Intuitively, this happens since every colliding layer reduces the width to $2$ with constant probability and thus with high probability in any $O(\log(n/\eps))$ consecutive colliding layers at least one would be set to the value that reduces the width to $2$.
This motivates the following definition.\looseness=-1
\begin{definition}
We call a ROBP $B$ a $(w, \ell,m)$-ROBP if $B$ can be written as $D_1 \circ \ldots \circ D_m$, with each $D_i$ being a width $w$ ROBP with the first and last layers having at most two vertices and each $D_i$ having at most $\ell$ colliding layers.
\end{definition}

We wish to show that the parameter $\ell$ (that bounds the maximal number of colliding layer in a subprogram $D_i$ with width-$2$ in the first and last layers) reduces by a constant factor under any iteration of step 1.
That is, to show that after iteration $i$ of step~1 we get with high probability a $(3, \ell_i,m_i)$-ROBP where $\ell_i = \ell_0/c^i$ for some constant $c<1$.
As long as $m_i \le \exp(O(\ell_i))$, an inductive argument works since the colliding layers in each individual $D_j$ reduces by a factor $c$ with probability $1-\exp(-\Omega(\ell_i))$ and we can afford a union bound over all $m_i$ subprograms.
However, we cannot afford such a union bound if $m_i \gg \exp(\ell_i)$.
To handle this, we prove the following structural result:
any $(3,\ell_i,m_i)$-ROBP can be well-approximated by $(3,\ell_i, C^{\ell_i})$-ROBPs for some constant $C$.
Furthermore, we show that the {\sf error indicator} of the approximator can be written as the AND of $C^{\ell_i}$ many $(3,\ell_i,1)$-ROBPs, and that its expectation under the uniform distribution is doubly-exponentially small in $\ell_i$. This allows us to show that the error indicator is small under the pseudo-random assignments as well, and we can safely replace a $(3,\ell_i,m_i)$-ROBP with its $(3,\ell_i, C^{\ell_i})$-ROBP approximator.

Applying the restriction and the structure result $O(\log\log n)$ times, we end up with a $(3, \ell', C^{\ell'})$ ROBP where $\ell' = O(\log(1/\eps))$. 
As a last step, we show that $(3, \ell', C^{\ell'})$-ROBPs are fooled by the INW generator \cite{INW94} with seed-length $\tilde{O}(\log(n)\log(1/\eps))$. This follows from the results of \cite{BravermanRRY14}.
For the unordered case, we use the generator from the recent work of \cite{CHHL18} for the last step, with seed-length $\tilde{O}(\log(n) \cdot \poly(1/\eps))$ (using a structural result by \cite{SteinkeVW17}).

\subsection{The proof of Theorem~\ref{thm:read-once poly}}
Theorem~\ref{thm:read-once poly} is a special case of the following theorem
\begin{restatable}{thm}{assignall}\label{thm:GXOR}
Let $n,w,b\in \N$, $\eps>0$. There exists a log-space explicit pseudorandom generator that $\eps$-fools any XOR of ROBPs of  width-$w$ and length-$b$ (defined on disjoint sets of variables), using seed-length $O(\log(b) + \log\log(n/\eps))^{2w+2} \cdot \log(n/\eps)$.
\end{restatable}
We consider $b$ as the progress measure, and wish to show that this parameter reduces under pseudorandom restrictions.
This is analogous to the the number of colliding layers $\ell$ in the previous section.
However, here, in some cases, we cannot guarantee that the application of the pseudorandom restriction from Theorem~\ref{thm:pseudorandom-restriction-XOR-short} would decrease $b$.
The problematic cases are when we have the XOR of more than $\exp(b)$ functions on $b$ variables each.
We show that in such cases, an ``aggressive'' pseudorandom restriction,  assigning $1-\exp(-b)$ fraction of the variables, simplifies the function significantly, while maintaining its acceptance probability. 
Combining applications of mild-restrictions and aggressive-restrictions in a ``decision tree of random restrictions'' results in an assignment that fools the function. 
However, this does not give a PRG as the decisions made along the tree depend adaptively on the function we try to fool, and PRGs cannot depend on the function they try to fool.
We fix this by taking the XOR of several pseudorandom assignments, one per each path in this decision tree in order to construct a PRG that fools this class of functions. (For a longer overview, see Section~\ref{sec:assign_all}.)

\subsection{Organization}
In Section~\ref{sec:prelim}, we state useful definitions and results from previous work to be used throughout the paper. The rest of paper is organized such that each section starts with an overview that highlights one or two main results proved in it. Section~\ref{sec:pseudorestriction} proves the reduction from 3ROBPs to the XOR of short 3ROBPs. 
In Section~\ref{sec:XOR_of_short}, we show how to pseudorandomly assign $1/\poly\log\log(n)$ of the input coordinates while preserving the acceptance probability of the XOR of short ROBPs of constant-width (Theorem~\ref{thm:pseudorandom-restriction-XOR-short}).
In Section~\ref{sec:assign_all}, we construct pseudorandom generators (assigning all the inputs) for the XOR of short ROBPs of constant-width. As an application, we prove Theorem~\ref{thm:read-once poly}.
Then, in Section~\ref{sec:PRG_3ROBPs}, we prove Theorems~\ref{thm:main-3ROBP-ordered} and \ref{thm:main-3ROBP-unordered}.
We remark that a reader interested only in the proof of Theorem~\ref{thm:read-once poly} may skip Sections~\ref{sec:pseudorestriction} and \ref{sec:PRG_3ROBPs}.
Similarly, a reader interested only in the proof of Theorems~\ref{thm:main-3ROBP-ordered} and~\ref{thm:main-3ROBP-unordered} may skip Section~\ref{sec:assign_all}.

\newcommand{\Sel}{\mathrm{Sel}}
\section{Preliminaries}\label{sec:prelim}

Denote by $U_n$ the uniform distribution over $\pmone^n$, and by $U_S$ for $S\subseteq [n]$ the uniform distribution over $\pmone^S$. 
Denote by $\log$ the logarithm in base $2$. For any function $f: \pmone^n \to \R$, we shorthand by $\E[f] = \E_{x\sim U_n}[f(x)]$ and by  $\Var[f] = \E[f^2]-\E[f]^2$.
For an event $E$ we denote by $\one_{E}$ its indicator function.
\subsection{Restrictions}
For a set $T \subseteq [n]$ and two strings $x \in \pmone^T$, $y\in \pmone^{[n]\setminus T}$ we denote 
by $\Sel_T(x,y)$ the string with $$\Sel_T(x,y)_i = \begin{cases}x_i,& i\in T\\ y_i, &\text{otherwise.}\end{cases}$$
\begin{definition}[Restriction]\label{def:restriction}
  Let $f:\pmone^n\to\R$ be a  function. A {\sf restriction} is a pair $(T,y)$ where $T \subseteq [n]$ and $y \in \pmone^{[n]\setminus T}$. We denote by $f_{T|y}:\pmone^n \to \R$ the function $f$ restricted according to $(T,y)$, defined by $f_{T|y}(x) = f(\Sel_T(x,y))$.
\end{definition}

\begin{definition}[Random Valued Restriction]\label{def:random_valued_restriction}
Let $n \in \N$.
A random variable $(T,y)$, distributed over restrictions of $\pmone^n$ is called {\sf random-valued} if conditioned on $T$, the variable $y$ is uniformly distributed over $\pmone^{[n]\setminus T}$.
\end{definition}

\begin{definition}[$p$-Random Restriction]\label{def:p_random_restriction}
  A {\sf $p$-random restriction} is a random-valued restriction over pairs $(T,y)$ sampled in the following way: For every $i\in[n]$, independently, pick $i$ to $T$ with
  probability $p$;  
  Sample $y$ uniformly from $\pmone^{[n]\setminus T}$.
  We denote this distribution of restrictions by $\Rp$.
\end{definition}

\begin{definition}[The Bias-Function]\label{def:bias_function}
Let $f: \pmone^n \to \R$.
Let $T \subseteq [n]$. 
We denote by $\Bias_T(f): \pmone^n \to  \R$ the function defined by $(\Bias_T(f))(x) = \E_{y\sim U_{[n]\setminus T}}[f_{T|y}(x)]$.
When $T$ is clear from the context, we shorthand $\Bias_T(f)$ as $\tilde{f}$.
\end{definition}

\subsection{Fourier analysis of Boolean functions}\label{subsec:Fourier}
Any function $f: \pmone^n \to \R$ has a unique Fourier representation:
\[ f(x) = \sum_{S\subseteq[n]} \hat{f}(S) \cdot \prod_{i\in S} x_i\;,\]
where the coefficients $\hat{f}(S) \in \R$ are given by $\hat{f}(S) = \E_{x\sim U_n} [f(x) \cdot \prod_{i\in S} x_i]$.
We have $\var[f] = \sum_{\emptyset \neq S \subseteq [n]}{\hat{f}(S)^2}$.
We denote the {\sf spectral-norm} of $f$ by $L_1(f) \triangleq \sum_{S \subseteq [n]} |\hat{f}(S)|$.
For any functions $f,g: \pmone^n \to \R$ it holds that $L_1(f\cdot g) \le L_1(f) \cdot L_1(g)$ where equality holds if $f$ and $g$ depends on disjoint sets of variables. Additionally, $L_1(f+g) \le L_1(f) + L_1(g)$.
The following fact relates the Fourier coefficients of a Boolean function and its bias-function.

\begin{fact}[\protect{\cite[Proposition~4.17]{OdonnellBook}}]\label{fact:bias-fnc-Fourier}
	Let $f: \pmone^n \to \R$ and $S, T \subseteq [n]$. Then,
$\widehat{(\Bias_T f)}(S) = \hat{f}(S) \cdot \one_{\{S \subseteq T\}}$
\end{fact}

\subsection{Small-biased distributions}
We say that a distribution $\D$ over $\pmone^n$ is {\sf $\delta$-biased}\footnote{Note that the terms bias-function and small-biased distributions are unrelated.}  if for any non-empty $S \subseteq [n]$ it holds that
$\abs{\E_{x\sim \D}[\prod_{i\in S} x_i]} \le \delta$. \cite{NaorNaor93,AlonGHP92,AlonBNNR92,Ben-AroyaT13,Ta-Shma17} show that $\delta$-biased distributions can be sampled using $O(\log(n/\delta))$ random bits. 

Let $p \in (0,1]$. We say that a distribution $\D_p$ over subsets of $[n]$ is {\sf $\delta$-biased with marginals~$p$} if for any non-empty $S \subseteq [n]$ it holds that
$\Pr_{T\sim \D_p} [S \subseteq T] = p^{|S|} \pm \delta.
$

\begin{claim}\label{claim:sampling-t}
Let $p = 2^{-a}$ for some integer $a>0$, 
let $\D$ be an $\eps$-biased distribution over $\pmone^{na}$.
Define $\D_p$ to be a distribution over subsets of $[n]$ as follows:
Sample $x\sim \D$. 
Output 
$T = 
\{i\in [n]: \bigwedge_{j\in [a]} (x_{(i-1)a + j} = 1)\}$.
Then $\D_p$ is $\eps$-biased with marginals $p$.
\end{claim}
\begin{proof}
For any fixed $S$, 
	the probability that $S \subseteq T$ is exactly the probability that $\bigwedge_{i\in S, j\in [a]} (x_{(i-1)a + j} = 1)$.
	In an $\eps$-biased distribution, the latter event happens with probability $2^{-a\cdot |S|} \pm \eps$ (See \cite{AlonGHP92}).
\end{proof}

\begin{claim}\label{claim:inclusion-exclusion}
If $\D_p$ is $\delta$-biased with marginals $p$, then for any disjoint $S,S'\subseteq [n]$ it holds that 
$\Pr_{T\sim \D_p}[S \cap T = \emptyset, S' \subseteq T] = (1-p)^{|S|}\cdot p^{|S'|} \pm \delta	\cdot 2^{|S|}$.
\end{claim}
\begin{proof}
	By inclusion-exclusion 
	\begin{align*}\Pr_{T \sim \D_p}[S \cap T = \emptyset, S' \subseteq T] &= \sum_{R \subseteq S} (-1)^{|R|} \cdot \Pr_{T \sim \D_p}[R \cup S' \subseteq T] \\
	&= \sum_{R \subseteq S} (-1)^{|R|} \cdot (\Pr_{T \sim U}[R \cup S' \subseteq T] \pm \delta)\\
	&= \Pr_{T \sim U}[S \cap T = \emptyset, S' \subseteq T] \pm 2^{|S|}\cdot \delta.\qedhere\end{align*}
\end{proof}

\subsection{Standard tail bounds for $k$-wise independence}

\begin{lemma}[\protect{\cite[Thm.~4, restated]{SchmidtSS95}}]
\label{lemma:tail_bounds}
Let $\ell$ be an even positive integer.
	Let $X_1, \ldots, X_m$ be some $\ell$-wise independent random variables bounded in $[-1,1]$ with expectation $0$.
	Let $V = \sum_{i=1}^{m}\var[X_i]$.
	Then,
	$\E[(X_1 + \ldots +X_m)^{\ell}] \le \max\{\ell^\ell, (\ell V)^{\ell/2}\}.$
\end{lemma}

\subsection{Branching programs}\label{prelim:ROBP}
A {\sf read-once branching program (ROBP)} $B$ of {\sf length} $n$ and {\sf width} $w$ is a directed layered graph with $n+1$ layers of vertices denoted $V_1, \ldots, V_{n+1}$.  Each $V_i$ consists of $w_i \le w$ vertices $\{v_{i,1}, \ldots, v_{i,w_i}\}$, and between every two consecutive layers $V_i$ and $V_{i+1}$ there exists a set of directed edges (from $V_i$ to $V_{i+1}$), denoted $E_i$, such that any vertex in $V_i$ has precisely two out-going edges in $E_i$, one marked by $1$ and one marked by $-1$.
The vertices in $V_{n+1}$ are marked with either `accept' and `reject'.

A branching program $B$ and an input $x\in \pmone^n$ naturally describes a {\sf computation~path} in the layered graph:
we start at node $v_1 = v_{1,1}$ in $V_1$.
For $i=1, \ldots, n$, we traverse the edge going out from $v_i$ marked by $x_i$ to get to a node $v_{i+1} \in V_{i+1}$.
The resulting computation path is $v_1 \to v_2 \to \ldots \to v_{n+1}$.
We say that $B$ accepts $x$ iff the computation path defined by $B$ and $x$ reaches an accepting node. Naturally $B$ describes a Boolean function $B:\pmone^n \to \pmone$ whose value is $-1$ on input $x$ iff $B$ accepts $x$.

{\sf Unordered branching programs} are defined similarly, expect that there exists a permutation $\pi \in S_n$ such that in step $i$ the computation path follows the edge marked by $x_{\pi_i}$, for $i\in [n]$.
We also consider unordered branching programs on $[n]$ of shorter length $n'\le n$. In such case, the program stops after reading $n'$ input bits.%
\footnote{Note that in the unordered case, the set of bits being read could be an arbitrary subset of $[n]$ of size $n'$.}

For two programs $B_1$ and $B_2$ defined over disjoint sets of variables and having the end width of $B_1$ equal the start width of $B_2$, we denote by $B_1 \circ B_2$ the concatenation of $B_1$ and $B_2$, defined in the natural way.

\paragraph{Locally Monotone Branching Programs.}
Let $B$ be a width-$w$ length-$n$ ROBP.
For any vertex $v$ in the ROBP, denote by $\beta_v$ the probability to accept a uniformly random input starting from the vertex $v$.
Since renaming the vertices in each layer does not affect the functionality of $B$, we may assume without loss of generality that the vertices in $V_i$ are ordered according to $\beta_v$. That is, for every $i\in[n+1]$  we have $\beta_{v_{i,1}} \le \beta_{v_{i,2}} \le \ldots \le \beta_{v_{i,w_i}}$. In case of equalities, we break ties arbitrarily but commit to a strict ordering of the nodes in each layer. $B$ is called {\sf locally monotone} if for any vertex $v$ in $B$ the vertex reached from $v$ using the $1$-edge has larger or equal index than the vertex reached from $v$ using the $(-1)$-edge.

For $i\in [n]$, denote by $E_{i,1}$ the set of edges in $E_i$ marked by $1$ and similarly define $E_{i,-1}$. We say that $E_i$ is a {\sf identity layer} if  $E_{i,1} = E_{i,-1}$ (in which case $x_i$ does not affect the output of of $B$). We say that $E_i$ is a {\sf permutation layer} if both $E_{i,1}$ and $E_{i,-1}$ form a matching between $V_i$ and $V_{i+1}$ (i.e., $|V_i| = |V_{i+1}|$ and for $b\in \{-1,1\}$ no two edges in $E_{i,b}$ enter the same vertex in $V_{i+1}$). The following is a key lemma from the work of \cite{BrodyV10}.
\begin{lemma}[Collision Lemma~\cite{BrodyV10}]\label{lemma:identity or collision}
	In a locally monotone branching program, every permutation layer is an identity layer.
\end{lemma}
To see it, note that if we think of the vertices in each layer $\{v_{i,1}, \ldots, v_{i,w_i}\}$ as written from top to bottom according to $\beta_v$, then in a locally monotone program for any vertex $v$ the $1$-edge leads to the same vertex or to a vertex below the one that follows the $(-1)$-edge. Thus, assuming  both $E_{i,-1}$ and $E_{i,1}$ form a matching, the only way this could happen is if they both form the same matching.

The following is a restatement of a result from \cite{CHRT17}. We give its proof for completeness in Appendix~\ref{app:CHRT}.
\begin{theorem}\label{thm:CHRTa}
Let $B$ be an unordered oblivious read-once branching programs with width-$w$ and length-$n$. Let $\eps>0$, $p \le 1/O(\log n)^w$, $k = O(\log(n/\eps))$, and $\D_p$ be a $\delta_T$-biased distribution over subsets of $[n]$ with marginals $p$, for some $\delta_T \le p^{2k}$.
Then, with probability at least $1-\eps$ over $T\sim \D_p$, 
\[
L_1(\tilde{B}) =  \sum_{S \subseteq T}{ |\hat{B}(S)|} \le O((nw)^3/\eps).\]
\end{theorem}

\begin{theorem}[Implied by \protect{\cite[Thm.~2]{CHRT17}} and \protect{\cite[Thm.~4.1]{SteinkeVW17}}]\label{thm:CHRT}
Let $\mathcal{C}$ be the class of all unordered oblivious read-once branching programs on $[n]$ of length at most $n'$ and width at most $w$.
Then,
there exists a log-space explicit pseudorandom generator 
\[
\mathbf{CHRT}:\pmone^{s_{n,n',w,\eps}}\to \pmone^n
\]
that $\eps$-fools $\mathcal{C}$, 
where $s_{n,n',w, \eps} = O(\log(n')^{w+1} \log\log(n') \log(n/\eps))$.
\end{theorem}

\subsection{Helpful lemmas}

\begin{lemma}\label{lemma:var}
Let $a,b>0$.
If $X$ is a real-valued random variable bounded in $[-b,a]$ with mean $0$, then $\var[X] \le ab$.
\end{lemma}
	\begin{proof}
		$\var[X] = \E[X^2]$ since $\E[X]=0$. As $x^2$ is convex and $X$ domain is bounded, the maximal value that $\E[X^2]$ can get is if all of $X$'s probability mass is on the boundary. 
		Denote by $p = \Pr[X=a]$.
		Since $\E[X]=0$ we get $0 = p\cdot a + (1-p) \cdot (-b)$, i.e., $p = b/(a+b)$, thus 
		\[ \var[X] = \E[X^2] \le a^2 p + (1-p)b^2 = \frac{a^2 b}{a+b} + \frac{ab^2}{a+b} = ab\;.\qedhere\]
	\end{proof}

\begin{theorem}[Hyper-contractivity of Variance]\label{thm:HC}
Let $f: \pmone^k \to \pmone$ be a Boolean function.
Then, 
$\E_{T\sim \Rp}[\var[\tilde{f}]] \le p \cdot \var[f]$. Furthermore, if $p \le 1/3$, then
$\E_{T\sim \Rp}[\var[\tilde{f}]] \le \var[f]^{3/2}$.
\end{theorem}

\begin{proof}
First, observe that using Fact~\ref{fact:bias-fnc-Fourier} and $\var[g] = \sum_{S\neq \emptyset} \hat{g}(S)^2$,  we have
$$\E_{T\sim \Rp}[\var[\tilde{f}]] = 
\E_{T\sim \Rp}\bigg[\sum_{S\neq \emptyset} \widehat{(\Bias_T f)}(S)^2\bigg] = 
\sum_{S\neq \emptyset} \hat{f}(S)^2 \cdot \Pr_{T\sim \Rp}[S \subseteq T] = 
\sum_{S\neq \emptyset} {p^{|S|} \cdot \hat{f}(S)^2 }.$$
For the first item, we get $\E_{T\sim \Rp}[\var[\tilde{f}]] = \sum_{S\neq \emptyset} {p^{|S|} \cdot \hat{f}(S)^2 } \le p \cdot \sum_{S\neq \emptyset} {\hat{f}(S)^2} = p\cdot \var[f]$.

For the second item, we use the Hyper-contractivity Theorem \cite{Bonami70} (cf. \cite[Ch.~9]{OdonnellBook}) stating that $\|N_{\rho} g\|_2 \le \|g\|_{1+\rho^2}$ for any function $g:\pmone^n \to \R$ (where $N_{\rho}$ is the noise operator that satisfies $\hat{N_{\rho} g}(S) = \rho^{|S|} \cdot \hat{g}(S)$ for all $S\subseteq [n]$).
Take $g = f- \E[f]$ and $\rho = \sqrt{p}$.
Then, \looseness=-1
$$
\E_{T\sim \Rp}[\var[\tilde{f}]] = \sum_{S\neq \emptyset} {p^{|S|} \cdot \hat{f}(S)^2 } =  
\|N_{\sqrt{p}} g\|^{2}_2 \le \|g\|_{1+p}^{2} = \E_{x\sim U_k}[|g(x)|^{1+p}]^{2/(1+p)}
$$
We analyze the RHS.
Let $\beta = \E[f]$. 
Then, $\beta\in [-1,1]$, $\var[f] = 1-\beta^2$, and under the uniform distribution $|g(x)|$ gets value $|1-\beta| = 1-\beta$ with probability $(1+\beta)/2$ and value $|-1-\beta| = 1+\beta$ with probability $(1-\beta)/2$.
We get 
\begin{align*}
	\E_{x\sim U_k}[|g(x)|^{1+p}]
	&= \frac{1+\beta}{2} \cdot (1-\beta)^{1+p} + \frac{1-\beta}{2} \cdot (1+\beta)^{1+p} \\
	&= (1-\beta^2) \cdot (\tfrac{1}{2} (1-\beta)^{p} + \tfrac{1}{2}(1+\beta)^{p}) 
	\le 1-\beta^2 = \var[f]
\end{align*}
where the last inequality follows by concavity of $x \mapsto x^{p}$.
Overall if $p\le 1/3$, then 
$\E_{T\sim \Rp}[\var[\tilde{f}]]  \le   \var[f]^{2/(1+p)} \le  \var[f]^{3/2}$.
\end{proof}

\begin{lemma}\label{lemma:vars_prod}
	Suppose $\D_p$ is $\delta_T$-biased distribution with marginals $p$.
Let $\ell \in \N$.
Let $f_1, \ldots, f_\ell: \pmone^n \to \R$ be real valued functions, not necessarily distinct.
Then, $$\abs{\E_{T \sim \D_p}\left[ \prod_{i=1}^{\ell} \var[\tilde{f_i}] \right] - 
\E_{T \sim \Rp}\left[ \prod_{i=1}^{\ell} \var[\tilde{f_i}]\right]} \;\le  \;\delta_T \cdot\prod_{i=1}^{\ell} \var[f_i].$$
\end{lemma}
\begin{proof}
Using Fact~\ref{fact:bias-fnc-Fourier}, for any fixed $T$, we have
\begin{align*}
\prod_{i=1}^{\ell} \var[\tilde{f_i}] &= 
\prod_{i=1}^{\ell} \sum_{S_i \neq \emptyset} \hat{f_i}(S_i)^2 \cdot \one_{\{S_i \subseteq T\}}	= \sum_{S_1, \ldots, S_{\ell} \neq \emptyset} \hat{f_1}(S_1)^2 \cdots \hat{f_\ell}(S_\ell)^2 \cdot \one_{\{S_1 \cup \ldots \cup S_{\ell} \subseteq T\}}.
\end{align*}
Thus,
\begin{align*}
\E_{T \sim \D_p}\left[ \prod_{i=1}^{\ell} \var[\tilde{f_i}]\right] = \sum_{S_1, \ldots, S_{\ell} \neq \emptyset} \hat{f_1}(S_1)^2 \cdots \hat{f_\ell}(S_\ell)^2 \cdot (p^{|S_1 \cup \ldots \cup S_{\ell}|} \pm \delta_T)
\end{align*}
and 
\begin{align*}
\E_{T \sim \Rp}\left[ \prod_{i=1}^{\ell} \var[\tilde{f_i}]\right] = \sum_{S_1, \ldots, S_{\ell} \neq \emptyset} \hat{f_1}(S_1)^2 \cdots \hat{f_\ell}(S_\ell)^2 \cdot p^{|S_1 \cup \ldots \cup S_{\ell}|}.
\end{align*}
The difference between the two is at most 
$$\left|\E_{T \sim \D_p}\left[ \prod_{i=1}^{\ell} \var[\tilde{f_i}]\right] -  \E_{T \sim \Rp}\left[ \prod_{i=1}^{\ell} \var[\tilde{f_i}]\right]\right| \le \delta_T \cdot \sum_{S_1, \ldots, S_{\ell} \neq \emptyset} \hat{f_1}(S_1)^2 \cdots \hat{f_\ell}(S_\ell)^2 = \delta_T \cdot \prod_{i=1}^{\ell} \var[f_i],$$
which completes the proof.
\end{proof}

\section{From width-3 ROBPs to the XOR of short ROBPs}\label{sec:pseudorestriction}
In Section~\ref{sec:XOR_of_short}, we prove the following theorem.
\mainintro*
The pseudorandom restriction assigns $p$ fraction of the variables as follows: 
\begin{enumerate}
	\item Choose a set of coordinates $T \subseteq [n]$ according to a $\delta_T$-biased distribution with marginals $p$, for $\delta_T := p^{ O(\log(n/\eps))}$.
	\item Assign the variables in $T$ according to a $\delta_x$-biased distribution, for $\delta_x := (\eps/n)^{O(\log b)}$.
\end{enumerate}
Known constructions of small-biased distributions \cite{NaorNaor93,AlonGHP92,AlonBNNR92,Ben-AroyaT13,Ta-Shma17} show that it suffices to use $O(\log(n/\delta_T) + \log(n/\delta_x)) \le O(w \cdot \log (n/\eps) \cdot (\log\log(n/\eps) + \log(b)))$ random bits to sample the restriction.

In this section, we show how to design pseudorandom restrictions for unordered width-3 ROBPs from pseudorandom restrictions to the XOR of many width-3 ROBPs of length $O(\log (n/\eps))$.
We get the following theorem.
\twosteps*

\paragraph{Proof Sketch.}
In this section, we shall show that under pseudorandom restrictions leaving each variable alive with probability $1/2$, with high probability, the bias function of a ROBP $B$ can be written as a linear combination (up to a small error) over functions of the form $f_1 \cdot f_2 \cdot \ldots \cdot f_m$ where each $f_i$ is a  short subprogram of the original program of length $O(\log (n/\eps))$, and each $f_i$ is defined on a disjoint set of coordinates. Each function $g$ in the linear combination will have a weight $\alpha_g \in [-1,1]$, and the sum of absolute values of weights over all functions participating in the linear combination will be at most $n$. This will show that any generator that $\eps/n$-fools the XOR of short width-$3$ ROBPs also $\eps$-fools width-$3$ length-$n$ ROBPs under random restrictions.

The reduction will first establish that with high probability (over the choice of the set of coordinates that are left alive) the bias function of a ROBP $B$ can be written as the average of width-$3$ length-$n$ ROBPs, whose vast majority have at most $O(\log (n/\eps))$ layers between every two layers with width-$2$. 
Then, we use a result of Bogdanov, Dvir, Viola, Yehudayoff~\cite{BogdanovDVY13} that  reduces branching programs with many width-$2$ layers to the XOR of short ROBPs.\looseness=-1

We focus on the first part of the reduction. 
First, consider the case when $B$ is locally-monotone. 
In this case, every layer of edges is either the identity layer or a colliding layer (Lemma~\ref{lemma:identity or collision}). Assume without loss of generality that there are no identity layers.
Then, under a pseudorandom restriction, with high probability, in every $O(\log (n/\eps))$ consecutive layers we will have a layer of edges whose corresponding variable is fixed to the value on which the edges in the layer collide, leaving at most $2$ vertices reachable in the next layer of vertices. Removing unreachable vertices, we get that with high probability under the random restriction, in every $O(\log (n/\eps))$ consecutive layers there is a layer of vertices with width-$2$.

However, in the case that $B$ is not locally-monotone (e.g.,  when $B$ is a permutation ROBP) it could the case that the widths of all layers of vertices remain $3$ under the random restriction.
Our main observation is that since the bias function takes the average over all assignments to the restricted variables, the bias function of $B$ does not depend on the labels of edges marked by the restricted variables. 
More formally, for any $T \subseteq [n]$, if $B$ and $C$ are two ROBPs with the same graph structure that only differ on the labels on the edges in layers $[n]\setminus T$, then $\Bias_T(B) = \Bias_T(C)$.
Thus, once $T$ is fixed we may relabel the layers in $[n]\setminus T$ so that they are locally-monotone, yielding a new ROBP $B'$, and then apply the bias function.
Using the analysis of the locally monotone case, we get that the bias function of $B'$ (and thus the bias function of $B$) is the average of $B'$ over all restrictions fixing the coordinates in $[n]\setminus T$, and we know that most of these restricted ROBPs have width-$2$ in every $O(\log (n/\eps))$ consecutive layers.

Essentially, the bias function allows us to imagine as if we are taking the average over restrictions of $B'$ rather than restrictions of $B$, and restrictions of $B'$ are ``simpler'' to fool than restrictions of $B$ since they have many layers with width-$2$.

The formal argument follows.

\begin{theorem}[From width-$3$ to almost width-$2$]\label{thm:the-bias-trick}
Let $B$ be a ROBP of width-$3$ and length-$n$. 
Let $\eps>0$. 
Let $\D_{1/2}$ be a $(\eps/n)^{10}$-biased distribution over subsets of $[n]$ with marginals $1/2$.
Let $T\sim \D_{1/2}$ be a random variable.
Let $B^T$ be the branching program $B$ where the layers in $[n] \setminus T$ are relabeled so that they are locally monotone.
Then, 
$$
\Bias_T(B)(x) = \Bias_T(B^T)(x) = \E_{y\sim U_{[n]\setminus T}}[(B^{T}_{T|y}(x)]
$$
and with  probability at least $1-\eps$ over the choice of $T$ and $y$, $B^{T}_{T|y}$ can be computed by a ROBP of the form $D_1 \circ \ldots \circ D_m$ where $\{D_i\}_{i=1}^{m}$ are defined over disjoint sets of at most $b = (3\log(n/\eps))$ variables, and each $D_i$ is a width-$3$ ROBP with at most $2$ vertices on the first and last layers.
\end{theorem}

\begin{proof}
We first observe that $\Bias_T(B)(x) = \Bias_T(B^T)(x)$.
Indeed, for any fixed $x$, $\Bias_T(B)(x)$ equals the probability that the following random-path in $B$ accepts:
\begin{quote}
Initiate $v_1$ to be the start node of $B$. For $i=1, \ldots, n$ if $i \in T$, take the edge exiting $v_i$ marked by $x_i$, otherwise (i.e., if $i\in [n]\setminus T$) pick a random edge out of the two edges exiting $v_{i}$. Denote by $v_{i+1}$ the node at the end of the edge taken in the $i$-th step. Accept if and only if $v_{n+1}$ is an accepting node.\end{quote}
 Observe that the following random process is oblivious to the labels of edges in layers $[n]\setminus T$, thus it would yield the same probability for $\Bias_T(B)$ and for $\Bias_T(B^{T})$. Overall, we got that $\Bias_T(B)$ and $\Bias_T(B^T)$ are equal as functions.

In the remainder of the proof, we analyze $\Bias_T(B^T)$.
Let $E_{i,1}$ and $E_{i,-1}$ denote the set of edges in the $i$-layer of $B$ marked by $1$ and $-1$ respectively.
We assume without loss of generality that in all layers of edges  $E_{i,1} \neq E_{i,-1}$, as otherwise the $i$-th layer is redundant and may be eliminated.
(Observe that under any relabeling of $B$ this property is preserved.)
By the collision lemma of Brody-Verbin~\cite{BrodyV10} (Lemma~\ref{lemma:identity or collision}), for any $i\in [n] \setminus T$, layer $i$ in $B^T$ has the following property: either $E_{i,1}$ or $E_{i,-1}$ has at most $2$ end-vertices.

Next, we consider the program $B^{T}_{T|y}$ for a pseudorandom $T$ and a random $y\in \pmone^{[n]\setminus T}$.
For $i=1, \ldots, n$ we say that the $i$-th layer of edges is ``good'' under the choice of $T$ and $y$, if $i \in [n] \setminus T$ and layer $E_{i,y_i}$ of $B^T$ has at most $2$ end-vertices.
Let $b = 3 \log(n/\eps)$.
For $i=1, \ldots, n-b+1$ let $\cE_i$ be the event that none of layers $\{i,i+1, \ldots, i+(b-1)\}$ is good.
Since $T$ is sampled from a $(\eps/n)^{10}$-biased distribution with marginals $1/2$, we have that $T$ is $(\eps/2n)$-almost $b$-wise independent.
Thus, up to an error of $\eps/2n$ we may analyze the event $\cE_i$ under uniform choice of a subset $T\subseteq[n]$.
Indeed, under a uniform choice for $T$ and $y$ each layer is good with probability at least $1/4$, and all $b$ layers are not good with probability at most $(3/4)^{b}$.
Overall, we get $\Pr[\cE_i] \le (3/4)^{b} + (\eps/2n) \le \eps/n$.
By the union bound, $$\Pr[\cE_1 \vee \cE_2 \vee \ldots \vee \cE_{n-b+1}] \le (n-b+1) \cdot (\eps/n) \le \eps.$$
Under the event that all $\cE_i$ are false, we get that $B^T_{T|y}$ has width $2$ in every $b$ layers. In such a case, we may write the restricted function $B^T_{T|y}$ as $D_1 \circ \ldots \circ D_m$ where each $D_i$ is a width-$3$ and length at most $b$ ROBP with at most $2$ vertices on the first and last layer.
\end{proof}

\begin{theorem}[from almost width-2 to the XOR of short ROBPs  - restatement of \protect{\cite[Thm.~2.1]{BogdanovDVY13}}]
	\label{thm:BDVY}
	Let $B$ be a ROBP of the form $D_1\circ \ldots \circ D_m$  where $\{D_i\}_{i=1}^{m}$ are defined over disjoint sets of 
	variables, and each $D_i$ is a width-$3$ ROBP with at most $2$ vertices on the first and last layers.
	Then, (as a real-valued function) $B$ can be written as a linear combination of 
	$\sum_{\alpha \in \B^m} c_{\alpha} \cdot \prod_{i=1}^{n} D_{i,\alpha_i}$
	where $D_{i,0}, D_{i,1}$ are subprograms of $D_{i}$ and $\sum_{\alpha\in \B^n}{|c_\alpha|} \le m$.
	\end{theorem}

\begin{proof}[Proof of Theorem~\ref{thm:main_two_steps}]
We prove that the following pseudorandom restriction maintains the acceptance probability of ROBPs of width-$3$ and length-$n$ up to error $\eps$.
Let $\eps_1 := \eps/2$, $\eps_2 := \eps/2n$.
\begin{enumerate}
	\item Pick $T_0 \subseteq [n]$ using a $(\eps_1/n)^{10}$-biased distribution with marginals $1/2$.
	\item
	\begin{enumerate}
	\item Pick $T\subseteq T_0$ using a $\delta_T$-biased distribution with marginals $p = 1/O(\log \log ( n/\eps_2))^{6}$.
	\item Assign the coordinates in $T$ using a $(\eps_2/n)^{O(\log \log (n/\eps_2))}$-biased distribution $\D_x$.
	\end{enumerate}
\end{enumerate}
Equivalently, we prove that the following distribution $\eps$-fools ROBPs of width-$3$ and length-$n$.
\begin{enumerate}
	\item Pick $T_0 \subseteq [n]$ using a $(\eps_1/n)^{10}$-biased distribution with marginals $1/2$.
	\item Assign the coordinates in $[n]\setminus T_0$ uniformly at random.
	\item
	\begin{enumerate}
	\item Pick $T\subseteq T_0$ using a $\delta_T$-biased distribution with marginals $p = 1/O(\log \log ( n/\eps_2))^{6}$.
	\item  Assign the coordinates in $T_0\setminus T$ uniformly at random.
	\item Assign the coordinates in $T$ using a $(\eps_2/n)^{O(\log \log (n/\eps_2))}$-biased distribution $\Dx$.
	\end{enumerate}
\end{enumerate}

Let $y\sim U_{[n]\setminus T_0}$. Let ${\cal G}$ be the event that $B^{T_0}_{T_0|y}$  can be computed by a ROBP of the form $D_1\circ \ldots \circ D_m$  where $\{D_i\}_{i=1}^{m}$ are defined over disjoint sets of at most $b=3\log(n/\eps_1)$ variables, and each $D_i$ is a width-$3$ ROBP with at most $2$ vertices on the first and last layers.
By Theorem~\ref{thm:the-bias-trick} $\Pr({\cal G}) \ge 1-\eps_1$. 
Assuming that ${\cal G}$ happened, then by Theorem~\ref{thm:BDVY},
$B^{T_0}_{T_0|y}$ can be written as 
$\sum_{\alpha \in \B^m} c_{\alpha} \cdot \prod_{i=1}^{n}D_{i,\alpha_i}$
	where $D_{i,\alpha_i}$ are subprograms of $D_{i}$ and $\sum_{\alpha\in \B^n}{|c_\alpha|} \le m$.
For each $\alpha\in \B^m$, using Theorem~\ref{thm:pseudorandom-restriction-XOR-short} we have that 
$$
\left| 
\E_{z\sim U_{T_0}}\left[\prod_{i=1}^{m}D_{i,\alpha_i}(z)\right]
- \E_{T}\E_{x\sim \Dx}\E_{y' \sim U_{T_0 \setminus T}} \left[\prod_{i=1}^{m}D_{i,\alpha_i}(\Sel_{T}(x, y'))\right]
 \right| 
 \le \eps_2\;.
$$
By linearity of expectation and the triangle inequality
$$
\bigg|\E_{z\sim U_{T_0}}\bigg[\sum_{\alpha} c_{\alpha} \cdot \prod_{i=1}^{m}D_{i,\alpha_i}(z)\bigg] - \E_{T}\E_{x\sim \Dx}\E_{y' \sim U_{T_0 \setminus T}}
 \bigg[\sum_{\alpha} c_{\alpha} \cdot \prod_{i=1}^{m} D_{i,\alpha_i}(\Sel_{T}(x, y'))\bigg]\bigg|$$
$$\le \sum_{\alpha}{|c_{\alpha}|} \cdot  \eps_2 \;\le\; m \cdot \eps_2 \;\le\; \eps/2
$$
Overall, we get
$$\bigg|\E_{z\sim U_n}[B(z)]-
\E_{\substack{T_0,\\y\in U_{\bar{T_0}}}}\;
\E_{\substack{T, x\sim \Dx\\y' \sim U_{T_0 \setminus T}}}\;
[B(\Sel_{T_0}(\Sel_T(x,y'),y)]\bigg| =
$$
$$ 
\bigg|\E_{z\sim U_n}[B(z)]-
\E_{\substack{T_0,\\y\in U_{[n]\setminus T_0}}}\;
\E_{\substack{T, x\sim \Dx\\y' \sim U_{T_0 \setminus T}}}\;\
[B^{T_0}(\Sel_{T_0}(\Sel_T(x,y'),y)]\bigg| = 
$$
\begin{equation}\label{eq:aa}
	\bigg|\E_{\substack{T_0,\\y\in U_{[n]\setminus T_0}}}\;
\E_{\substack{T, z\sim U_{T} \\y' \sim U_{T_0 \setminus T}}}\;\
[B^{T_0}(\Sel_{T_0}(\Sel_T(z,y'),y)]-
\E_{\substack{T_0,\\y\in U_{[n]\setminus T_0}}}\;
\E_{\substack{T, x\sim \Dx\\y' \sim U_{T_0 \setminus T}}}\;\
[B^{T_0}(\Sel_{T_0}(\Sel_T(x,y'),y)]\bigg|
\end{equation}
where the last equality is due to the fact for any $T,T_0$ the distribution of $\Sel_{T_0}(\Sel_T(z,y'),y)$ is the uniform distribution over $\pmone^{n}$.
We bound Expression~\eqref{eq:aa} by 
\begin{align*}
	&\E_{T_0,y\in U_{[n]\setminus T_0}}\left[\left|\E_{T} \E_{y' \sim U_{T_0 \setminus T}} \left(\E_{z\sim U_{T}} [B^{T_0}_{T_0|y}(\Sel_T(z,y'))]- \E_{x\sim \Dx} [B^{T_0}_{T_0|y}(\Sel_T(x,y'))]\right)\right|\right]
	\\&\le \Pr\left[\neg {\cal G}\right]  + \E_{T_0,y\in U_{[n]\setminus T_0}}\left[\Big|\E_{T,y' \sim U_{T_0 \setminus T}} \Big(\E_{z\sim U_{T}} [B^{T_0}_{T_0|y}(\Sel_T(z,y'))]- \E_{x\sim \Dx} [B^{T_0}_{T_0|y}(\Sel_T(x,y'))]\Big)\Big| \; \bigg| \; {\cal G} \right]
	\\& \le \eps/2 + \eps/2
\end{align*}
where the second summand is bounded by $\eps/2$ according to the above discussion using Theorem~\ref{thm:BDVY} and Theorem~\ref{thm:pseudorandom-restriction-XOR-short}.
	\end{proof}

\newcommand{\Tvar}{\mathbf{{Tvar}}}

\section{Pseudorandom restrictions for the XOR of short ROBPs}\label{sec:XOR_of_short}
In this section, we prove Theorem~\ref{thm:pseudorandom-restriction-XOR-short}. 
Let $B_1, \ldots, B_m$ be pairwise disjoint subsets of $[n]$, each of size at most $b$. For $i=1, \ldots, m$ let $f_i : \pmone^{B_i} \to \pmone$ be a width $w$ ROBP. We construct a pseudorandom generator that $\eps$-fools $f = \prod_{i=1}^{m}{f_i}$.
We recall the statement of Theorem~\ref{thm:pseudorandom-restriction-XOR-short} and the construction.
\mainintro*

Recall that the pseudorandom restriction assigns $p$ fraction of the variables as follows: 
\begin{enumerate}
	\item Choose a set of coordinates $T \subseteq [n]$ according to a $\delta_T$-biased distribution with marginals $p$, for $\delta_T := p^{ O(\log(n/\eps))}$.
	\item Assign the variables in $T$ according to a $\delta_x$-biased distribution, for $\delta_x := (\eps/n)^{O(\log b)}$.
\end{enumerate}

\paragraph{Analysis.} We shall assume without loss of generality that for all $i=1, \ldots, m$  it holds that $\E[f_i] \ge 0$.
We shall also assume without loss of generality that for all $i=1, \ldots, m$ it holds that $\var[f_i] > 0$ (i.e., that the functions are non-constant). 
Since the functions $f_i$ are Boolean and depend on at most $b$ bits, 
we have $\var[f_i] =\Pr[f_i=1]\cdot \Pr[f_i=-1] \ge  2^{-b}\cdot(1-2^{-b}) \ge 2^{-1-b}$.

We partition the functions into $O(\log b)$ buckets according to their variance.
Let $\sigma_0 = 1$, 
for every $j \in \{1,\ldots, \log_{1.1}(b+1)\}$, let $\sigma_j = 2^{-1.1^j}$ and
$I_j = \{i\in [m]: \var[f_i] \in (\sigma_{j},\sigma_{j-1}]\}$.
Let $C>0$ be a sufficiently large constant.
We consider two cases in our analysis: \begin{description}
\item[Low-Variance Case:] 
For every $j \in \{1,\ldots, \log_{1.1}(b+1)\}$ we have
$$\sum_{i\in I_j} \var[f_i] \le C \cdot \log^2(n/\eps)/(\sigma_{j-1})^{0.1}\;.$$
\item[High-Variance Case:] 
There exists a $j \in \{1,\ldots, \log_{1.1}(b+1)\}$ with $$\sum_{i\in I_j} \var[f_i] > C\cdot \log^2(n/\eps)/(\sigma_{j-1})^{0.1}\;.$$
\end{description}

\paragraph{Setting Up Parameters:}
Let $C'>1$ be a sufficiently large constant.
Set 
\begin{align}
\label{eq:delta_T}
	&\delta_{T} \triangleq p^{2C' \cdot \log(n/\eps)},\\
\label{eq:delta}
&\delta \triangleq  (\eps/n)^{10 C'}, \\
\label{eq:delta'_x}
	&\delta'_x \triangleq (\eps/n)^{100 C'},\\
\label{eq:delta_x}
	&\delta_x \triangleq  (\delta'_x)^{\log_{1.1}(b+1)}.
\end{align}

\subsection{Low-Variance Case}
For $j = 1, \ldots, \log_{1.1}(b+1)$, 
let $F_j(x) = \prod_{i \in I_j} f_j(x)$. 
Thus, $f = \prod_{j} F_j$.
Let $\D_p$ be any $\delta_T$-biased distribution with marginals $p$. 
For $j\in \{1,\ldots,\log_{1.1}(b+1)\}$, we shall show that with probability at least $1-\eps/2n$ over the choice of $T\sim \D_p$, it holds that
\begin{equation}\label{eq:n/40}
\abs{\E_{x\sim \Dtagx}[\tilde{F_j}(x)] - 
\E_{z\sim U_T}[\tilde{F_j}(z)]}  \le \eps/n^{40}\;,\footnote{recall that we denote by $\tilde{g} = \Bias_T(g)$ for any function $g$.}
\end{equation}
for any $\delta'_x$-biased distribution $\Dtagx$ over $\pmone^{n}$.
Thus, by union bound Eq.~\eqref{eq:n/40} holds
for all $j \in \{1, \ldots, \log_{1.1}(b+1)\}$ simultaneously with probability at least $1-\eps/2$ over $T\sim\D_p$.
Using the following XOR lemma for small-biased distributions from \cite{GopalanMRTV12} we get that any 
 $(\delta'_x)^{\log_{1.1}(b+1)}$-biased distribution, fools $\tilde{f}(x) = \prod_{j=1}^{\log_{1.1}(b+1)}{\tilde{F_j}(x)}$ with error at most $16^{\log_{1.1}(b+1)}\cdot 2(\eps/n^{40})\le \eps/2$ (using $b\le n$).
\begin{lemma}[\protect{\cite[Thm.~4.1]{GopalanMRTV12}}, restated]\label{lemma:3.2}
	Let $0 < \eps< \delta\le 1$.
 	Let $F_1, \ldots, F_k : \pmone^n \to [-1,1]$ be functions on disjoint input variables such that each $F_i$ is $\delta$-fooled by any $\eps$-biased distribution.
 	Let $H:[-1,1]^k \to [-1,1]$ be a multilinear function in its inputs.
 	Then $H(F_1(x), \ldots, F_k(x))$ is $(16^k \cdot 2\delta)$-fooled by any $\eps^k$-biased distribution.
\end{lemma}
In Appendix~\ref{app:GMRTV}, we show how to derive Lemma~\ref{lemma:3.2} from \cite[Thm.~4.1]{GopalanMRTV12}.

\medskip
\noindent
In the remainder of this section, we focus on fooling a single $F_j$, that is, fooling the product (i.e., XOR) of functions $\{f_i\}_{i\in I_j}$ for which  $\var[f_i] \in (\sigma_{j}, \sigma_{j-1}]$.
We note that since we are in the ``Low-Variance Case'', then 
 \begin{equation} \label{eq:m}
 |I_j| \le C \cdot \sigma_{j}^{-1} \cdot \sigma_{j-1}^{-0.1}\cdot  \log^2(n/\eps)\;.
 \end{equation}
We handle two cases depending on whether $\sigma_{j-1}$ is big or not.

\paragraph{The case of $\sigma_{j-1} \ge 1/ (C \cdot \log(n/\eps))^{20}$ :}
In this case there are at most $O(\sigma_{j-1}^{-1.2}\cdot \log^2(n/\eps)) \le \poly\log(n/\eps)$ functions in $I_j$, each computed by a width-$w$ ROBP on at most $b$ bits. Thus, $F_j := \prod_{i\in I_j} f_i$ can be computed by a ROBP of length at most $n' = b\cdot \poly\log(n/\eps)$ and width at most $2w$. 
Using Theorem~\ref{thm:CHRTa} on $F_j$ (which has length $n'$ and width $2w$), with probability at least $1-\delta$ the spectral-norm of $\tilde{F_j}$ is at most $O((n' w)^3/\delta)$, thus any $\delta'_x$-biased distribution $O(\delta'_x \cdot (n' w)^3/\delta)$-fools  $\tilde{F_j} = \prod_{i \in I_j}{\tilde{f_i}(x)}$. For a large enough choice for $C'$, $O(\delta'_x \cdot (n' w)^3/\delta) \le \eps/n^{40}$ and we are done.

\paragraph{The case of $\sigma_{j-1} < 1/ (C \cdot \log(n/\eps))^{20}$ :}
In this case all variances in $I_j$ are certainly smaller than $0.5$, and hence for all $i\in I_j$, we have $\E[f_i]^2 = \E[f_i^2]-\var[f_i] = 1-\var[f_i] \in [0.5,1]$.
Let $$\mu_i = \E[f_i]\qquad\text{and}\qquad g_i(x) \triangleq \frac{f_i(x)}{\mu_i} - 1.$$
Then, $$\prod_{i}{f_i(x)} = \prod_{i}{\mu_i} \cdot (1+g_i(x)).$$
We have $\E[g_i] = 0$ and $\Var[g_i] = \Var[f_i]/\mu_i^2 \in [\var[f_i], \var[f_i] \cdot 2]$.
We will show that with high probability over $T$, any $\delta'_x$-biased distribution fools $\prod_{i}{\mu_i} \cdot \prod_{i}{(1+\tilde{g_i}(x))}$. 

For ease of notation, in this case we think of $I_j$ as $[m]$ and denote by $\sigma = \sigma_{j-1}$.
The proof strategy for this part follows the work of Gopalan and Yehudayoff \cite{GopalanY14}. 
We note that 
$$\prod_{i=1}^{m}{(1+\tilde{g_i}(x))} = 1+\sum_{k=1}^{m} S_k(\tilde{g_1}(x), \tilde{g_2}(x), \ldots, \tilde{g_{m}}(x)),$$ 
where $S_k$ is the $k$-symmetric polynomial given by $S_k(y_1, \ldots, y_m) = \prod_{R \subseteq [m], |R|=k}{\prod_{i \in R} y_i}$.
We show that $x$ and $T$ fool the low-degree symmetric polynomials. Then, the following theorem by Gopalan and Yehudayoff \cite{GopalanY14} bootstraps this to show that $x$ and $T$ also fool the sum of all high-degree symmetric polynomials.

\begin{theorem}[Gopalan-Yehudayoff Tail Inequalities \cite{GopalanY14}]\label{thm:GY}
Let $y_1, \ldots, y_{m} \in \R$.
Suppose 
$|S_{\ell}(y_1, \ldots, y_{m})| \le \frac{t^{\ell}}{\sqrt{\ell!}}$ 
and 
$|S_{\ell+1}(y_1, \ldots, y_m)| \le \frac{t^{\ell+1}}{\sqrt{(\ell+1)!}}$
for some $t$ and $\ell$. 
Then, for every 
$k \in \{\ell, \ldots, {m}\}$ 
it holds that 
$|S_{k}(y_1, \ldots, y_{m})| \le (6et)^{k} \cdot (\ell/k)^{k/2}$.
Furthermore, if $6et \le 1/2$, then
$$\sum_{k=\ell}^{{m}} |S_{k}(y_1, \ldots, y_{m})| \le 2\cdot (6et)^{\ell}.$$
\end{theorem}

\subsubsection*{Analyzing the Symmetric Polynomials}
From Eq.~\eqref{eq:m} and our assumption that 
$\sigma < 1/(C\cdot \log(n/\eps))^{20}$ 
we get that $m \le \sigma^{-1.3}$.
Recall that $C'$ is a sufficiently large constant and recall the definition of $\delta, \delta'_x, \delta_T$ from Eqs.~\eqref{eq:delta_T}, \eqref{eq:delta} and \eqref{eq:delta'_x}.
We set 
\begin{equation}
\ell \triangleq C' \cdot \log(n/\eps)/ \log(1/\sigma)
\end{equation}
In the following, we shall use the facts that $\sigma^{-\ell}, m^{\ell} \ll 1/\delta$ and $\delta'_x \ll \delta$.

\begin{claim}\label{claim:good}
Let $T \sim \D_p$.
Let $R \subseteq [m]$ be a set of size at most $\ell$.
Then, with probability at least $1-O(b \ell w)^3 \cdot \delta$ over the choice of $T$,
$\prod_{i\in R} \tilde{f_i}(x)$ has spectral-norm at most $1/\delta$.
\end{claim}
\begin{proof}
Note that $\prod_{i\in R} f_i(x)$ can be computed by a ROBP with length $b \cdot \ell \le O(b \cdot \log(n/\eps))$ and width $2w$ (as in the case where $\sigma_j$ is big).
Apply Theorem~\ref{thm:CHRTa} to $\prod_{i\in R} f_i(x)$. 
\end{proof}

We say that $T\subseteq[n]$ is a {\sf good} set if for all sets $R \subseteq [{m}]$ of  size at most $\ell$, 
the spectral-norm of $\prod_{i\in R}{\tilde{f}_i}$ is at most $1/\delta$.
We observe that by Claim~\ref{claim:good}, the probability that $T$ is good is at least 
$1-({m}+1)^{\ell}\cdot O(b\ell w)^3 \cdot \delta \ge 1-\eps/10n$ (using Eq.~\eqref{eq:m} and \eqref{eq:delta}).
\begin{claim}\label{claim:S_k small spectral}
If $T$ is good, then for any $k\le \ell+1$, 
$S_{k}(\tilde{g_1}, \tilde{g_2},\ldots, \tilde{g_{m}})$ has  spectral-norm at most 
$\delta^{-1}\cdot (4{m})^{k} \le \delta^{-2}$.
\end{claim}

\begin{proof}
We expand the $k$-symmetric polynomial:
$S_{k}(\tilde{g_1}(x),\ldots, \tilde{g_m}(x)) = \sum_{R \subseteq [m], |R|=k} \prod_{r\in R} \widetilde{g_r}(x)$.
Since $T$ is good, each summand has spectral-norm
\begin{align*} 
L_1\bigg(\prod_{r\in R}  \widetilde{g_r}(x)\bigg)
&=L_1\Bigg(\prod_{r\in R}  \Big(\frac{\widetilde{f_r}(x)}{\E[f_r]} -1\Big)\Bigg) 
\le L_1\bigg(\sum_{Q \subseteq R}  (-1)^{|R|-|Q|} \prod_{r \in Q} \frac{\widetilde{f_r}(x)}{\E[f_r]}\bigg) 
\le 2^{k} \cdot \delta^{-1}\cdot 2^{k}\;,
\end{align*}
(using $\E[f_r] \ge 1/2$).
Summing over all $\binom{m}{k} \le m^k$ summands completes the proof.
\end{proof}

We wish to show that with high probability the total variance under restrictions $\sum_{i}\var[\tilde{f}_i]$ is small. Towards this goal, we prove a bound on the $\ell$-th moment of the total variance.
\begin{claim}\label{claim:Var_fi ell}
$
\E_{T\sim \D_p}[(\sum_{i=1}^m \var[\tilde{f}_i])^{\ell}] \le 
 2 \cdot (2\sigma^{0.2})^\ell
$
\end{claim}
\begin{proof}
Fix $(i_1, \ldots, i_{\ell})\in [m]^{\ell}$, not necessarily distinct indices. By Lemma~\ref{lemma:vars_prod}
\begin{align*}
	\E_{T\sim \D_p}\left[\prod_{j=1}^{\ell} \Var[\tilde{f_{i_j}}]]\right] 
	&\le \E_{T\sim \Rp}\left[\prod_{j=1}^{\ell}\Var[\tilde{f_{i_j}}]\right] 
	+ \delta_T \cdot\prod_{j=1}^{\ell} \Var[{f_{i_j}}],
\end{align*}
from which we deduce
$$
\E_{T\sim \D_p}\left[\Big(\sum_{i=1}^m \var[\tilde{f}_i(z)]\Big)^{\ell}\right] \le 
 \E_{T\sim \Rp}\left[\Big(\sum_{i=1}^m \var[\tilde{f}_i(z)]\Big)^{\ell}\right] + \delta_T \cdot m^{\ell} \sigma^{\ell} \;.
$$
We are left to bound $\E_{T\sim \Rp}[(\sum_{i=1}^m \var[\tilde{f}_i])^{\ell}]$. 
By Fact~\ref{fact:bias-fnc-Fourier}, for any $i\in[m]$, the random variable $X_i = \Var[\tilde{f_i}]/\var[f_i]$ (whose value depends on the choice of $T\sim \Rp$) is bounded in $[0,1]$.
By Theorem~\ref{thm:HC}, its expected value is at most $\Var[f_i]^{0.5} \le \sigma^{0.5}$.
Taking $X = \sum_{i=1}^{m}X_i$, we get that $X$ is the sum of $m$ independent random variables bounded in $[0,1]$.
Using $m \le \sigma^{-1.3}$, we have that $\E[X] \le \sigma^{0.5} \cdot m \le \sigma^{-0.8}$.
Thus, by Chernoff's bounds, with probability at least $1-\exp(-\Omega(\sigma^{-0.8}))$
we have
$X \le 2\cdot \sigma^{-0.8}$.
In such a case 
$\sum_{i} \var[\tilde{f_i}] 
\le 2 \cdot \sigma^{-0.8} \cdot \sigma 
\le 2\sigma^{0.2}$.
We get $\E_{T\sim \Rp}[(\sum_{i=1}^m \var[\tilde{f}_i])^{\ell}] \le \exp(-\Omega(\sigma^{-0.8}))\cdot (\sigma m)^{\ell} + 
(2\sigma^{0.2})^{\ell}$, which gives
\[\E_{T\sim \D_p}\left[\Big(\sum_{i=1}^m \var[\tilde{f}_i]\Big)^{\ell}\right]  \le 
\delta_T \cdot m^{\ell} \sigma^{\ell}  + \exp(-\Omega(\sigma^{-0.8}))\cdot (\sigma m)^{\ell} + 
(2\sigma^{0.2})^{\ell} \le 2\cdot (2 \sigma^{0.2})^{\ell}.\qedhere
\]
\end{proof}

We say that a set $T\subseteq [n]$ is {\sf excellent} if $T$ is good and $\sum_{i}\var[\tilde{g}_i] \le \sigma^{0.1}$.

\begin{claim}
$\Pr_{T\sim \D_p}[\text{$T$ is not excellent}] \le \eps/10n + O(\sigma)^{0.1\ell} \le \eps/2n$	
\end{claim}
\begin{proof}
Note that $\sum_{i}\Var[\tilde{g_i}] \le 2\sum_{i}\Var[\tilde{f_i}]$ and apply Markov's inequality on $(2\sum_{i}\Var[\tilde{f_i}])^{\ell}$ using Claim~\ref{claim:Var_fi ell}.
\end{proof}

\begin{claim}\label{claim:symmetric-excellent}
Let $T$ be an excellent set. Let $\Dtagx$ be any $\delta'_x$-biased distributions.
Then, for $k = 1, \ldots, \ell+1$ we have
$$\E_{x\sim \Dtagx}[ S_{k}^2(\tilde{g}_1(x), \ldots, \tilde{g}_m(x))] \le  \frac{2}{k!} \cdot \sigma^{0.1 k}$$
and $$\left|\E_{x \sim \Dtagx}[S_{k}(\tilde{g}_1(x), \ldots, \tilde{g}_m(x))]\right| \le (\eps/n)^{C'}.$$
\end{claim}

\begin{proof}
	Recall that $\delta = (\eps/n)^{10C'}$ and $\delta'_x = (\eps/n)^{100C'}$. The first claim relies on the following:
	 \begin{enumerate}
		\item $S_k^2$ has small spectral-norm (using Claim~\ref{claim:S_k small spectral}, since $T$ is good) and hence is fooled by $\Dtagx$. In details, its spectral-norm is at most $L_1(S_k)^2 \le \delta^{-4}$ and $\Dtagx$ is $\delta'_x$-biased. Thus $$\Big|\E_{x\sim U_n}[S_k^2(\tilde{g}_1(x), \ldots, \tilde{g}_m(x))] - \E_{x\sim \Dtagx}[S_k^2(\tilde{g}_1(x), \ldots, \tilde{g}_m(x))]\Big| \le \delta^{-4} \cdot \delta'_x \le \delta \ll \frac{1}{k!} \cdot \sigma^{0.1 k}.$$
		\item 	The expectation of $S_k^2(\tilde{g}_1(x), \ldots, \tilde{g}_m(x))$ on a uniformly chosen $x$ is at most 
\begin{align*}\E_{x\sim U_n}[S_k^2(\tilde{g}_1(x), \ldots, \tilde{g}_m(x))] &= \sum_{T, T' \subseteq [m], |T|=|T'|=k} \E_{x\sim U_n} \bigg[\prod_{i\in T} \tilde{g}_i(x)\prod_{i'\in T'} \tilde{g}_{i'}(x)\bigg]\\
&=	\sum_{T\subseteq [m], |T|=k} \E_{x\sim U_n} \bigg[\prod_{i\in T} (\tilde{g}_i(x))^2\bigg] \tag{Since $\E[\tilde{g_i}] = 0$}\\
&
=	\sum_{T\subseteq [m], |T|=k} \prod_{i\in T} \var[{g}_i] \le  \frac{1}{k!} \cdot \Big(\sum_{i=1}^{m}{\var[\tilde{g_i}]}\Big)^{k} \le \frac{1}{k!} \cdot \sigma^{0.1k} \tag{Maclaurin's inequality}
\end{align*}

	\end{enumerate} 
	The second claim relies on the following:	
	\begin{enumerate}
		\item $S_k$ has small spectral-norm (using Claim~\ref{claim:S_k small spectral}, since $T$ is good) and hence is fooled by $\Dtagx$. In details, its spectral-norm is at most $\delta^{-2}$ and $\Dtagx$ is $\delta'_x$-biased. Thus $$\Big|\E_{x\sim U_n}[S_k(\tilde{g}_1(x), \ldots, \tilde{g}_m(x))] - \E_{x\sim \Dtagx}[S_k(\tilde{g}_1(x), \ldots, \tilde{g}_m(x))]\Big| \le \delta^{-2} \cdot \delta'_x \le \delta \le (\eps/n)^{C'}.$$

		\item The expectation of $S_k(\tilde{g}_1(x), \ldots, \tilde{g}_m(x))$ on a uniformly chosen $x$ is $0$.	
		\qedhere
		\end{enumerate} 
\end{proof}

The next lemma combined with Claim~\ref{claim:symmetric-excellent} concludes the low-variance case, since it shows that with high probability, $T$ is excellent, and then $\Dtagx$ is an $(\eps/n^{40})$-PRG for $\prod_{i=1}^{m}{\tilde{f_i}}$ (for a sufficiently large choice of $C'$).

\begin{lemma}
If $T$ is excellent, then $\E_{x\sim \Dtagx}[	\prod_{i=1}^{m}{\tilde{f_i}}] = (\prod_{i=1}^{m}\mu_i) \pm (\eps/n)^{\Omega(C')}$.
\end{lemma}
\begin{proof}
Let $x\sim \Dtagx$, 
and let $E$ be the event that
$|S_{\ell}(\tilde{g_1}(x), \ldots, \tilde{g_m}(x))| \le \frac{t^\ell}{\sqrt{\ell!}}$ 
and 
$|S_{\ell+1}(\tilde{g_1}(x), \ldots, \tilde{g_m}(x))| \le \frac{t^{\ell+1}}{\sqrt{(\ell+1)!}}$.
Picking $t = \sigma^{0.01}$, and using Claim~\ref{claim:symmetric-excellent} the event $E$ happens with probability at least $1-\sigma^{\Omega(\ell)} \ge 1-(\eps/n)^{\Omega(C')}$.
Assuming $E$ occurs, 
 Theorem~\ref{thm:GY} gives
$$\sum_{k=\ell}^{m} |S_{k}(\tilde{g_1}(x), \ldots, \tilde{g_m}(x))| \le 2\cdot (6et)^{\ell} \le \sigma^{\Omega(\ell)} \le (\eps/n)^{\Omega(C')}.$$
Furthermore, for sets of smaller cardinality, i.e., for $k \in \{1,\ldots, \ell-1\}$, Claim~\ref{claim:symmetric-excellent} gives
$$\Big|\E_{x\sim \Dtagx}[S_{k}(\tilde{g_1}(x), \ldots, \tilde{g_m}(x))]\Big| \le (\eps/n)^{C'} \qquad\text{and}\qquad\Big|\E_{x\sim \Dtagx}[S_{k}^2(\tilde{g_1}(x), \ldots, \tilde{g_m}(x))]\Big| \le 1\;.$$
We would like to bound $|\E_{x\sim \Dtagx}[S_{k}(\tilde{g_1}(x), \ldots, \tilde{g_m}(x)) \cdot \one_{E}]|$ for $k\in \{1,\ldots, \ell-1\}$. Towards this end,
we consider the expectation of $S_{k}(\tilde{g_1}(x), \ldots, \tilde{g_m}(x))$ by partitioning into the two cases depending on whether the event $E$ occurred or not.
\begin{align*} \E_{x\sim \Dtagx}&[S_{k}(\tilde{g_1}(x), \ldots, \tilde{g_m}(x))] \\&= \E_{x\sim \Dtagx}[S_{k}(\tilde{g_1}(x), \ldots, \tilde{g_m}(x)) \cdot \one_{E}] + \E_{x\sim \Dtagx}[S_{k}(\tilde{g_1}(x), \ldots, \tilde{g_m}(x)) \cdot \one_{\bar{E}}]\\
&= \E_{x\sim \Dtagx}[S_{k}(\tilde{g_1}(x), \ldots, \tilde{g_m}(x)) \cdot \one_{E}] \pm \sqrt{\E_{x\sim \Dtagx}[S_{k}^2(\tilde{g_1}(x), \ldots, \tilde{g_m}(x))] \cdot \Pr[\bar{E}]	\tag{Cauchy-Schwarz}}\\
&= \E_{x\sim \Dtagx}[S_{k}(\tilde{g_1}(x), \ldots, \tilde{g_m}(x)) \cdot \one_{E}] \pm \sqrt{\Pr[\bar{E}]}\\
&=\E_{x\sim \Dtagx}[S_{k}(\tilde{g_1}(x), \ldots, \tilde{g_m}(x)) \cdot \one_{E}] \pm (\eps/n)^{\Omega(C')}
\end{align*}
Thus, $|\E_{x\sim \Dtagx}[S_{k}(\tilde{g_1}(x), \ldots, \tilde{g_m}(x)) \cdot \one_{E}]| \le  (\eps/n)^{\Omega(C')}$ and we get 
$$
\E_{x\sim \Dtagx}\left[\prod_{i=1}^{m}\tilde{g_i}(x) \cdot \one_{E}\right] = \E_{x\sim \Dtagx}[\one_{E}] \pm  \sum_{k=1}^{m} |\E_{x\sim \Dtagx}[S_{k}(\tilde{g_1}(x), \ldots, \tilde{g_m}(x)) \cdot \one_{E}]| = 1 \pm (\eps/n)^{\Omega(C')}\;.$$
Since the $\tilde{f_i}$'s and $\mu_i$'s are bounded in $[-1,1]$, we get 
\begin{align*}
\E_{x}\left[ \prod_{i}\tilde{f_i}\right] &= 
\E_{x}\left[  \prod_{i}\tilde{f_i} \cdot \one_{E}\right] + \E_{x}\left[ \prod_{i}\tilde{f_i} \cdot \one_{ \neg E}\right]\\
&= \left(\prod_{i}{\mu_i} \cdot \E_{x\sim \Dtagx}\left[ \prod_{i=1}^{m} \tilde{g_i}(x) \cdot \one_{E}\right]\right)  \pm \Pr[\neg E] \\
&= \prod_{i}{\mu_i} \cdot \left(1\pm (\eps/n)^{\Omega(C')}\right)  \pm (\eps/n)^{\Omega(C')} =  \Big(\prod_{i}{\mu_i}\Big)   \pm (\eps/n)^{\Omega(C')}.\qedhere	
\end{align*}
\end{proof}

\subsection{High-Variance Case}

In the high-variance case, there exists a $\sigma \in (0,1]$ and an interval $I_\sigma = \{i: \Var[f_i]  \in (0.4\cdot\sigma^{1.1},\sigma]\}$ (the constant $0.4$ handles the case $\sigma=1$) satisfying:
$$
\sum_{i \in I_\sigma} \Var[f_i] > C\cdot \sigma^{-0.1}\cdot  \log^2(n/\eps)\;.
$$
In this case, the expected value of $\prod_{i=1}^m{f_i}$ under the uniform distribution is rather small:
\begin{align*}
\abs{\E\left[\prod_{i=1}^{m} f_i\right]}
 = \prod_{i=1}^{m} |\E[f_i]|
=  \prod_{i=1}^{m}{\sqrt{1-\var[f_i]}}
\le e^{-\sum_{i=1}^{m}{\var[f_i]/2}} \le e^{-C\cdot \log^2(n/\eps)/2} \le \eps/2.
\end{align*}
Recall that the pseudorandom restriction samples a set $T$ according to some $\delta_T$-biased distribution $\D_p$ with marginals $p$, and a partial assignment to the bits in $T$ according to some $\delta_x$-biased distribution $\Dx$. 
In the high variance case, it suffices to show that $\abs{\E_{T\sim \D_p, x\sim \Dx}\left[\prod_{i=1}^{m}{\tilde{f_i}(x)}\right]} \le \eps/2$.
Fix $T, x$.  Denote by $f^T_{i,x}(y) = (f_{i})_{T|y}(x)$.
Similarly to the calculation in the  case of the uniform distribution, we have
\begin{align*}
\abs{\prod_{i=1}^{m}{\tilde{f_i}(x)}} = \abs{\prod_{i=1}^{m}{\E_{y\sim U_{[n]\setminus T}}[f^T_{i,x}(y)]}}\le e^{-\sum_{i=1}^{m} \var[f^T_{i,x}]/2}
\end{align*}
Thus, it suffices to show that for most $T\sim \D_p$, $x\sim \Dx$ we have $\sum_{i=1}^{m} \var[f^T_{i,x}] \ge 10 \cdot \log(1/\eps)$.

\begin{theorem}[Theorem~\ref{thm:pseudorandom-restriction-XOR-short} - High Variance Case]\label{thm:main high var}
	With probability $1-\eps/4$ over $T\sim \D_p$ and $x\sim \Dx$, it holds that  $\sum_{i\in I_{\sigma}}{\Var[f^T_{i,x}]} \ge 10 \cdot \log(1/\eps)$ .
\end{theorem}

\begin{proof}
Denote by $\Tvar := \sum_{i\in I_{\sigma}}\var[f_i]$.
By our assumption, $\Tvar \ge C \cdot \log^2(n/\eps) \cdot \sigma^{-0.1} \ge 5 \sigma^{-0.1}$.
Since all functions in $I_{\sigma}$ have variance at least $0.4 \cdot \sigma^{1.1}$ we have 
\begin{equation}\label{eq:Tvar vs m}
|I_{\sigma}| \le  \Tvar \cdot \tfrac{1}{0.4} \cdot \sigma^{-1.1} \le \Tvar^{12}
\end{equation}
We remark that in this case, unlike the low-variance case, we do not know how to handle large $\sigma$ easily, so for the rest of the proof $\sigma$ can be anything between $2^{-1-b}$ and $1$.

Fix $T$ and $x$. We expand $\var[f^T_{i,x}]$ 
$$
\var[f^T_{i,x}] = \E_{y\sim U_{[n]\setminus T}}[f^T_{i,x}(y)^2] - \E_{y\sim U_{[n]\setminus T}}[f^T_{i,x}(y)]^2 = 1-\tilde{f_i}(x)^2 \;.
$$
For any fixed $T$, using $\E[f_i] =\E[\tilde{f_i}]$ gives
\begin{align*}
	\E_{z\sim U_T} [\var[f^T_{i,z}]]   
&= 1- \E[(\tilde{f_i})^2] 
= (1-\E[f_i]^2) - (\E[(\tilde{f_i})^2] - \E[\tilde{f_i}]^2)
= \var[f_i] - \var[\tilde{f_i}]
\end{align*}

\begin{claim}[Most $T$'s preserve variance in expectation]\label{claim:preserve_var}
With probability at least $1-\eps/16$ over the choice of  $T \sim \D_p$, it holds that $\E_{z\sim U_T} \left[\sum_{i \in I_{\sigma}}{\Var[f^T_{i,z}]}\right] \ge \Tvar/2.$
\end{claim}
\begin{proof}
Since 
$\E_{z\sim U_T}[\var[f^T_{i,z}]] = \var[f_i] - \var[\tilde{f_i}]$, it suffices to show that with probability $1-\eps/16$ over the choice of $T\sim \D_p$ we have $\sum_{i} \var[\tilde{f_i}] \le \sum_{i} \var[f_i]/2$.
To show that 
$\sum_{i} \var[\tilde{f_i}]$ 
is well-concentrated we analyze its $k$-th moment for $k = C'  \log(1/\eps)$ where $C'$ is a sufficiently large constant.
\begin{align*}
\E_{T\sim \D_p}\left[\Big(\sum_{i \in I_{\sigma}} \var[\tilde{f_i}]\Big)^{k}\right]
= \sum_{i_1, i_2, \ldots, i_{k}\in I_{\sigma}} \E_{T}\left[\prod_{j=1}^{k} 	\var[\tilde{f_{i_j}}]\right]\;.
\end{align*}
Fix $i_1, \ldots, i_k \in I_{\sigma}$, (not necessarily distinct), then by Lemma~\ref{lemma:vars_prod}
\begin{align*}
\E_{T\sim \D_p}\left[\prod_{j=1}^{k} 	\var[\tilde{f_{i_j}}]\right] &\le   \E_{T\sim \Rp}\left[\prod_{j=1}^{k} 	\var[\tilde{f_{i_j}}]\right] + \delta_T	\cdot \prod_{j=1}^{k} 	\var[f_{i_j}]
\end{align*}
Overall, we get
\begin{align*}\E_{T\sim \D_p}\left[\Big(\sum_{i\in I_{\sigma}} \var[\tilde{f_i}]\Big)^{k}\right] \le  \E_{T\sim \Rp}\left[\Big(\sum_{i\in I_{\sigma}} \var[\tilde{f_i}]\Big)^{k}\right] + \delta_T \cdot \Tvar^{k}\;.
\end{align*}
To bound $\E_{T\sim \Rp}[(\sum_{i=1}^{m} \var[\tilde{f_i}])^{k}]$
we use the fact that by Theorem~\ref{thm:HC} 
$$\E_{T \sim \Rp}[\var[\tilde{f_i}]] \le p \cdot \var[f_i] \le 0.1 \cdot \var[f_i]$$
and then by Chernoff's bound 
$\sum_{i\in I_{\sigma}} \var[\tilde{f_i}] \le 0.2 \cdot \Tvar$ 
with probability at least $1-\exp(-\Omega(\Tvar))$.
Since $\sum_{i} \var[\tilde{f_i}]$ is always upper bounded by $\Tvar$, the $k$-moment of the sum is at most
$$(0.2 \cdot \Tvar)^{k} + (\Tvar)^{k} \cdot \exp(-\Omega(\Tvar)) \le 2(0.2 \cdot \Tvar)^{k}$$
We get that 
$\E_{T\sim \D_p}[(\sum_{i\in I_{\sigma}} \var[\tilde{f_i}])^{k}] \le 2(0.2 \cdot \Tvar)^{k} + \delta_T \cdot \Tvar^{k}$.
Since $\delta_T \ll 2^{-4 k}$ this is at most 
$3\cdot (0.2 \cdot \Tvar)^{k}$.
Thus, using Markov's inequality, the probability that $ \sum_{i\in I_{\sigma}} \var[\tilde{f_i}]  \ge  0.5 \cdot \Tvar$ is at most $3\cdot (0.2/0.5)^{k} \le \eps/16$ which completes the proof.
\end{proof}

Let \begin{equation}\label{eq:def ell}
	\ell \triangleq C'\cdot \log(n/\eps)/\log(|I_{\sigma}|)
\end{equation}
where $C'$ is a sufficiently large constant declared before Eq.~\eqref{eq:delta_T}. 
Assume that $\ell$ is an even integer.
Recall that $\delta = (\eps/n)^{-10C'} = |I_{\sigma}|^{-10 \ell}$.
We again define $T$ to be a {\sf good} set if $\prod_{i\in R}\tilde{f_i}$ has spectral-norm at most $1/\delta$ for all sets $R\subseteq I_{\sigma}$ of size at most $\ell$. As in Claim~\ref{claim:good} the probability that $T$ is good is at least $1-(|I_{\sigma}|+1)^{\ell} \cdot O(\ell b w)^3 \cdot \delta \ge 1-\eps/16$.
We define $T$ to be an {\sf excellent} set if $T$ is good and Claim~\ref{claim:preserve_var} holds for $T$. Then, $\Pr[T\text{~is excellent}] \ge 1-\eps/8$.
\begin{claim}
 If $T$ is a good set, then at most $\ell$ of the $\tilde{f_i}$'s have $L_1(\tilde{f_i}) \ge \delta^{-1/\ell}$.
\end{claim}
\begin{proof}
	If $\tilde{f_{i_1}}, \ldots, \tilde{f_{i_\ell}}$ have
	$L_1(\tilde{f_{i_j}})\ge \delta^{-1/\ell}$, then their product  has spectral-norm at least $\delta^{-1}$, since $L_1(\prod_{j=1}^{\ell}\tilde{f_{i_j}}) = \prod_{j=1}^{\ell} L_1(\tilde{f_{i_j}})$ for functions defined on disjoint variables.
\end{proof}

Fix an excellent set $T$. Let $G$ be the of indices $i\in I_{\sigma}$ with $L_1(\tilde{f_i})\le \delta^{-1/\ell}$. We show that with high probability over $x$, $\sum_{i\in G}{\var[f^T_{i,x}]} \ge 0.1 \cdot \Tvar$. We denote by 
$$\error_i(x) := \var[f^T_{i,x}] - \E_{z \sim U}[\var[f^T_{i,z}]] = \var[f^T_{i,x}] - (\var[f_i] - \var[\tilde{f_i}]).$$
Obviously $\E_{z\sim U}[\error_i(z)] = 0$ and $\error_i$ is bounded in $[-\sigma,1]$.
Furthermore, we have that $$\error_i(x) = (1-\tilde{f_i}(x)^2) -(1 - \E_{z\sim U}[\tilde{f_i}(z)^2]) = \E_{z\sim U}[\tilde{f_i}(z)^2] - \tilde{f_i}(x)^2$$
Thus, the error term have small spectral-norm since $L_1(\error_i) \le L_1(\tilde{f_i})^2$. We use this fact to bound 
$\E_{x\sim \Dx}[ (\sum_{i\in G} \error_i(x))^{\ell}]$. (recall that $\ell$ is an even integer.)
\begin{claim}\label{claim:small l-moment err}
	$$\E_{x\sim \Dx}\Big[ \big(\sum_{i\in G} \error_i(x)\big)^{\ell}\Big] \le 2 \cdot (\ell \cdot \Tvar)^{\ell/2}.$$
\end{claim}
\begin{proof}
	The spectral-norm of $(\sum_{i\in G} \error_i(x))^{\ell}$ is at most $(|G| \cdot \delta^{-2/\ell})^{\ell}  = |G|^{\ell} \cdot \delta^{-2}$. Thus, any $\delta_x$-biased distribution fools $(\sum_{i\in G} \error_i(x))^{\ell}$ with error at most $\delta_x \cdot |G|^{\ell} \cdot \delta^{-2}$ and we get
	$$
	\E_{x\sim \Dx}\Big[ \big(\sum_{i\in G} \error_i(x)\big)^{\ell}\Big] \le \E_{z\sim U}\Big[ \big(\sum_{i\in G} \error_i(x))\big)^{\ell}\Big]+ \delta_x \cdot |G|^{\ell} \cdot \delta^{-2}.
	$$
	
	To bound
	$\E_{z\sim U}[ (\sum_{i\in G} \error_i(z))^{\ell}]$ we use Lemma~\ref{lemma:tail_bounds}.
	We observe that $\{\error_i(z)\}_{i\in G}$ are independent random variables, where each $\error_i(z)$ is bounded in $[-\var[f_i],1]$ with mean zero, and hence $\var[\error_i] \le \var[f_i]$ (See Lemma~\ref{lemma:var}). Applying Lemma~\ref{lemma:tail_bounds} gives
	$$
	\E_{z\sim U}[ (\sum_{i\in G} \error_i(z))^{\ell}] \le  
	\max\{\ell^{\ell}, (\ell \cdot \Tvar)^{\ell/2}\}.$$

Since $\ell \le \sqrt{\ell \cdot \Tvar}$, the upper bound on $\E_{z\sim U}[ (\sum_{i\in G} \error_i(z))^{\ell}]$ is at most $\left( \ell \cdot \Tvar\right)^{\ell/2}$.
Finally, the upper bound with respect to $x\sim \Dx$ is at most \[
\E_{x\sim \Dx}\Big[ \big(\sum_{i\in G} \error_i(x)\big)^{\ell}\Big] 
\;\le\;  \left(\ell  \cdot  \Tvar\right)^{\ell/2} 
	+ \delta_x \cdot |G|^{\ell} \cdot \delta^{-2}
\;\le\; 2\cdot \left(\ell  \cdot  \Tvar\right)^{\ell/2}\;.\qedhere
\]
\end{proof}
Using Markov's Inequality and Claim~\ref{claim:small l-moment err} gives \begin{align*}
 \Pr_{x\sim\Dx} \left[\Big|\sum_{i\in G} \error_i(x)\Big| \ge \Tvar / 4\right]
 \le 2\cdot \left(\frac{\sqrt{\ell \cdot \Tvar}}{\Tvar/4}\right)^{\ell} \le O(\sqrt{\ell/\Tvar})^{\ell}\le O(1/\Tvar)^{\ell/4}\;.
 \end{align*}
using $\Tvar \ge \Omega(\log^2(n/\eps))$ and $\ell \le O(\log(n/\eps))$ in the last inequality.
Furthermore, using Eqs.~\eqref{eq:Tvar vs m} and \eqref{eq:def ell}: $O(1/\Tvar)^{\ell/4} \le \abs{I_{\sigma}}^{-\Omega(\ell)} \le (\eps/n)^{\Omega(C')} \le \eps/8$.
In the complement event,
$$
\sum_{i\in G} \var[f^T_{i,x}] =  \sum_{i\in G}(\E_{z}[\var[f^T_{i,z}]] + \error_i(x)) \ge \Tvar/2 - \ell  - \Tvar/4 \ge 0.1 \cdot \Tvar.
$$
Since $\Tvar \ge \Omega(\log^2(n/\eps))$, we get that with  probability at least $1-\eps/4$ over $T\sim \D_p$ and $x\sim \Dx$,  $\sum_{i} \var[f^T_{i,x}]  \ge 10 \log(1/\eps)$. (End of Proof of Theorem~\ref{thm:main high var})
\end{proof}
\newcommand{\VBad}{\mathsf{VarBad}}
\newcommand{\V}{\mathsf{Var}}
\newcommand{\Bad}{\mathsf{Bad}}
\newcommand{\Good}{\mathsf{Good}}

\newcommand{\GMany}{\mathbf{G_{\oplus Many}}}
\newcommand{\GXOR}{\mathbf{GXOR}}

\section{Assigning all the variables: a pseudorandom generator for the XOR of short ROBPs}\label{sec:assign_all}

In Theorem~\ref{thm:pseudorandom-restriction-XOR-short}, we proved that we can pseudorandomly assign $p$-fraction of the coordinates of $f(x) = \prod_{i=1}^{m}{f_i(x)}$, while maintaining its acceptance probability up to an additive error  of $\eps$, using $\tilde{O}(\log(n/\eps)\cdot \log(b))$ random bits.
In this section, we will construct a pseudorandom generator $\eps$-fooling $f$, by applying Theorem~\ref{thm:pseudorandom-restriction-XOR-short} $\poly\log\log(n/\eps)$ times, combined with Lovett's \cite{Lovett08} or Viola's \cite{Viola08} pseudorandom generator for low-degree polynomials, and CHRT's pseudorandom generator for constant-width ROBPs~\cite{CHRT17}. Our main result is:
\assignall*

\paragraph{Assigning $0.9999$-fraction of the variables.}
The first step is rather standard. 
By making $t$ recursive calls to Theorem~\ref{thm:pseudorandom-restriction-XOR-short} we can assign all but $(1-p)^t \le e^{-pt}$ fraction of the coordinates while maintaining the acceptance probability.
Note that we rely on the fact that under restrictions, the restricted function is still of the form $\prod_{i=1}^{m}{g_i(x)}$ where each $g_i$ is a ROBP of width $w$ that depends on at most $b$ bits (in other words, the class of functions we are trying to fool is closed under restrictions).
Setting $t = O(1/p)$, we can assign $0.9999$ fraction of the inputs bits while changing the acceptance probability by at most $\eps/n$.

\begin{claim}\label{claim:assigning-0.9999}
Let $f = \prod_{i=1}^{m} f_i(x)$, with block-length $b$.
Then, there is a pseudorandom restriction $\rho = (J,y)$ 
using at most $
 O\big(\log(b) +\log\log(n/\eps)\big)^{2w+1} \cdot \log(n/\eps)$ random bits, and changing the acceptance probability by at most $(\eps/n)$.
Furthermore, $J$ is $(\eps/n)^{\omega(1)}$ biased with marginals $0.0001$, and
$y$ is $(\eps/n)^{\omega(1)}$ biased.
\end{claim}
\begin{proof}[Proof Sketch]
Apply Theorem~\ref{thm:pseudorandom-restriction-XOR-short} with error $\eps/n^2$ for $t = \log(0.0001)/\log(1-p) = O(1/p)$ times recursively, with independent random bits per each iteration.
Denoting by $J_0 = [n]$, this generates $t$-pseudorandom restrictions
$(J_1, y_1), (J_2, y_2), \ldots, (J_t,y_t)$ where $J_i \subseteq J_{i-1}$  and $y_i \in \pmone^{J_{i-1} \setminus J_i}$ for all $i\in [t]$.
To be more precise, for each $i$, $T_i = (J_{i-1} \setminus J_{i})$ is the $\delta_T$-biased subset with marginals $p$ in the description of the generator in Theorem~\ref{thm:pseudorandom-restriction-XOR-short} and $y_i\in \pmone^{T_i}$ is its  assignment sampled from a  $\delta_x$-biased distribution.
We take $J = J_t$ and $y\in \pmone^{[n]\setminus J}$ to be the concatenation of $y_1, \ldots, y_t$.
 By the hybrid argument, $$\Big|\E_{z\sim U_n}[f] - \E_{J,y}\E_{x\sim U_J}[f(\Sel_J(x,y))]\Big|\le (\eps/n^2)\cdot t \le \eps/n.$$
The amount of random bits used to sample the restriction is 
$$O(p^{-1}
 \cdot w  \log (n/\eps)  (\log\log(n/\eps) +  \log(b))
 \le 
 O\big(\log(b) +\log\log(n/\eps)\big)^{2w+1} \cdot \log(n/\eps).$$

Next, we claim that $J$ is a $(\eps/n)^{\omega(1)}$-biased with marginals $0.0001$.
Recall that  $J_{i-1}\setminus J_i$ are $\delta_T$-biased with marginals $p$, and that $\delta_T = (\eps/n)^{\omega(1)}$. 
By Claim~\ref{claim:inclusion-exclusion}, for any subset $S\subseteq [n]$ of size at most $\log(\delta_T)/10$ we have 
\begin{align*}\Pr[S\subseteq J] &= \Pr[S\subseteq J_0] \cdot \Pr[S\subseteq J_1|S\subseteq J_0]\cdots \Pr[S\subseteq J_t|S\subseteq J_{t-1}] \\
	&= ((1-p)^{|S|} \pm 2^{|S|} \delta_T)^{t} =(0.0001)^{|S|} \pm (\eps/n)^{\omega(1)}
\end{align*}
which implies by monotonicity that for larger subsets $S$, we have 
$\Pr[S\subseteq J]  \le (\eps/n)^{\omega(1)}$.

Finally, we claim that conditioned on $J$ and in fact for any choice of $J_1, \ldots, J_t$, $y$ is $(\eps/n)^{\omega(1)}$-biased.
This is due to the fact that $y$ is the concatenation of $y_1, \ldots, y_t$ where each $y_i$ is $\delta_x = (\eps/n)^{\omega(1)}$-biased.
\end{proof}

We would like to claim that $f = \prod_{i=1}^{m}f_i$ simplifies after assigning $0.9999$ of the coordinates. 
For a particular function $f_i$, with high probability, at least $1-1/32^b$, the block length decreases under a random restriction by a factor of $2$.
This is due to the fact that on expectation  at most $0.0001\cdot b$ of the variables will survive, and we can apply Chernoff's bound.
Now, if $m\le 16^b$, we can apply a union bound and get that with high probability the block-length decreases by a factor of $2$ in all functions $f_1, \ldots, f_m$ simultaneously.
We seem to have been making progress, going from block-length $b$ to block-length $b/2$, and we might hope that $\log(b)$ iterations of Claim~\ref{claim:assigning-0.9999} are enough to get a function that depends on $O(1)$ many variables (which is easy to fool). 
But, in order to carry the argument, even in the second step, we need to be able to afford the union bound on all functions. 
Ideally, the number of functions that are still alive also decreases from at most $16^b$ to at most $16^{b/2}$, and a similar union bound works replacing $b$ by $b/2$. We can continue similarly as long as in each iteration the block-length decreases by half and the number of functions by a square root.

We run into trouble if at some iteration we have more than $16^{b'}$ functions of block-length $b'$.
The first observation is that in this case the total variance of the functions is extremely high, exponential in $b'$. Recall that the expected value of the product is exponentially small in the total variance. This means that the expected value of the product is doubly-exponentially small in $b'$.
The second observation is that under $(1-\alpha)$-random restrictions, on average, the total variance decreases by a factor of $\alpha$.
Hence, we aim to apply a pseudorandom restriction assigning $(1-\exp(-b'))$ fraction of the variables alive, while keeping the total variance higher than $\log(n/\eps)$. This restriction is extremely aggressive, keeping only a polynomial fraction of the remaining variables alive (compared to say a constant fraction in Claim~\ref{claim:assigning-0.9999}).
However, we claim that in this case, such a restriction maintains the total variance high and thus the expected value of $\prod_{i=1}^{m}{f_i(x)}$ small (at most $\poly(\eps/n)$) in absolute value.

The nice thing about these ``aggressive pseudorandom restrictions'' is that they keep variables alive with such small probability that with high probability each function $f_i$ will depend on at most $O(1)$ variables after the restriction, except for a small number of functions covering at most $O(\log(n/\eps))$ ``bad variables''. This will allow us to fool the restricted function using Lovett's \cite{Lovett08} or Viola's \cite{Viola08} pseudorandom generator for low-degree polynomials.
In the next section, we explain how to handle this case in more details. 
Then, in Section~\ref{sec:4.2} we describe as a thought experiment a ``fake PRG'': an adaptive process that fools the XOR of short ROBPs, but depends on the function being fooled. In Section~\ref{sec:4.3} we show how to eliminate the adaptiveness and construct a true PRG for this class of functions. 

\subsection{PRG for the XOR of many functions with block-length $b$}\label{sec:4.1}

Let $\mathcal{F}_{b,n,t}$ be the class of functions of the form
$f(x) = f_0(x) \cdot \prod_{i=1}^{m}{f_i(x)}$
where $f_0, \ldots, f_m$ are Boolean functions on disjoint sets of variables, $f_0$ (the `junta') depends on at most $t$ variables, $f_1, \ldots, f_m$ are {\bf non-constant} and depend on at most $b$ variables and 
$16^b \le m \le 2\cdot 16^{2b}$.

\begin{lemma}\label{lemma:Gb}
There exists a constant $C>0$ such that the following holds.
For all $n,b,t$ such that  $C\cdot \log\log(n/\eps) \le b \le \log(n)$, there exists a log-space explicit pseudorandom generator $\GMany(b,n,t,\eps):\pmone^{O(t + \log n/\eps)} \to \pmone^n$ that $\eps$-fools $\mathcal{F}_{b,n,t}$.
\end{lemma}

\begin{algorithm}[H]
\caption{The Pseudorandom Generator $\GMany(b,n,t,\eps)$} 
\begin{algorithmic}[1]
\Require{A block-length $b$, the output length $n$, a junta-size $t$, an error parameter $\eps\in (0,1)$}
\State Set $x:= 1^n$
\State Pick $T\subseteq [n]$ using a $(\eps/n)^{10C}$-biased distribution with marginals $2^{-b}$.
\State Assign coordinates of $x$ in $[n]\setminus T$ using a $(\eps/n)^{10C}$-biased distribution.
\State Assign coordinates of $x$ in $T$ using Viola's generator with error $\frac{\eps}{4} \cdot (\frac{\eps}{n})^{C} \cdot 2^{-t}$ and degree~$16$.
\State \Return $x$.
\end{algorithmic}\label{test:1}
\end{algorithm}

\begin{lemma}\label{lemma:aggressive-real}
Let $C>0$ be a sufficiently large constant.
Let $C\cdot \log\log(n/\eps) \le b \le \log(n)$ be some integer.
Let $f_1, \ldots, f_m$ be non-constant Boolean functions that depend on disjoint sets of at most $b$ variables each.
Assume $m \ge 16^b$.
Suppose $T$ is $(\eps/n)^{10C}$-biased distribution with marginals $2^{-b}$.
Suppose $x$ is sampled from a $(\eps/n)^{10C}$-biased distribution.
Then, with  probability at least $1-(\eps/n)^{C/4}$, at least $4^b$ of the functions $(f_i)_{T|x}$ will be non-constant. 
\end{lemma}

\begin{proof}
Without loss of generality $m = 16^b$.
Let $k = C \log(n/\eps)/b$, and note that $k \le 2^b$ since $b\ge C \cdot \log \log(n/\eps)$ for a sufficiently large constant $C>0$.

Let $B_1, \ldots, B_m\subseteq [n]$ be the disjoint sets of variables on which $f_1, \ldots, f_m$ depend respectively.
For any function $f_i: \pmone^{B_i} \to \pmone$, there exists a sensitive pair of inputs $(\alpha^{(i)},\beta^{(i)}) \in \pmone^{B_i}$ such that $\alpha^{(i)}$ and $\beta^{(i)}$ differ in exactly one coordinate $j_i$ and such that $f_i(\alpha^{(i)}) \neq f_i(\beta^{(i)})$.
We say that the sensitive pair ``survives'' the random restriction defined by $(T,x)$ if both $\alpha^{(i)}$ and $\beta^{(i)}$ are consistent with the partial assignment defined by the restriction (i.e., if they agree with $x$ on $B_i \setminus T$).
For each function $f_1,\ldots, f_m$ fix one sensitive pair $(\alpha^{(1)}, \beta^{(1)}), \ldots, (\alpha^{(m)}, \beta^{(m)})$ and denote by $\cE_1, \ldots, \cE_m$ the events that these sensitive pairs survive.
Next, we claim that $\cE_1, \ldots, \cE_m$ are almost $k$-wise independent.
We compare them to the events $\cE'_1, \ldots, \cE'_m$ that indicate whether the sensitive pairs survive under a truly random restriction sampled from $\mathcal{R}_{2^{-b}}$.
Denote by $p_i = \Pr(\cE'_i)$.
Observe that $p_i \ge 2^{1-2b}$ since in order for the pair to survive it is enough that the sensitive coordinate remains alive (happens with probability $2^{-b}$) and that the partial assignment on the remaining coordinates agrees with $\alpha^{(i)}$ (happens with probability at least $2^{1-b}$). 
Then,
	$$
	\E\left[\Big(\sum_{i=1}^{m}{(\one_{\cE_i}-p_i)}\Big)^k\right] \le 
		\E\left[\Big(\sum_{i=1}^{m}(\one_{\cE'_i}-p_i)\Big)^k\right] + (2m)^k \cdot \max_{K\subseteq [m]: |K|\le k} \left|\E[\prod_{i\in K} \one_{\cE_i}] - \E[\prod_{i\in K} \one_{\cE'_i}]\right|.$$ 
	We upper bound the first and second summands separately. 
	By Lemma~\ref{lemma:tail_bounds} and Lemma~\ref{lemma:var}, the first summand is upper bounded by $\max\{k^k, (Vk)^{k/2}\}$ where 
	$$
	V := \sum_{i=1}^{m}{p_i}.
	$$ 
	Since $V \ge m\cdot 2^{1-2b} = 2\cdot 4^b  \ge k$, the first summand is upper bounded by $(Vk)^{k/2}$.
		
	Next, we upper bound the second summand. By Vazirani's XOR lemma, since $x$ is $(\eps/n)^{10C}$-biased, we have that the marginal distribution of
	 any set of at most $k\cdot b$ bits in $x$ is $(\eps/n)^{10C} \cdot 2^{kb/2}$-close to uniform in statistical distance.
	Since $T$ is $(\eps/n)^{10C}$-biased with marginals $2^{-b}$, using Claim~\ref{claim:inclusion-exclusion} we have that the marginal distribution on any set of at most $k$ coordinates in $T$ is $(\eps/n)^{10C} \cdot 4^{k}$-close in statistical distance to the distribution sampled according to $\mathcal{R}_{2^{-b}}$.
	Thus, $|\E[\prod_{i\in K} {\one_{\cE_i}}] - \E[\prod_{i\in K} {\one_{\cE'_i}}]| \le (2^{kb/2} +4^{k}) \cdot (\eps/n)^{10C} \le 2^{kb/2+1} \cdot (\eps/n)^{10C}$ and we get 
	$(2m)^{k} \cdot 2^{kb/2+1} \cdot (\eps/n)^{10C}\le 1$.
	
	Combining the bounds on both summands we get
	$$
	\E\left[\Big(\sum_{i=1}^{m}{(\one_{\cE_i}-p_i)}\Big)^k\right] 
	\;\le\; 
	(Vk)^{k/2} + 1 
	\;\le\; 
	2\cdot (Vk)^{k/2}.
	$$
	Using $V = \sum_{i=1}^{m} p_i$ we get 
	$$\Pr\left[\sum_{i=1}^{m}\one_{\cE_i} \le V/2 \right] \le 
	 \Pr\left[\Big(\sum_{i=1}^{m}{(\one_{\cE_i}-p_i)}\Big)^k  \ge (V/2)^k \right]  \le 2(Vk)^{k/2} \cdot (V/2)^{-k}$$
Using $V \ge 2\cdot 4^b$ and $k\le 2^b$ we get $\Pr\left[\sum_{i=1}^{m}\one_{\cE_i} \le V/2 \right] \le 2(4k/V)^{k/2} \le 2\cdot(2/2^b)^{k/2}\le  (\eps/n)^{C/4}$.
In the complement event, at least $V/2 \ge 4^b$ of the functions $(f_1)_{T|x}, \ldots, (f_m)_{T|x}$ are non-constant.
\end{proof}

\begin{lemma}\label{lemma:low-deg}
Let $f: \F_2^n \to \pmone$.
Suppose $f(x) = h(x) \cdot (-1)^{g(x)}$ where $h$ is a $k$-junta 
and $g$ is a polynomial of degree-$d$ over $\F_2$.
	If $\D$ fools degree-$d$ polynomials over $\F_2$ with error $\eps$, then $\D$ fools $f$ with error $\eps \cdot 2^{k/2}.$
\end{lemma}
\begin{proof}
Let $J$ be the set of variables on which $h$ depends.
Using the Fourier transform of $h$:
$h(x) = \sum_{S \subseteq J} \hat{h}(S) \cdot (-1)^{\sum_{i\in S} x_i}$
we write $f$ as
$f(x) = \sum_{S \subseteq J} \hat{h}(S) \cdot (-1)^{\sum_{i\in S} x_i + g(x)}$.
Note that $\sum_{i\in S} x_i + g(x)$ is a polynomial of degree-$d$ over $\F_2$ as well, thus we get
\begin{align*} \left|\E_{x\sim U}[f(x)]-\E_{x\sim \D}[f(x)]\right| 
&\le \sum_{S}|\hat{h}(S)| \cdot 
\Big|\E_{x\sim U}[(-1)^{\sum_{i\in S} x_i + g(x)}] - \E_{x\sim \D}[(-1)^{\sum_{i\in S} x_i + g(x)}]\Big| \\
&\le L_1(h) \cdot \eps \le 2^{k/2} \cdot \eps\;.\qedhere\end{align*}
\end{proof}

\begin{proof}[Proof of Lemma~\ref{lemma:Gb}]
First note that $f$ has very small expectation under the uniform distribution
$$
\left|\E_{z\sim U}\Big[f_0(z) \cdot \prod_{j=1}^{m} f_j(x)\Big]
\right| \le (1-2^{-b})^{16^{b}} \ll \frac{\eps}{4}.
$$
using the assumption $b\ge C \log \log(n/\eps)$.
Thus, we need to maintain  low-expectancy under the pseudorandom assignment.
By Lemma~\ref{lemma:aggressive-real}, with probability at least $1-(\eps/n)^{C/4}\ge 1-\frac{\eps}{100}$ after the aggressive random restriction at least $4^{b}$ of the functions $f_1, \ldots, f_m$ remain non-constant. Since $b\ge C \log\log(n/\eps)$ we maintained the low-expectancy under aggressive random restrictions. That is, whenever $4^{b}$ of the functions $f_1, \ldots, f_m$ remain non-constant under restriction, the expected value of the restricted function under the uniform distribution is at most $(1-2^{-b})^{4^b} \ll \eps/4$ in absolute value.

Furthermore, we wish to show that with high probability, except for a set of at most $C\log (n/\eps)$ ``bad variables'' all functions have block-length at most $16$.
Recall that there are at most $2\cdot 16^{2b}$ functions.
The probability that any particular $k$ variables survive is at most $2^{-b k} + (\eps/n)^{10C}$.
Pick $k = C \log(n/\eps)/b \le C \log(n/\eps)$.
The probability that at least $k$ variables in at most $\ell$ functions survive is 
$$
\binom{2\cdot 16^{2b}}{\ell} \cdot \binom{\ell \cdot b}{k} \cdot (2^{-b k} + (\eps/n)^{10C})
\le 2 \cdot 2^{\ell + 9b\ell - bk} \le 2^{10b\ell-bk}
$$
If $\ell \le k/16$, then this probability is at most $2^{-6bk/16} = (\eps/n)^{6C/16} \ll \frac{\eps}{100}$.
This means that, with high probability, there are less than $k$ variables from all functions with more than $16$ effective variables remaining.
Otherwise, there would have been $\ell \le k/16$ functions accountable to a total number of more than $k$ variables that remained alive and effective, under the restriction.

Overall, with probability at least $1-\eps/50$ we are left with the XOR of a small-junta, on at most $t + C\log(n/\eps)$ variables, and an XOR of at least $4^b$ non-constant functions on at most $16$ variables (i.e., a degree $16$ polynomial). Moreover, the restricted function has expected value at most $\eps/4$ in absolute value under the uniform distribution.
Using Claim~\ref{lemma:low-deg} we get that Viola's~\cite{Viola08} or Lovett's~\cite{Lovett08} PRG for low-degree polynomials $\eps/4$-fools the remaining function. Combining all estimates we get that the expected value of the restricted function under our distribution is at most $3\eps/4$ in absolute value which completes the proof.
\end{proof}

\subsection{A thought experiment}\label{sec:4.2}
We are ready to describe the pseudo-random restriction process in full detail.
We start by describing a process that iteratively ``looks'' at the restricted functions in order to decide which pseudorandom restriction to apply next: the one described in Lemma~\ref{claim:assigning-0.9999} or the one from Lemma~\ref{lemma:aggressive-real}.
This ultimately defines a decision tree of random restrictions.
We then show in Section~\ref{sec:4.3} how to transform the adaptive process into a non-adaptive pseudorandom generator that (by definition) does not depend on the function it tries to fool.
Namely, we would generate a pseudorandom string that  fools the function no matter what path was taken in the decision tree.

We start with $m \le n$ blocks of length $b$.
We assume that $m \le 16^b$ (if not set $b= \log_{16}(m)$).
\begin{algorithm}[H]
\caption{an ``adaptive pseudorandom generator''} 
\begin{algorithmic}[1]
\For{$i=0,1,\ldots$} 
\State Let $b_i = b/2^i$.
\If{$b_i \le C\log \log(n/\eps)$} 
apply CHRT's PRG on the remaining coordinates, and Halt!
\EndIf
\If{more than $16^{b_i}$ of the restricted functions are non-constant and depend on at most $b_i$ variables} 
apply $\GMany(b_i,10\log(n/\eps),n,\eps/2)$ from Lemma~\ref{lemma:Gb} on the remaining variables, and Halt!
\Else{ apply the pseudorandom restriction from Lemma~\ref{claim:assigning-0.9999} on the remaining variables.}
\EndIf
\EndFor
\end{algorithmic}\label{alg:adaptive}
\end{algorithm}

Next, we show that the process yields a pseudorandom string fooling $f = \prod_{i=1}^m{f_i(x)}$.
First, note that the process either stops at Step 3 or at Step 4. In both cases we assign all the variables according to some pseudorandom generator, hence all the variables will be assigned by the end of the process.

For $i=0,1, \ldots, $.
Let $T_i$ be the set of coordinates that remain alive at the beginning of the $i$-th iteration.
Denote by $f_j^{(i)}$ the $j$-th function under the restriction at the beginning of the $i$-th iteration. 
Define $\V[f_j^{(i)}]$ to be the set of variables that affect the output of $f_j^{(i)}$.
For example if $f_j^{(i)}$ is a constant function, then $\V[f_j^{(i)}] = \emptyset$.

Let $\Good_i = \{j:  1 \le |\V[f_j^{(i)}]| \le b_i\}$ be the set of functions that depend on some but not more than $b_i$ variables, $\Bad_i = \{j:  |\V[f_j^{(i)}]| > b_i\}$ be the set of functions that depend on more than $b_i$ variables and  $\VBad_{i} = \bigcup_{j\in \Bad_{i} }{\V[f_j^{(i)}]}$.

\begin{claim}\label{claim:VBad}
Let $b_i > C\log \log(n/\eps)$.
Suppose $|\VBad_{i}| \le 10 \log(n/\eps)$ and $|\Good_{i}| \le 16^{b_{i}}$.
Then, with probability at least $1-(\eps/n)$ 
we have $|\VBad_{(i+1)}| \le 10 \log(n/\eps)$.
\end{claim}
\begin{proof}
Under the assumptions we reach Step 5 in Algorithm~\ref{alg:adaptive}. We show that:
\begin{enumerate}
	\item With probability at least $1-\frac{1}{2}(\eps/n)$, at most $5\log(n/\eps)$ of the variables in $\VBad_{i}$ remain alive in Step 5.
	\item With probability at least $1-\frac{1}{2}(\eps/n)$, at most $5 \log(n/\eps)$ new variables are added to $\VBad_{(i+1)}$.
\end{enumerate}
Both claims rely on the fact that  any set of $k \le 5 \log(n/\eps)$ variables remain alive under the pseudorandom restriction in Lemma~\ref{claim:assigning-0.9999} with probability at most $2\cdot 0.0001^k$.

This  first item follows since the probability that more than $5\log(n/\eps)$ variables in $\VBad_i$ survive is at most 
$$\binom{10\log(n/\eps)}{5\log(n/\eps)} \cdot 2\cdot0.0001^{5\log(n/\eps)} \le \tfrac{1}{2} (\eps/n).$$

As for the second item, we start with the case where $b_i \le 2\log(n/\eps)$.
Assume that more than $5\log(n/\eps)$ new variables were added to $\VBad_{(i+1)}$. 
This implies that there is a set of $k = \lceil{5\log(n/\eps)/(b_i/2)\rceil}$ good functions in step $i$ that are accountable to at least $5\log(n/\eps)$ bad variables in step $i+1$.
The latter event happens with probability at most 
$$\binom{16^{b_i}}{k} \cdot \binom{k b_i}{5\log(n/\eps)} \cdot 2\cdot 0.0001^{5\log(n/\eps)} \le 32^{b_i k} \cdot 2 \cdot 0.0001^{5\log(n/\eps)} \le \tfrac{1}{2} (\eps/n)$$
(where we used $k b_i \le b_i+10\log(n/\eps) \le 12\log(n/\eps)$)
which finishes the case $b_i \le 2\log(n/\eps)$.

In the case where $b_i > 2\log(n/\eps)$, we show that with high probability all good functions remain good.
For each individual function, using Markov's inequality
\begin{align*}\Pr[ |\V[f^{(i+1)}_j]| \ge b_i/2] 
&\le \frac{\E\left[\binom{|\V[f^{(i+1)}_j]|}{ \log (n/\eps)}\right]}{\binom{b_i/2}{\log(n/\eps)}}
\le 2\cdot 0.0001^{\log n/\eps} \cdot \frac{\binom{b_i}{\log(n/\eps)}}{\binom{b_i/2}{\log(n/\eps)}}\\
&\le 2\cdot  0.0001^{\log n/\eps} \cdot \frac{(e \cdot b_i/ \log(n/\eps))^{\log(n/\eps)}}{((b_i/2)/ \log(n/\eps))^{\log(n/\eps)}} \\
&= 2\cdot(0.0001 \cdot e\cdot 2)^{\log(n/\eps)} \le \tfrac{1}{2} (\eps/n^2).\end{align*}
Thus, we can apply a union bound and show that all good functions remain good with  probability at least $1-\frac{1}{2}(\eps/n)$.
\end{proof}

Say the process finished. 
We shall assume that $|\VBad_i| \le 10 \log(n/\eps)$ for every iteration $i$ until the process stopped. By Claim~\ref{claim:VBad} this happens with probability at least $1-\log(b)\cdot(\eps/n) \ge 1-\eps/2$ by applying a union bound on the at most $\log(b)$ iterations.
We wish to show that we constructed a pseudorandom string fooling $f$. We consider two cases:
\begin{enumerate}
	\item 
	We stopped on Step~3 at some iteration $i$. If $i=0$ then $m \le 16^{b_0} \le \poly\log(n/\eps)$ and at most $\poly\log(n/\eps)$ variables remain that affect the functions $f_{j}^{(i)}$.
	 Otherwise, since $|\Good_{(i-1)}| \le 16^{2b_i} \le \poly\log(n/\eps)$ and $|\VBad_{(i-1)}| \le 10 \log(n/\eps)$, at most $\poly\log(n/\eps)$ variables remain that affect the functions $f_{j}^{(i-1)}$, and thus at most $\poly\log(n/\eps)$ variables remain that affect the functions $f_{j}^{(i)}$.
Thus, we can write $\prod_{j=1}^{m} f_j^{(i)}$ as a ROBP of width $2w$ and length $\poly\log(n/\eps)$, which is $(\eps/2)$-fooled by the pseudorandom generator from Theorem~\ref{thm:CHRT} using  $\tildeO(\log(n/\eps))$ random bits.
	 
\item We stopped at Step~4 at some iteration $i$.
Certainly, $|\Good_{i}|\le  |\Good_{(i-1)}|+|\Bad_{(i-1)}| \le 16^{2b_i} + 10 \log(n/\eps) \le 2\cdot 16^{2b_i}$.
Thus, we are in the case that was handled in Section~\ref{sec:4.1}, with $t \le 10 \log(n/\eps)$. Indeed, Lemma~\ref{lemma:Gb} guarantees that $\GMany(b_i,10\log(n/\eps), n, \eps/2)$ fools the remaining function with error at most $\eps/2$ using $O(\log(n/\eps))$ random bits.
\end{enumerate}

\subsection{The actual generator}\label{sec:4.3}
Algorithm~\ref{alg:adaptive} described the pseudo-random generator as if we knew whether or not the condition in step 3 holds. 
However, a pseudorandom generator cannot depend on the function it tries to fool.
To overcome this issue, we use the following general observation regarding pseudorandom generators.

\begin{claim}\label{claim:XOR-of-PRGs}
Say there are two families of functions $\mathcal{F}_1$ and $\mathcal{F}_2$ that are both closed under shifts (i.e., closed under XORing a constant string to the input). Say that $\D_1$ is an $\eps$-PRG for $\mathcal{F}_1$ and $\D_2$ is an $\eps$-PRG for $\mathcal{F}_2$  then $\D_1 \oplus \D_2$ is an $\eps$-PRG for $\mathcal{F}_1 \cup \mathcal{F}_2$.\end{claim}
\begin{proof}
Let $f \in \mathcal{F}_1 \cup \mathcal{F}_2$, we show that $\D_1 \oplus \D_2$ fools $f$.
By symmetry assume $f\in \mathcal{F}_1$.
\begin{align*}
\abs{\E_{\substack{x_1\sim \D_1\\ x_2 \sim \D_2}}[f(x_1 \oplus x_2)] -\E_{\substack{z\sim U}}[f(z)]}
&=\abs{\E_{\substack{x_1\sim \D_1\\ x_2 \sim \D_2}}[f(x_1 \oplus x_2)] -\E_{\substack{x_1\sim U\\ x_2 \sim \D_2}}[f(x_1 \oplus x_2)]}\\ 
&\le  \E_{x_2 \sim \D_2} 
\abs{\E_{x_1 \sim \D_1}[f(x_1 \oplus x_2)] -\E_{x_1\sim U}[f(x_1 \oplus x_2)]}\\
&=  \E_{x_2 \sim \D_2} \abs{\E_{x_1 \sim \D_1}[f_{x_2}(x_1)] -\E_{x_1\sim U}[f_{x_2}(x_1)]}\end{align*}
where $f_{y}(x):= f(x\oplus y)$. Since $\mathcal{F}_1$ is closed under shifts, we have that $f_{x_2} \in \mathcal{F}_1$ thus $\D_1$ $\eps$-fools $f_{x_2}$ and we get
$
\E_{x_2 \sim \D_2} \abs{\E_{x_1 \sim \D_1}[f_{x_2}(x_1)] -\E_{x_1\sim U}[f_{x_2}(x_1)]} \le \eps\;.$
\end{proof}

The actual generator would proceed as follows.
\begin{algorithm}[H]
\caption{The Pseudorandom Generator $\GXOR(T, w, b,\eps)$}
\begin{algorithmic}[1]
\Require a set $T\subseteq [n]$ of the ``live'' coordinates, a width $w$,  an integer $b$, a parameter $\eps\in (0,1)$.
\If{$b \le C\log\log(n/\eps)$} \Return $\mathbf{CHRT}(n, n', 2w,  \eps)|_{T}$ for $n' = 2\cdot 16^{2b}\cdot b +10\log(n/\eps)$
\EndIf
\State Let $x := \GMany(b,t,n,\eps)|_{T}$ for $t = 10 \log(n/\eps)$.
\State Pick $T' \subseteq T$, $y\in \pmone^{T\setminus T'}$ according to Claim~\ref{claim:assigning-0.9999}
\State Let $z  := \GXOR(T', w, b/2,\eps/2)$.
\State \Return  $x \oplus \Sel_{T'}(z,y)$.
\end{algorithmic}
\end{algorithm}

\begin{claim}[Proof of Correctness]\label{claim:correctness}
Let $T \subseteq [n]$.
Suppose $f_1, \ldots, f_m$ are functions on disjoint sets of $T$.
Suppose each function depends on at most $b$ variables except for a total of at most $10 \log(n/\eps)$ variables, and the number of non-constant functions is at most $2 \cdot 16^{2b}$.
Then, $\GXOR(T,w, b,\eps)$ fools $f = \prod_{i=1}^{m} f_i$ with error $\eps$.
\end{claim}
\begin{proof}
We prove the claim by induction on $b$.
	If $b \le C \log \log(n/\eps)$ then Theorem~\ref{thm:CHRT} implies correctness.
	If $b > C \log \log(n/\eps)$ then we
	consider the following two cases:
	\begin{enumerate}
		\item If there are more than $16^b$ good functions, then $x = \GMany(b,t,n,\eps)|_{T}$ fools $\prod_{i=1}^{m}f_i$ with error $\eps$.
		\item Otherwise, there are at most $16^b$ good functions and we apply Step~3. According to Claim~\ref{claim:assigning-0.9999}, the average acceptance probability of $f_{T'|y}$ is $\eps/4$ close to that of $f$. 
	Furthermore, with  probability at least $1-\eps/4$ all functions $(f_1)_{T'|y}, \ldots, (f_m)_{T'|y}$ depend on at most $b/2$ variables except for at most $10 \log(n/\eps)$ variables (by Claim~\ref{claim:VBad}).
	In such a case, the number of non-constant functions among $(f_1)_{T'|y}, \ldots, (f_m)_{T'|y}$ is at most $16^{b} + 10\log(n/\eps) \le 2\cdot 16^{b}$.
	Using induction, $z = \GXOR(T',w,b/2,\eps/2)$ fools $f_{T'|y}$ with error $\eps/2$, and we get that $\Sel_{T'}(z,y)$ fools $f$ with error $\eps$.

	\end{enumerate}
	Since we have a pseudorandom generator fooling the function in each case,
	 Claim~\ref{claim:XOR-of-PRGs} shows that $x \oplus \Sel_{T'}(z,y)$ fools $f$ with error $\eps$.
	\end{proof}

\begin{claim}[Seed Length]\label{claim:seed-length}
	The amount of random bits used to calculate $\GXOR([n],w, b,\eps)$ is at most $O(\log(b) + \log\log(n/\eps))^{2w+2} \cdot \log(n/\eps)$.
\end{claim}
\begin{proof}
Unwrapping the recursive calls in the evaluation of $\GXOR([n],w, b,\eps)$ we see that there are at most $\log(b)$ recursive calls to the procedure and that the error parameters are at least $\eps/2^{\log(b)} \ge \eps/n$ in all of them.

We apply the generator from Theorem~\ref{thm:CHRT} only once during these recursive calls, on a ROBP of width-$w$ and length $\poly\log(n/\eps)$. Thus, the application of Theorem~\ref{thm:CHRT} uses at most  $O(\log\log(n/\eps)^{w+2} \log(n/\eps))$ random bits.

The partial assignment from Claim~\ref{claim:assigning-0.9999} uses  at most $O\big(\log(b) + \log\log(n/\eps)\big)^{2w+1} \cdot \log(n/\eps)$ each time we invoke it, and we invoke it at most $\log(b)$ times.

The generator $\GMany$ uses $O(\log (n/\eps))$ random bits each time we invoke it, and we invoke it at most $\log(b)$ times.
\end{proof}

Claims~\ref{claim:correctness} and~\ref{claim:seed-length} complete the proof Theorem~\ref{thm:GXOR} with $\GXOR([n], w, b, \eps)$ as the generator.

\subsection{Pseudorandom generator for read-once polynomials}
Next, we restate and prove Theorem~\ref{thm:read-once poly}.
\thmReadOncePoly*
\begin{proof}
We show that $\GXOR([n],2, \log(8n/\eps),\eps/8n)$ fools any read-once polynomial with error at most $\eps$. Its seed length is $O((\log \log(n/\eps))^6 \cdot \log(n/\eps))$.

	A read-once polynomial can be written as the XOR of AND functions on disjoint variables, i.e., as the XOR of width-$2$ ROBPs on disjoint variables.
	It remains to show that these ROBPs are short.
	Rather, we show that any PRG that $(\eps/8n)$-fools read-once polynomials of degree at most $b = \log(8n/\eps)$ also $\eps$-fools all read-once polynomials.
Let	$$f(x) = \sum_{i=1}^{m} \prod_{j\in B_i} x_j$$ be a read-once polynomial over $\F_2$, where $B_1, \ldots, B_m$ are disjoint subsets of $[n]$.
Without loss of generality let $B_1, \ldots, B_\ell$ be the blocks of length bigger than $b$. 
Let $$f'(x) = \sum_{i=\ell+1}^{m} \prod_{j\in B_i} x_j,$$ be the sum over monomials of degree at most $b$ of $f$.
Let $\D = \GXOR([n],2, \log(8n/\eps),\eps/8n)$. By triangle inequality
\begin{align}\nonumber\abs{\E_{x\sim \D}[f(x)] -\E_{x\sim U_n}[f(x)]} &\le \Pr_{x\sim \D}[f(x) \neq f'(x)] + \Pr_{x\sim U_n} [f(x) \neq f'(x)] + 
\abs{\E_{x\sim \D}[f'(x)] -
\E_{x\sim U_n}[f'(x)]}\\
&\le \sum_{i=1}^{\ell} \Pr_{x\sim \D}[\wedge_{j\in B_i}(x_j=1)] + \sum_{i=1}^{\ell} \Pr_{x\sim U_n}[\wedge_{j\in B_i}(x_j=1)] + \eps/8n\label{eq:read-once polys}
\end{align}
For $i\in \{1,\ldots, \ell\}$, since $|B_i|\ge b$, we have
$\Pr_{x\sim U_n}[\wedge_{j\in B_i}(x_j=1)] \le 2^{-b} \le \eps/8n$. 
As for the distribution $\D$, by monotonicity 
$$\Pr_{x\sim \D}[\wedge_{j\in B_i}(x_j=1)] \le
\Pr_{x\sim \D}[\wedge_{j\in B'_i}(x_j=1)]$$ where $B'_i$ is any arbitrary subset of exactly $b$ variables from $B_i$. Since $\D$ fools degree-$b$ read-once polynomials with error at most $\eps/8n$, and $\wedge_{j\in B'_i}(x_j=1)$ is such a polynomial, we get that $\Pr_{x\sim \D}[\wedge_{j\in B'_i}(x_j=1)]$  is at most $2^{-b} + \eps/8n \le \eps/4n$.
Plugging both bounds into Eq.~\eqref{eq:read-once polys} we get $\abs{\E_{x\sim \D}[f(x)] -\E_{x\sim U_n}[f(x)]}\le (\eps/4n)\cdot \ell + (\eps/8n)\cdot \ell + \eps/8n \le \eps$.
\end{proof}


\newcommand{\mE}{{\mathcal{E}}}
\newcommand{\BRRY}{{\mathbf{BRRY}}}
\newcommand{\Col}{{\mathsf{Col}}}
\newcommand{\FCol}{{\mathsf{FCol}}}
\newcommand{\rest}{|}

\section{Pseudorandom generators for width-3 ROBPs}\label{sec:PRG_3ROBPs}

In this section, we construct pseudorandom generators fooling width-3 ROBPs ({\sf 3ROBPs}, in short) with seed-length $\tilde{O}(\log n)$. For ordered width-3 ROBPs we can guarantee error $1/\poly\log(n)$ using seed-length $\tilde{O}(\log n)$:
\thmMainOrdered*

Note that in comparison, even for constant $\eps > 0$, the best previous generators had seed-length $O(\log^2 n)$ for ordered 3ROBPs. We also get similar improvements for unordered 3ROBPs but with worse dependence on the error $\eps$. 
\thmMainUnordered*

\subsection{Proof overview}

We heavily rely on the pseudorandom restriction from Theorem~\ref{thm:main_two_steps} that assigns $p = 1/\poly \log \log (n)$ of the variables while changing the acceptance probability by at most $1/\poly(n)$. As a first step we assign a constant fraction of the coordinates.

\paragraph{Assigning most of the coordinates.} The first step is rather simple: we apply iteratively $O(1/p)$ times the pseudorandom restriction from Theorem~\ref{thm:main_two_steps} to get the following analog result to Claim~\ref{claim:assigning-0.9999}.
The proof is the same as that of Claim~\ref{claim:assigning-0.9999} and is omitted.

\begin{claim}\label{claim:assigning}
Let $\delta>0$. For all constants $\alpha \in (0,1)$, there is a pseudorandom restriction 
$\rho = (T,y)$ 
 using  $\tilde{O}(\log(n/\delta))$ random bits, changing the acceptance probability of 3ROBPs by at most $\delta$.
Furthermore, $T$ is $(\eps/n)^{\omega(1)}$ biased with marginals $\alpha$ and
$y$ is $(\eps/n)^{\omega(1)}$ biased. 
\end{claim}

Let $B$ be a 3ROBP of length-$n$. First, we claim that after applying the pseudorandom restriction $\rho$ in Claim~\ref{claim:assigning}, with high probability (at least $1-\poly(\eps/n)$), $B\rest_{\rho}$ has a simpler structure in that between any two width-$2$ layers the subprogram has at most $O(\log(n/\eps))$ {\sf colliding layers}. Concretely, we use the following definitions. 

\begin{definition}
Given a ROBP $B$, we call a layer of edges {\sf colliding} if either the edges marked by $-1$ and the edges marked by $1$ collide.
\end{definition}
\begin{definition}
We call a ROBP $B$ a $(w, \ell,m)$-ROBP if $B$ can be written as $D_1 \circ \ldots \circ D_m$, with each $D_i$ being a width $w$ ROBP with the first and last layers having at most two vertices and each $D_i$ having at most $\ell$ colliding layers.
\end{definition}

We show that after applying the pseudorandom restriction $\rho$ in Claim~\ref{claim:assigning}, with high probability the restricting ROBP $B\rest_\rho$ is a $(3,O(\log(n/\eps)),m)$-ROBPs. Now, similar to Section~\ref{sec:assign_all}, we wish to iteratively apply Claim~\ref{claim:assigning}, making the ROBP simpler in each step. We will have one progress measures on the restricted ROBP: the maximal number of colliding layers in a subprogram (denoted $\ell$). We show that the number of colliding layers reduces by a constant-factor in each iteration. To do so, we prove a structural result on $(3,\ell,m)$-ROBPs, showing that such ROBPs can be well-approximated by $(3,\ell, C^\ell)$-ROBPs for some constant $C$. This allows us to not worry about the number of sub-programs and use the number of colliding layers as a progress measure. Applying the restriction and the structure result $O(\log\log n)$ times, we end up with a ROBP where $\ell = O(\log(1/\eps))$. We also show that ROBPs with few colliding layers are fooled by the INW generator. This follows from the results of \cite{BravermanRRY14}.

\subsection{Reducing the length of $(3,\ell,m)$-ROBPs}

Here, we show that $(3,\ell,m)$-ROBPs can be approximated by $(3,\ell,C^\ell)$-ROBPs for some constant $C$. 
A crucial point in the analysis is that we need the approximation to hold not just under the uniform distribution but also under the pseudo-random distribution.
Fortunately, we are able to do so by arguing that the \emph{error} function detecting when our approximation is wrong is itself computable by a conjunction of negations of width $3$-ROBPs with few colliding layers. 

\begin{lemma}[Main Structural Result]\label{lem:structcolliding}
For any $C \geq 1$ the following holds. Any $(3,\ell,m)$-ROBP $B$ can be written as $B' + E$ where $B'$ is a $(3,\ell,C^\ell)$-ROBP 
and either $E\equiv 0$ or for any $x$, $|E(x)| \leq F(x) = \wedge_{i=1}^{C^\ell} (\neg F_i(x))$ where $F_i$ are non-zero events that can be computed by $(3,\ell,1)$-ROBPs on disjoint variables. 
\end{lemma}
We shall also show (in the next claim) that any non-zero event $F_i$ that can be computed by $(3,\ell,1)$-ROBP, happens with probability at least $4^{-(\ell+1)}$ under the uniform distribution. Thus, $\Pr_{x\sim U_n}[\wedge_{i=1}^{C^\ell} (\neg F_i(x))] \le (1-4^{-(\ell+1)})^{C^{\ell}} \le \exp(4^{-(\ell+1)} \cdot C^{\ell})$ which is doubly-exponentially small in $\ell$ provided that $C$ is a large enough constant.

For any vertex $v$ in a ROBP, we denote by $p_v$ the probability to reach $v$ under a uniform random assignment to the inputs.

\begin{claim}\label{claim:pv-large}
In a ROBP with width $w$ and at most $\ell$ colliding layers, every vertex whose $p_v>0$ has $p_v \ge 2^{-(\ell+1)\cdot(w-1)}$.
\end{claim}
We remark that this bound is sharp.
\begin{proof}
	We prove by induction (on the length of the program) that any program with width at most $w$,  exactly $\ell$ colliding layers and exactly $t$ reachable states in the last layer, has $p_v \ge 2^{-\ell\cdot(w-1) -(t-1)}$ for any reachable vertex $v$.
	Without loss of generality all nodes in the program are reachable (otherwise, we remove vertices that aren't reachable).
	
	Consider a program $B$ of length $n$ with  parameters $(t, \ell,w)$.
	Removing the last layer gives a program $B'$ of length $n-1$ with parameters $(t',\ell',w)$. By the induction hypothesis for any $v'$ in the last layer of $B'$ we have $p_{v'}\ge \delta$ for $\delta := 2^{-\ell'\cdot(w-1) -(t'-1)}$.
	
	We perform a case analysis. The following simple bound will be used in all cases. Let $v$ be a vertex in the last layer of $B$. Assume that $e$ edges enter $v$ from vertices in the second to last layer. Then, $p_v \ge \frac{1}{2} \cdot \delta \cdot e$. In particular, since we assumed all vertices are reachable, any vertex in the last layer have $p_v \ge \delta/2$. 

	If $\ell'=\ell$ and $t' = t$, then the last layer of edges in $B$ is regular, i.e., any node in the last layer in $B$ has exactly two ingoing edges. In this case any vertex $v$ in the last layer has  $p_v \ge \frac{1}{2} \cdot \delta\cdot 2= \delta = 2^{-\ell\cdot(w-1) - (t-1)}$.
	
	If $\ell'=\ell$, then $t' \le t$, since there are no collisions in the last layer of edges. Since we already handled the case $t'=t$, we may assume $t'\le t-1$.
	For any vertex $v$ in the last layer we have 
	$p_v \ge \delta/2 
	 \ge \frac{1}{2} \cdot 2^{-\ell'(w-1)-(t'-1)}
	 \ge \frac{1}{2} \cdot 2^{-\ell(w-1)-(t-2)} 
	 = 2^{-\ell(w-1)-(t-1)}$.
	
	If $\ell'<\ell$, then we consider two sub-cases:
	if $t = 1$ then only one vertex is reachable in the last layer and its $p_v$ equals $1$.
	Otherwise, $t\ge 2$ and $t'\le w$ thus $t' \le t+(w-2)$ and 
	 for any  vertex $v$ in the last layer we have  $p_v \ge \delta/2
	 \ge \frac{1}{2} \cdot 2^{-\ell'(w-1)-(t'-1)}
	 \ge \frac{1}{2} \cdot 2^{-(\ell-1)(w-1)-(t + (w-2)-1)} 
	 = 2^{-\ell(w-1)-(t-1)}$.
\end{proof}

We say that two vertices $v$ and $v'$ in a ROBP are {\sf locally-equivalent} if the $1$-edges exiting  $v$ and $v'$ reach the same vertex and the $(-1)$-edges exiting  $v$ and $v'$ reach the same vertex.
We say that a ROBP has {\sf no-redundant vertices} if any vertex in the program is reachable, and there are no locally-equivalent vertices. In the following, without loss of generality we can assume that ROBPs have no-redundant vertices, because we can eliminate unreachable vertices and merge locally-equivalent vertices.

\begin{claim}[Colliding Layers $\implies$ Colliding]\label{claim:XOR or Colliding}
Let $B$ be a 3ROBP with width-2 at the start and finish, at least one colliding layer and no-redundant vertices.
Let $v_{1,1}$ and $v_{1,2}$ be the two start nodes.
Then, there exists a string on which the two paths from $v_{1,1}$ and $v_{1,2}$ collide. 
\end{claim}
\begin{proof}
	First consider the case that $B$ has width $2$. Then, there exists a layer $i$ and a value $b\in \pmone$ such that the two edges marked by $b$ in the $i$-th layer collide. Any string whose $i$-th bit equals $b$ results in colliding paths.

	For the rest of the proof assume that $B$ has a layer with width $3$.	Let $V_1, \ldots, V_{n+1}$ be the layers of vertices in $B$.
	Let $i$ denote the index of the last layer in $B$ with width $3$.
	Since $B$ has width-2 at the end, $i<n+1$.
	
	There are  six edges between $V_{i}$ and $V_{i+1}$: three edges marked with $x_{i}=-1$ and three edges marked with $x_{i}=1$.
	Since $|V_{i+1}|=2$, by the Pigeon-hole principle, there are two edges marked with $x_i = -1$ going to some vertex $v \in V_{i+1}$, and two edges marked with $x_i = 1$ going to some vertex $v'\in V_{i+1}$ ($v'$ is not necessarily different from $v$). 
	These two pairs of edges cannot be starting from the same two nodes in $V_{i}$ since then the two nodes will be locally-equivalent. 
	By renaming the nodes in $V_i$, we can assume that the two edges from $v_{i,1}, v_{i,2} \in V_i$ marked with $-1$ go to $v\in V_{i+1}$ 
	and the two edges from $v_{i,2}, v_{i,3}\in V_i$ marked with $1$ go to $v'\in V_{i+1}$. 
	
	Since $v_{i,2}$ is reachable, there is an input $(x_1, \ldots, x_{i-1})$ that leads from $v_{1,1}$ or $v_{1,2}$ to $v_{i,2}$.
	Without loss of generality, we assume that $v_{i,2}$ is reachable from  $v_{1,1}$.
	Let $\tilde{v} \in V_i$ be the vertex reached by following the same input $(x_1, \ldots, x_{i-1})$ starting from the other start vertex $v_{1,2}$.
	If $\tilde{v} = v_{i,2}$, then we already found a collision. 
	If $\tilde{v} = v_{i,1}$ then for the choice $x_i = -1$ the two paths defined by $(x_1, \ldots, x_i)$  starting from $v_{1,1}$ and $v_{1,2}$ collide on $v \in V_{i+1}$.
	Similarly, if $\tilde{v} = v_{i,3}$, then for the choice $x_i = 1$ the two paths collide on $v' \in V_{i+1}$.
	\end{proof}

\begin{claim}[``First Collisions'' can be detected by 3ROBPs]\label{claim:checkcollision}
Let $B$ be a 3ROBP with 2 vertices at the first layer, denoted $v_{1,1}, v_{1,2}$.
Suppose there are at most $\ell$ colliding layers in $B$ and that there exists a string on which the two paths from $v_{1,1}$ and $v_{1,2}$ collide.
Let $u$ be the first vertex on which a collision can occur, and let $E$ be the event that a collision happened on $u$. Then,  $E$ can be computed by another width-$3$ ROBP with at most $\ell$-colliding layers. 
\end{claim}

\begin{proof}
To simulate whether the paths starting from $v_{1,1}$ and $v_{1,2}$ collide at $u$, we consider the 3ROBP that keeps the {\bf unordered} pair corresponding to the states of the two paths during the computation. In each layer until $u$, we have only states corresponding to $\{0,1\}, \{0,2\}$ or $\{1,2\}$. 
When we reach the layer of $u$ we have two states: ``accept'' (corresponding to a collision on $u$) and ``reject'' (corresponding to anything else).
Observe that any non-colliding layer in the original program defines a non-colliding layer in the new branching program (as a permutation over a finite set also defines a permutation over unordered pairs from this set). Thus, there are at most $\ell$ colliding layers in the 3ROBP computing $E$.
\end{proof}

We are now ready to prove the main structural lemma --  Lemma~\ref{lem:structcolliding}. In the following, we consider branching programs with two initial nodes $v_{1,1}, v_{1,2}$. We interpret the value of the program on input $x$ as its average value on the two paths starting from $v_{1,1}$ and $v_{1,2}$. That is, the program can get value $1, 0$ or $-1$ depending on whether the two paths from $v_{1,1}$ and $v_{1,2}$ accept or not.

Throughout this section we think of the error terms as $\{0,1\}$-indicators (instead of the usual $\pmone$-notation for other Boolean functions). We shall use $A \wedge B$ and $\bar{A}$ to denote the standard AND and negation of these Boolean values.

\begin{lemma}\label{lemma:short programs plus error term}
	Let $B = D_1\circ \ldots \circ D_m$ be a ROBP where each $D_i$ is a width-$3$ ROBP with at most $2$ vertices on the first and last layers.
	Then, for any $j\in \{2, \ldots m\}$ we can write $B(x)$ as the sum of $(D_j \circ  \ldots \circ D_m)(x)$ and an error term $E(x)$, that is bounded in absolute value by $\bar{\FCol_j(x)} \wedge  \ldots \wedge \bar{\FCol_m(x)}$ where $\FCol_{i}(x)$ denotes the event that the two paths in $D_{i}$ collide on input $x$ at the first vertex on which it is possible to collide in $D_i$.
\end{lemma}
\begin{proof}
Assume without loss of generality that no layer of vertices has width-$1$ except for maybe the first. 
For $j=2, \ldots, m$, let $v_{j,1}$ and $v_{j,2}$ be the two nodes at the first layer of the subprogram $D_j$.
If $D_1$ has two nodes at the first layer, then denote them by $v_{1,1}$ and $v_{1,2}$, otherwise denote the single node by $v_{1,1}$. 
Let $x$ be an input to the branching program $B$.
If the two paths defined by $x$ from $\{v_{j,1}, v_{j,2}\}$ collide at some point, then the value of $B(x)$ equals the value of $(D_j \circ \ldots \circ D_m)(x)$.
If the two paths do not collide, then $(D_j \circ \ldots \circ D_m)(x) = 0$, since it is the average of two paths with different outcomes, thus $E(x)= B(x) - (D_j \circ \ldots \circ D_m)(x)$ is at most $1$ in absolute value. 
Furthermore, in such a case, for all $i \in \{j, \ldots, m\}$ it holds that both paths in the subprogram $D_i$ starting from $v_{i,1}$ and $v_{i,2}$ on input $x$ do not collide, i.e., $\FCol_i(x) = 0$.
Overall, we got that $B(x) = E(x) + (D_j \circ \ldots \circ D_m)(x)$, and  $E(x)\neq 0$, it holds that $\bar{\FCol_j(x)} \wedge  \ldots \wedge \bar{\FCol_m(x)} =1$ (i.e., $|E(x)| \le \bar{\FCol_j(x)} \wedge  \ldots \wedge \bar{\FCol_m(x)}$).
\end{proof}

\begin{proof}[Proof of Lemma~\ref{lem:structcolliding}]
Let $B$ be a $(3,\ell,m)$-ROBP $B = D_1 \circ \ldots \circ D_m$.
If $B$ has no colliding layers, then there is nothing to prove since $B$ itself is a $(3,\ell,1)$-ROBP.
If $B$ has colliding layers, then without loss of generality each $D_i$ has at least one colliding layer (since otherwise we can merge subprograms with no colliding layers with their successors or predecessors).
If $m \leq C^\ell$, there is nothing to prove and we can take $B' =B$ and $E=0$. 
Suppose that $m > C^\ell$. 
Let $j = m - C^\ell + 1> 1$. Let $B' = D_j \circ \cdots \circ D_m$ and let $F(x) = \bar{\FCol_j(x)} \wedge  \ldots \wedge \bar{\FCol_m(x)}$ where $\FCol_{i}(x)$ denotes the event that the two paths in $D_{i}$ collide on input $x$ at the first vertex on which it is possible to collide in $D_i$. Then, by the previous claim, we can write $B = B' + E$ where for any input $x$, $|E(x)| \leq F(x)$. 
We argue that this gives the desired decomposition. 
Indeed, By Claim~\ref{claim:checkcollision}, for $i\in \{j, \ldots, m\}$ the event $\FCol_i(x)$ can be computed by a $(3,\ell,1)$-ROBP. Further, by Claim~\ref{claim:XOR or Colliding} each $D_i$ has a possible collision, and thus each $\FCol_{i}$ is a non-zero event. 
\end{proof}

\subsection{PRGs for ROBPs with few colliding layers}
In this section we show that we can $\eps$-fool ordered ROBPs with at most $\ell$-colliding layers with $\tilde{O}(\log(\ell/\eps) \cdot \log (n))$ seed-length. 

\begin{theorem}\label{thm:foolfewcollisions}
For any $\eps > 0$, there is a log-space explicit PRG that $\eps$-fools ordered width $w$-ROBPs with length $n$ and at most $\ell$ colliding layers using seed length
$$O((\log \log n + \log(1/\eps) + \log(\ell) + w) \cdot \log n.$$
\end{theorem}

The above relies on the PRGs for regular branching programs and generalizations of them due to Braverman, Rao, Raz, and Yehudayoff \cite{BravermanRRY14}. 
In the following, we say that a read-once branching program $B$ is {\sf $\delta$-reachable} if for all reachable vertices $v$ in $B$ we have $p_v(B) \ge \delta$, where  
$$
p_v(B):= \Pr_{x \sim U_n}[\text{reaching $v$ on the walk on $B$ defined by $x$}].$$

We start by quoting a result by Braverman, Rao, Raz, Yehudayoff~\cite{BravermanRRY14}.

\begin{theorem}[\cite{BravermanRRY14}]\label{thm:BRRY}
There is a log-space explicit PRG that  $\eps$-fools all 
$\delta$-reachable ROBPs of length-$n$ and width-$w$
 using seed length
$$O(\log \log n + \log(1/\eps) + \log(1/\delta) + \log(w)) \cdot \log n.$$
\end{theorem}

Next, we reduce the task of fooling ROBPs with at most $\ell$-colliding layers to the task of fooling $\delta$-reachable ROBPs.
The reduction is similar to that in \cite{CHRT17}. The main difference is that we simulate a ROBP with width $w$ by a $\delta$-reachable ROBP of width $w+1$ by adding a new sink state that should be thought of as ``immediate stop''. 
This change seems essential in our case, and the reduction from \cite{CHRT17} does not seem to satisfy the necessary properties here.

\begin{lemma}\label{lemma:negligible}
Let $\delta \le 2^{-(w-1)}$.
Let $\D$ be a distribution on $\pmone^n$ that $\eps$-fools all $\delta$-reachable ROBPs of length $n$ and width $w+1$.
Then, $\D$ also fools width-$w$ ROBPs with at most $\ell$ colliding layers with error at most $(\ell w +1)\cdot \eps +  (2^{w} w \ell) \cdot \delta$.
\end{lemma}
\begin{proof}
	Let $\D$ be a distribution on $\pmone^n$ that $\eps$-fools all $\delta$-reachable ROBPs of length-$n$ and width-$w$.
The first observation is that $\D$ also fools prefixes of these programs. This reason is simple: to simulate the prefix of length-$k$ of a $\delta$-reachable ROBP $B$, one can just reroute the last $n-k$ layers of edges in $B$ so that they would ``do nothing'', i.e. that they would be the identity transformation regardless of the values of $x_{k+1}, \ldots, x_n$.

Let $B$ be a length $n$ width-$w$ ROBP with at most $\ell$ colliding layers.
Next, we introduce $B'$, a $\delta$-reachable ROBP of length-$n$ and width-$(w+1)$, that would help bound the difference between 
$$
B(U_n) := \Pr_{x\sim U_n} [B(x)=1] \qquad\text{and}\qquad 
B(\D) := \Pr_{x\sim \D} [B(x)=1]\;,
$$ 
where $U_n$ is the uniform distribution over $\pmone^n$.
Let $B'$ be the the following modified version of $B$.
To construct $B'$ we consider a sequence of $\ell+1$ branching programs $B_0, \ldots, B_\ell$ where $B_0 = B$ and $B' = B_\ell$.
Let $i_1, \ldots, i_\ell$ be the colliding layers in $B$.
For $j=1, \ldots, \ell$ we take $B_j$ to be $B_{j-1}$ except we may reroute some of the edges in the $i_j$-th layer.
We explain the rerouting procedure. 
For $j=1,\ldots, \ell$ we calculate the probability to reach vertices in layer $V_{i_j}$ of $B_{j-1}$. If some vertex $v$ in the $i_j$-th layer has probability smaller than $2^{w-1} \cdot \delta$, then we reroute the two edges going from the vertex $v$ to go to ``immediate stop''. We denote by $\Vsmall$ the set of vertices for which we rerouted the outgoing edges from them.

First, we claim that any reachable vertex $v$ in $B_\ell$ has $p_v  \ge \delta$. 
Let $i_{\ell+1} = n+1$ for convenience. 
We apply induction and show that for $j=0,1, \ldots, \ell$ any vertex reachable by $B_{j}$ in layers $1,\ldots, i_{j+1}$ has $p_v \ge \delta$.
The base case holds because up to layer $i_{1}$ the branching program has no colliding layers and we may apply Claim~\ref{claim:pv-large} to get that $p_v \ge 2^{-(w-1)} \ge \delta$.
To apply induction assume the claim holds for $B_{j-1}$ and show that it holds for $B_{j}$.
The claim obviously holds for all vertices in layers $1, \ldots, i_{j}$ in $B_j$ since we didn't change any edge in those layers going from  $B_{j-1}$ to $B_j$.
Let $v$ be a reachable vertex in layer $i$ where $i_j < i \le i_{j+1}$ in $B_j$. It means that there is a vertex in $v'$ with $p_{v'}(B_j)  \ge 2^{w-1}\cdot \delta$ in the $i_{j}$-th layer of $B_j$ (and also in $B_{j-1}$) and a path going from $v'$ to $v$. 
Looking at the subprogram from $v'$ to $v$ we note that this is a subprogram with no colliding edges (only the first layer has the potential to be colliding, but in a ROBP the first layer can never be colliding as there is only one edge marked by $(-1)$ and only one edge marked by $-1$). By Claim~\ref{claim:pv-large} the probability to get from $v'$ to $v$ is at least $2^{-(w-1)}$.
Thus, the probability to reach $v$ is at least $p_{v'}(B_j) \cdot \Pr[\text{reach $v$}|\text{reached $v'$}] \ge 2^{w-1} \cdot \delta \cdot 2^{-(w-1)} = \delta$.

Next, we bound $|B(U_n)-B(\D)|$ by using the triangle inequality
\begin{equation}\label{eq:hybrid}|B(U_n)-B(\D)|  \le |B(U_n)-B'(U_n)| + |B'(U_n)-B'(\D)| + |B'(\D)-B(\D)|\end{equation}
and bounding each of the three terms separately.
\begin{enumerate}
\item	
The first term is bounded by the probability of reaching one of the nodes in $\Vsmall$ in $B'$ when taking a uniform random walk. This follows since if the path defined by $x$ didn't pass through $\Vsmall$ then we would end up with the same node in both $B$ and $B'$ (since no rerouting affected the path). Each vertex $v$ in $\Vsmall$ has $p_v(B') < 2^{w-1} \cdot \delta$.
By union bound, the probability to pass through $\Vsmall$ is at most $|\Vsmall| \cdot 2^{w-1} \cdot \delta$.
\item The second term is at most $\eps$ since the program $B'$ is $\delta$-reachable.
\item Similarly to the first term, the third term is bounded by the probability of reaching one of the nodes in $\Vsmall$ in $B'$ when taking a walk sampled by $\D$.
	\begin{align*}
		|B'(\D) - B(\D)|
		&\le \Pr_{x\sim \D}[\text{reaching $\Vsmall$ on the walk on $B'$ defined by $x$}] \\
		&\le \sum_{v\in \Vsmall} \Pr_{x\sim \D}[\text{reaching $v$ on the walk on $B'$ defined by $x$}]
	\end{align*}
	However since $\D$ is pseudorandom for prefixes of $B'$, for each $v \in \Vsmall$ the probability of reaching $v$ when walking according to $\D$ is  $\eps$-close to the probability of reaching $v$ when walking according to $U_n$.
	\begin{align*}
		|B'(\D) - B(\D)|
		&\le \sum_{v\in \Vsmall} (\eps + \Pr_{x\sim U_n}[\text{reaching $v$ on the walk on $B'$ defined by $x$}])\\
		&= \sum_{v\in \Vsmall} \left(\eps + p_v(B')\right)
		\le |\Vsmall| \cdot (\eps + 2^{w-1} \delta)
	\end{align*}
\end{enumerate}

Summing the upper bound on the three terms in Eq.~\eqref{eq:hybrid} gives:
\[|B(U_n)-B(\D)| \le |\Vsmall|\cdot (\eps +  2^{w} \delta) + \eps  \le \ell w \cdot (\eps +  2^{w} \delta) + \eps.\;\qedhere\]
\end{proof}

\begin{proof}[Proof of Theorem~\ref{thm:foolfewcollisions}]
	Take $\eps' = \eps/(2(\ell w+1))$ and $\delta = \eps'/2^{w}$. 
	Take the generator from Theorem~\ref{thm:BRRY} with parameters $\delta$ and $\eps'$.
	Applying Lemma~\ref{lemma:negligible}, the error of this generator on the class of ROBPs with width $w$ length $n$ and at most $\ell$ colliding layers is at most 
	$(\ell w+1) \cdot \eps' + (2^{w} \cdot w\cdot \ell) \cdot \delta \le \eps/2 + \eps/2 = \eps$. 
	By Theorem~\ref{thm:BRRY}, its seed length is 
	$$O(\log \log n + \log(1/\eps') + \log(1/\delta) + \log(w)) \cdot \log(n)$$
	which is at most $O(\log \log n + \log(1/\eps) + \log(\ell) + w)\cdot \log(n)$.
	\end{proof}

\subsection{Proof of Theorem~\ref{thm:main-3ROBP-ordered}}
We are now ready to prove our main result on fooling 3ROBPs. Our generator is obtained by applying Claim~\ref{claim:assigning} iteratively $O(\log \log n)$ times and then using a PRG fooling 3ROBPs with at most $O(\poly(1/\eps))$ colliding layers as in Theorem~\ref{thm:foolfewcollisions}. The intuition is as follows.

Let $B$ be a 3ROBP and let $\rho_0$ be a pseudorandom restriction as in Claim~\ref{claim:assigning}. We first show that with probability at least $1-\eps/n$ over $\rho_0$,  $B^0 = B\rest_{\rho_0}$ is a $(3,\ell_0,m)$-ROBP for $\ell_0 = O(\log(n/\eps))$. Let $B^0 = D_1^0 \circ \cdots \circ D_m^0$ where each $D_i^0$ has at most $\ell_0$ colliding layers and begins and ends with width two layers. Let $\rho_1$ be an independent pseudo-random restriction as in Claim~\ref{claim:assigning}. Then $B^1 \equiv B^0\rest_{\rho_1} = D_1^0\rest_{\rho_1} \circ \cdots \circ D_m^0\rest_{\rho_1}$ and it is easy to check that with probability at least $1- 2^{-\Omega(\ell_0)}$, each $D_i^0\rest_{\rho_1}$ has at most $\ell_0/2$ colliding layers. 
Ideally, we would like to apply a union bound over the different $D_i^0$ and conclude that $B^1$ is a $(3,\ell_0/2,m)$-ROBP. 
In the first step, this approach works since $m\le  C^{\ell_0}$ for a large enough constant $C$ (by the definition on $\ell_0$), and we can afford a union bound. 
We get that  with probability at least $1 - 2^{-\Omega(\ell_0)}$, $B^1$ is a $(3,\ell_0/2,m_1)$-ROBP (for some $m_1\le m$).
Continuing this process by induction, at step $i$ we have that $B^{i}$ is a $(3,\ell_0/2^i, m_i)$-ROBP. To carry the union bound in the $i$-th step we need $m_i \le C^{\ell_0/2^i}$, however $m_i$ could be much larger than that. Nevertheless, we know that we can always approximate $B^i$ with a $(3,\ell_0/2^i,C^{\ell_0/2^i})$-ROBP by Lemma~\ref{lem:structcolliding}. This approximation allows us to apply the union bound and conclude that the number of colliding layers in each block decreases by a factor of $2$. We iterate this approach until the maximal number of colliding layers in a subprogram is at most $O(\log (1/\eps))$, and then use the PRG from Theorem~\ref{thm:foolfewcollisions}.

To carry the induction forward as outlined above, we need the following lemma that shows that the error terms simplify as well under the pseudorandom restrictions.

\begin{lemma}\label{lemma:error under restriction}
For any constant $C \ge 20$, there exists $\alpha \in (0,1)$ such that the following holds.
Let $\ell,n\in \N$ be sufficiently large and $m = C^{\ell} \le n$. Let $F = \bar{\FCol_1} \wedge \ldots \wedge \bar{\FCol_m}$ where $\FCol_i(x)$ are non-zero events on disjoint variables computed by $(3,\ell,1)$-ROBPs.
Let $\rho$ be a pseudorandom restriction as in Claim~\ref{claim:assigning} with parameter $\alpha$ and error parameter $\delta \le 1/n^{5}$. 
Then, with probability at least $1 - 2C^{-\ell/2}$, we have $F\rest_\rho \le \bar{\FCol'_i}\wedge \ldots \wedge \bar{\FCol'_{\sqrt{m}}}$ where $\FCol'_i(x)$ are non-zero events on disjoint variables computed by $(3,\ell/2,1)$-ROBPs.
\end{lemma}

\begin{proof}
	First, we show that with high probability, each $\FCol_i$ has at most $\ell/2$ colliding layers under the pseudo-random restriction. 
	To see it, note that any colliding layer that is restricted can be either:
	\begin{itemize}
		\item Assigned to a value that reduces the width of the original program to $2$, and thus the width of $\FCol_i$ to 1, in which case any previous layer in $\FCol_i$ is not affecting its value.
		\item Assigned to a value that applies a permutation on the states of the program, thus reducing the number of colliding layers.
	\end{itemize}  
	In either case, if $k$ colliding layers are unassigned, then $\FCol_i\rest_{\rho}$ can be a computed by a 3ROBP with at most $k$ colliding layers.
	By Claim~\ref{claim:assigning} the probability that less than $\ell/2$ colliding layers are unassigned is at least $1-\binom{\ell}{\ell/2} \cdot \alpha^{\ell/2} -(\eps/n)^{\omega(1)} \ge 1-2^{\ell} \alpha^{\ell/2}$.	Taking a union bound over the $C^{\ell}$ functions $\{\FCol_i\}_{i=1}^{m}$ we get that with probability at least $1- (2C)^\ell \cdot \alpha^{\ell/2} \ge 1-C^{-\ell/2}$ (for a suitable choice of $\alpha$) all functions $\{\FCol_i\rest_{\rho}\}_{i=1}^{m}$ can be computed by 3ROBPs with at most $\ell/2$ colliding layers.
	
	We move to show that with high probability at least $C^{\ell/2}$ of the functions $\FCol_i\rest_{\rho}$ are non-zero.
	We apply the second moment method. 
	Denote by $p_i = \Pr_{z\sim U}[\FCol_i(z)]$ for $i\in [m]$.
	Let $A_1, \ldots, A_m$ be the events that $\{(\FCol_i\rest_{\rho})(y) =1\}_{i=1}^{m}$ respectively, where $\rho$ is the pseudo-random restriction from Claim~\ref{claim:assigning} and $y$ is uniformly distributed. By Claim~\ref{claim:assigning}
	$$\Pr[A_i] = \Pr_{\rho,y\sim U}[(\FCol_i\rest_{\rho})(y) =1] = \E_{z\sim U}[\FCol_i(z)] \pm \delta = p_i \pm \delta,$$
	and by the next lemma, whose proof is deferred to Appendix~\ref{app:composition 3ROBPs}, we get
	\begin{align*}\Pr[A_i \wedge A_j] &= \Pr_{\rho,y}[(\FCol_i\rest_{\rho})(y) \wedge (\FCol_j\rest_{\rho})(y) =1]\\
	&= \E_{z\sim U}[\FCol_i(z)\wedge \FCol_j(z)] \pm  \delta\cdot (n+1)^2 = p_i p_j \pm \delta\cdot (n+1)^2.
	\end{align*}	

\begin{lemma}\label{lemma:fooling H of width-3}
Let $f_1, \ldots, f_k$ be 3ROBPs on disjoint sets of variables of $[n]$. 
Let $H:\pmone^k \to \pmone$ be any Boolean function.
Then, $f = H(f_1, f_2, \ldots, f_k)$ is $\delta \cdot (n+1)^k$-fooled by the pseudorandom restriction in Claim~\ref{claim:assigning}.
\end{lemma}
Thus, the covariance of the two events $A_i$ and $A_j$ is at most $\delta' := \delta ((n+1)^2+3)$.	
Denote by $M = \sum_{i=1}^m{p_i}$.
By Claim~\ref{claim:pv-large} we have that $M \ge 4^{-(\ell+1)} \cdot m \ge 8C^{\ell/2}$ (since $m=C^{\ell}$, $C\ge 20$ and $\ell$ is sufficiently large).
Let $Z = \sum_{i=1}^{m} \one_{A_i}$. 
Then, $\E[Z] \ge M -\delta m$ and by Chebyshev's inequality 
	\begin{align*}\Pr\left[Z<M/2\right] \le 
	\Pr\left[|Z-\E[Z]| \ge M/2-\delta m\right]
	&\le \frac{\Var[Z]}{(M/2-\delta m)^2}\;.\end{align*}
We  bound 
$$\Var[Z] = \sum_{i}{\Var[A_i]} + \sum_{i\neq j}{\cov[A_i,A_j]} \le \sum_{i=1}^{m}(p_i + \delta) + \sum_{i\neq j}\delta' \;\le\; M + m^2\cdot \delta',$$
 which gives $\Pr[Z<M/2] \le (M + \delta' m^2)/(M/2 - \delta m)^2 \le (M+1)/(M/2-1)^2 \le 8/M \le C^{-\ell/2}$ using $\delta \le 1/n^5$.
 In the complement event, at least $M/2 \ge C^{\ell/2}$ of the events $A_1, \ldots, A_m$ occur, and in particular at least $C^{\ell/2}$ of the restricted functions $\{\FCol_i\rest_\rho\}_{i=1}^m$ are non-zero. 

Suppose that at least $C^{\ell/2}$ of the restricted functions $\{\FCol_i\rest_\rho\}_{i=1}^m$ are non-zero, and that all restricted functions has at most $\ell/2$ colliding layers. By the above analysis this happens with probability at least $1-2C^{-\ell/2}$. Under this assumption, we can reduce the number of functions to be exactly $\sqrt{m} = C^{\ell/2}$, resulting in an upper bound on $F\rest_{\rho}$ which we denote by $\bar{\FCol'_{1}(x)} \wedge \ldots \wedge \bar{\FCol'_{\sqrt{m}}(x)} $.
	\end{proof}

We are now ready to prove the main theorem, Theorem~\ref{thm:main-3ROBP-ordered}.
\begin{proof}[Proof of Theorem~\ref{thm:main-3ROBP-ordered}]
Let $C\ge 20$.
Let $\alpha \in (0,1)$ be a constant to be chosen later. Let $\ell_0 = O(\log (n/\eps))$. Let $k$ be a parameter to be chosen later and let $\ell_i = \ell_0/2^i$ for $1 \leq i \leq k$. 

Our generator is as follows. First choose $\rho_0, \rho_1,\ldots,\rho_k$ independent pseudo-random restrictions as in Claim~\ref{claim:assigning} with parameter $\alpha$ and $\delta = \eps/n^{10}$. After iteratively applying the restrictions $\rho_0,\rho_1,\ldots,\rho_k$, we set the remaining bits using the generator from Theorem~\ref{thm:foolfewcollisions} for a parameter $\ell = \ell_k \cdot C^{\ell_k}$ and error parameter $\eps'$ to be chosen later. Let $Y$ be the output distribution of the generator. 

Let $B^0 = B\rest_{\rho_0}$. We first claim that $B^0$ is a $(3,\ell_0,m)$-ROBP with high probability. 
In the following let $X$ be uniformly random over $\pmone^n$. 
\begin{claim}\label{claim:main1}
With probability at least $1-\eps/n$, $B\rest_{\rho_0}$ is a $(3,\ell_0,m)$-ROBP and $\E_{\rho_0,X}[B\rest_{\rho_0}(X)] = \E_X[B(X)] \pm \delta$. 
\end{claim}

For $0 \leq i \leq k$, let $\rho^i \triangleq \rho_0 \circ \cdots \rho_i$. We will show the following claim by induction on $i$.
\begin{claim}\label{claim:main2}
For $0 \leq i \leq k$, with probability at least $1- \frac{\eps}{n} - 4i \cdot C^{-\ell_i}$, $B\rest_{\rho^i}$ can be written as $B^i + E^0 + E^1 + \cdots + E^i$ where $B^i$ is a $(3,\ell_i,C^{\ell_i})$-ROBP and the error terms $E^j$ for $0 \leq j \leq i$ satisfy:
 Either $E_j\equiv 0$ or $|E^j(x)| \leq F^j(x)$ with $F^j(x) = \wedge_{h=1}^{m_j} (\neg F_h^j(x))$ where $F_h^j$ are non-zero events computed by $(3,\ell_i,1)$-ROBPs on disjoint sets of variables and $m_j = C^{\ell_i}$. 

Furthermore, $\E_{\rho^i,X}[B\rest_{\rho^i}(X)] = \E_X[B(X)] \pm (i+1) \delta$. 
\end{claim}

A crucial point in the above is that the functions $F^0,\ldots,F^i$ bounding the error terms are conjunctions of negations of $(3,\ell_i,1)$-ROBPs and there exactly $C^{\ell_i}$ in each of them. 
\begin{proof}
For $i= 0$, the claim follows immediately by applying Lemma~\ref{lem:structcolliding} to $B\rest_{\rho_0}$. Now, suppose the claim is true for $i$. Suppose, we can write $B\rest_{\rho^i} = B^i + {\cal E}^i$, where ${\cal E}^i = E^0 + E^1 + \cdots + E^i$ as in the claim. By the induction hypothesis, this happens with probability at least $1 - \frac{\eps}{n} -  4i \cdot C^{- \ell_i}$. 

Clearly, $B\rest_{\rho^{i+1}} = B^i\rest_{\rho_{i+1}} + {\cal{E}}^i\rest_{\rho_{i+1}}$. Let $B^i = D_1 \circ \cdots \circ D_{m'}$ be a decomposition where each $D_j$ has at most $\ell_i$ colliding layers, starts and ends with width-$2$ layers and $m' \leq C^{\ell_i}$.  

Now, observe that as each $D_j$ has at most $\ell_i$ colliding layers, the probability that at least $\ell_i/2$ of these colliding layers are unfixed under $\rho_{i+1}$ is at most $\binom{\ell_i}{\ell_i/2}\cdot (\alpha^{\ell_i/2} + (\eps/n)^{\omega(1)}) \le 2^{\ell_i}\alpha^{\ell_i/2}$ by Claim~\ref{claim:assigning}. Thus, by a union bound over $1 \leq j \leq m'$, with probability at least $1 - 2^{\ell_i} \alpha^{\ell_i/2} \cdot C^{\ell_i} \ge 1-C^{-\ell_i}$ (for a suitable choice of $\alpha$), over $\rho_{i+1}$, $B^i\rest_{\rho_{i+1}}$ is a $(3,\ell_i/2,C^{\ell_i})$-ROBP.  Now, conditioning on this event, by Lemma~\ref{lem:structcolliding}, we can write $B^i\rest_{\rho_{i+1}}$ as $B^{i+1} + E^{i+1}$, where $B^{i+1}$ is a $(3,\ell_{i+1}, C^{\ell_{i+1}})$-ROBP and $E^{i+1}$ satisfies the conditions of the claim. Thus, with probability at least $1 - \frac{\eps}{n} -  4i \cdot C^{-\ell_{i}} - C^{-\ell_i}$, 
\begin{align*}
B\rest_{\rho^{i+1}} &= B^i\rest_{\rho_{i+1}} + {\cal E}^i\rest_{\rho_{i+1}} \\
&= B^{i+1} + {\cal E}^i\rest_{\rho_{i+1}} + E^{i+1},
\end{align*}
where $B^{i+1}$, and $E^{i+1}$ satisfy the conditions of the claim. 

We just need to argue that ${\cal E}^i\rest_{\rho_{i+1}}$ can be written in the requisite form. To this end, note that for $0 \leq j \leq i$, $|E^j\rest_{\rho_{i+1}}| \leq F^j\rest_{\rho_{i+1}}$. By the induction hypothesis, we either have $E^{j}\equiv 0$ or we can write $|E^{j}| \le F^j = \wedge_{h=1}^{m_j} (\neg F_h^j(x))$ where $F_h^i$ are $(3,\ell_i,1)$-ROBPs on disjoint sets of variables and $m_j = C^{\ell_i}$. We can now apply Lemma~\ref{lemma:error under restriction} to conclude that with probability at least $1 - 2C^{-\ell_{i}/2}$, we can write $F^j\rest_{\rho_{i+1}} = \wedge_{h=1}^{m_j'} (\neg H_h^j(x))$ where $H_h^j$ are non-zero events computed by $(3,\ell_i/2,1)$-ROBPs on disjoint sets of variables and $m_j' = C^{\ell_i/2}$. This satisfies the constraints of the claim. 

Adding up the failure probabilities over the choice of $\rho_{i+1}$, we get the desired decomposition for $i+1$ with probability at least 

$$1 - \frac{\eps}{n} - 4i\cdot C^{-\ell_i} -  C^{-\ell_i} - (i+1)\cdot 2 C^{-\ell_i/2}
\ge 1 - \frac{\eps}{n} - 4(i+1) C^{-\ell_{i+1}}.$$
(since $2C^{-\ell_i} \le  C^{-\ell_{i+1}}$).
The furthermore part follows immediately from Claim~\ref{claim:assigning}. The claim now follows by induction. 
\end{proof}

We are now ready to prove the theorem. By the above claim, we have that with probability at least $1 - \frac{\eps}{n} - 4k C^{-\ell_k}$ over the choice of $\rho_0, \rho_1,\ldots,\rho_k$, we can write
$$B\rest_{\rho^k} = B^k + E^0 + \cdots + E^k,$$
where $B^k$ is a $(3, \ell_k,C^{\ell_k})$-ROBP and $E^0,\ldots,E^k$ can be bounded by functions $F^0,\ldots,F^k$ that are  conjunctions of negations of $C^{\ell_k}$ non-zero events computed by $(3,\ell_k,1)$-ROBPs. 

Note that each such $F^j$ can be written as a width-$4$ ROBP, say $H^j$, by adding an additional layer to compute the conjunction and that the number of collisions in the width $4$ ROBP is at most $\ell_k \cdot C^{\ell_k}$. Therefore, if we let $Y$ be the output distribution of the generator from Theorem~\ref{thm:foolfewcollisions} with $\ell = \ell_k \cdot C^{\ell_k}$ and error parameter $\eps'$, we get that for all $0 \leq j \leq k$, and $X$ uniformly random over $\pmone^n$, 
\begin{align*}
\E[B^k(X)] &= \E[B^k(Y)] \pm \eps'\\
\E[|E^j(Y)|] &\leq \E[H^j(Y)] \leq \E[H^j(X)] + \eps' \leq (1-4^{-(\ell_k+1)})^{C^{\ell_k}} + \eps'
\end{align*}
where we used Claim~\ref{claim:pv-large} to bound $\E[H^{j}(X)]$. Since $C \ge 20$, $\E[|E^j(Y)|] \le \exp(-2^{\ell_k}) + \eps'$.

Combining the above inequalities we get that with probability at least $1 - \frac{\eps}{n} - 4k C^{-\ell_k}$ over the choice of $\rho_0, \rho_1,\ldots,\rho_k$
$$\Big|\E_X [B\rest_{\rho^k}(X)] - \E_Y [B\rest_{\rho^k}(Y)]\Big| \le  \eps' + (k+1)\cdot (\exp(-2^{\ell_k})  + \eps').$$
Finally, as we also have that 
$$\Big|\E_{\rho_0,\ldots,\rho_k}[ B\rest_{\rho^k}(X)] - \E[B(X)] \Big|\le  (k+1) \cdot \delta,$$
we get
$$\Big|\E_{\rho_0,\ldots,\rho_k}[B\rest_{\rho^k}(Y)] - \E[B(X)]\Big| \le \left( (k+1) \cdot \delta\right) \;+\; \left(\frac{\eps}{n} + 4k C^{-\ell_k}\right)\; +\; \left(\eps' + (k+1)\cdot (\exp(-2^{\ell_k})  +  \eps')\right).$$

To get $\ell_k = \log\log\log(n) + \log(1/\eps)$  we set $k = \log(\ell_0/\ell_k) = O(\log \log n)$. Furthermore, setting  $\delta = \eps/n^{10}$ and $\eps' = \eps/4k$, the above error bound becomes
$$\Big|\E_{\rho_0,\ldots,\rho_k}[B\rest_{\rho^k}(Y)] - \E[B(X)]\Big| \le  \eps.$$

Finally, we estimate the seed-length of our generator. Choosing the random restrictions takes $\tilde{O}(\log(n/\delta)) = \tilde{O}(\log(n/\eps))$ random bits. Sampling $Y$ requires seed-length
$$
O(\log \log n + \log(1/\eps') + \log(\ell_k \cdot C^{\ell_k}) + 4) \cdot\log n = O(\log \log n + \log(1/\eps)) \cdot \log n.
$$
Thus, the final seed-length is $\tilde{O}(\log(n/\eps)) + O(\log(1/\eps) (\log n))$. The theorem follows. 
\end{proof}

\subsection{Proof of Claim~\ref{claim:main1}}
\begin{claimNoNum}
With probability at least $1-\eps/n$, $B\rest_{\rho_0}$ is a $(3,\ell_0,m)$-ROBP and $\E_{\rho_0,X}[B\rest_{\rho_0}(X)] = \E_X[B(X)] \pm \delta$. 
\end{claimNoNum}
\begin{proof}
The second part follows from Claim~\ref{claim:assigning}. We are left to prove the first part.

Let $\rho_0 = (T, y)$ be the pseudorandom restriction, where $T\subseteq[n]$ and $y\in \pmone^{[n]\setminus T}$.
Assume there are $L$ colliding layers in $B$ and let $i_1, i_2, \ldots, i_{L}$ be their indices.
For $j \in [L]$, call a layer $i_j$ ``good'' under the choice of $(T,y)$ if $i_j \in [n]\setminus T$ and the edges in the $i_j$-layer of $B$ marked by $y_{i_j}$ collide.

For $j \in \{1,\ldots, L-\ell_0+1\}$ let $\cE_{j}$ be the event that none of layers $\{i_{j}, i_{j+1}, \ldots, i_{j+(\ell_0-1)}\}$ is good.
Recall that $T$ is sampled from a $(\eps/n)^{\omega(1)}$-biased distribution with marginals $\alpha$, and
$y$ is sampled from a $(\eps/n)^{\omega(1)}$-biased distribution.
For $\cE_j$ to happen, we must have a partition  $S_1 \cup S_2 = \{j,j+1, \ldots, j+\ell_0-1\}$ such that all layers $i_{j'}$ for $j'\in S_1$ are in $T$ and all layers $i_{j''}$ for $j''\in S_2$ are in $[n]\setminus T$ but the edged marked by $y_{i_{j''}}$ in the $i_{j''}$-th layer do not collide.
For any fixed $j$ and fixed partition $S_1 \cup S_2 = \{j,j+1, \ldots, j+\ell_0-1\}$, the above event happens with probability at most
$$(\alpha^{|S_1|}+(\eps/n)^{\omega(1)}) \cdot 
(2^{-|S_2|}+(\eps/n)^{\omega(1)}) \le 2\cdot  \alpha^{|S_1|}\cdot 2^{-|S_2|} = 2\cdot  \alpha^{|S_1|}\cdot 2^{-(\ell_0-|S_1|)}
$$
(using $|S_1|+|S_2|=\ell_0 = O(\log(n/\eps))$).
Overall, 
$$\Pr[\cE_j] \le \sum_{S_1 \subseteq\{j, \ldots, j+\ell_0-1\}}(2\cdot  \alpha^{|S_1|}\cdot 2^{-(\ell_0-|S_1|}) = 2 \cdot (\tfrac{1}{2}+\alpha)^{\ell_0} \le \eps/n^2
$$
assuming $\alpha>0$ is a sufficiently small constant and $\ell_0 = c  \log(n/\eps)$ for a sufficiently large constant $c>0$.
By the union bound, $$\Pr[\cE_1 \vee \cE_2 \vee \ldots \vee \cE_{L-\ell_0+1}] \le (L-\ell_0+1) \cdot \eps/n^2 \le \eps/n.$$
Under the event that all $\cE_j$ are false, we get that $B\rest_\rho$  can be written as $D_1 \circ \ldots \circ D_m$ where each $D_i$ is a width-$3$ ROBP with at most $\ell_0$ colliding layers and at most $2$ vertices on the first and last layer.
\end{proof}

\subsection{Pseudorandom generator for unordered 3ROBPs}
In this section, using the recent generator of Chattopadhyay, Hatami, Hosseini, Lovett~\cite{CHHL18}, 
and a Fourier bound by Steinke, Vadhan and Wan~\cite{SteinkeVW17}, 
we show that we can also handle unordered 3ROBPs, thus proving Theorem~\ref{thm:main-3ROBP-unordered}.
\begin{lemma}[Lemma~3.14~\cite{SteinkeVW17}]
	Let $\ell\in \N$ and let $B$ be a width-$w$ ROBP with at most $\ell$ colliding layers. 
	Then, for all $k=1, \ldots, n$ it holds that $L_{1,k}(f) \le O(w^{3}\cdot \ell)^{k}$.
\end{lemma}

\begin{theorem}[Theorem~4.5 \cite{CHHL18}] 
Let $\mathcal{F}$ be a family of $n$-variate Boolean functions closed under restrictions. Assume that for all $f\in \mathcal{F}$ for all $k=1,\ldots, n$, $L_{1,k}(f) \le a\cdot b^k$.
Then, for any $\eps>0$, there exists a log-space explicit PRG which fools $\mathcal{F}$ with error $\eps$, whose seed length is 
$O( \log(n/\eps) \cdot (\log\log(n) + \log(a/\eps)) \cdot b^2)$. 
\end{theorem}
	
\begin{corollary}\label{cor:CHHL}
There is a log-space explicit PRG that $\eps$-fools unordered ROBPs with width $w$, length $n$ and at most $\ell$ colliding layers using seed length
$$O(\log(n/\eps) \cdot (\log \log(n) + \log(1/\eps)) \cdot w^6 \ell^2 )$$
\end{corollary}

\begin{proof}[Proof of Theorem~\ref{thm:main-3ROBP-unordered}]
The proof is essentially the same as that of Theorem~\ref{thm:main-3ROBP-ordered}, where instead of using the generator from Theorem~\ref{thm:foolfewcollisions} to set the bits after the pseudorandom restrictions, we use the generator from the above corollary. The final seed-length has a worse dependence on $\eps$ as we need to set $\ell = C^{\log(1/\eps) + \log\log\log(n)} = \poly(1/\eps) \cdot \poly\log\log(n)$ in Cor.~\ref{cor:CHHL}.
\end{proof}

\subsection{Pseudorandom generator for locally-monotone width-3 ROBPs}
In this section, we construct pseudorandom generators that $\eps$-fool unordered locally-monotone 3ROBPs with seed-length $\tilde{O}(\log(n/\eps))$.
Our dependency on $\eps$ is much better than in Section~\ref{sec:PRG_3ROBPs}, and we get nearly logarithmic (in $n$) seed-length even for error $\eps = 1/\poly(n)$.

We remark that read-once CNFs and read-once DNFs are special cases of  locally-monotone 3ROBPs, hence our result extends the result of GMRTV \cite{GopalanMRTV12} that constructs an $\eps$-PRG for the former classes of functions using seed length $\tilde{O}(\log(n/\eps))$.

\begin{thm}\label{thm:main-3ROBP-locally-monotone}
For any $\eps > 0$, there exists a log-space explicit PRG that $\eps$-fools  unordered locally-monotone 3ROBPs with seed-length $\tilde{O}(\log(n/\eps))$. 
\end{thm}

\begin{proof}
The pseudorandom generator samples a string in $\pmone^n$ as follows:
\begin{enumerate}
	\item Apply the pseudorandom restriction $\rho_0 = (T,y)$ from Claim~\ref{claim:assigning} with $\alpha = \frac{1}{10}$ and $\delta = \frac{\eps}{3}$.
	\item Assign the coordinates in $T$ using the pseudorandom generator $\GXOR(T,3, \ell_0, \eps/3n)$ with $\ell_0 = O(\log(n/\eps))$.
\end{enumerate}
It is clear that the seed-length is $\tilde{O}(\log(n/\eps))$ by Claims~\ref{claim:assigning} and~\ref{claim:seed-length}.

Next, we show that the generator $\eps$-fools unordered locally monotone 3ROBPs.
By Claim~\ref{claim:main1}, with probability at least $1-\eps/n$, $B\rest_{\rho_0}$ is a $(3,\ell_0,m)$-ROBP for $\ell_0 = O(\log(n/\eps))$. Furthermore, $\E_{\rho_0,x\sim U}[B\rest_{\rho_0}(x)] = \E[B] \pm \eps/3$.
Now, observe that since $B$ is locally monotone, for each layer $i$ either $E_{i,1}$ or $E_{i,-1}$ has at most $2$ end-vertices (Lemma~\ref{lemma:identity or collision}). 
 By the proof of Claim~\ref{claim:main1} we get that with probability at least $1-\eps/n$, $B\rest_{\rho_0}$ is of the form $D_1 \circ \cdots \circ D_m$ where each $D_i$ has at most $\ell_0$ layers, and at most $2$ vertices on the first and last layers. In other words, in this case not only does each subprogram have a few colliding layers, it actually has a few layers!
Whenever $B\rest_{\rho_0}$ is of the above form we say that it {\sf simplified} under the restriction.

Thm.~\ref{thm:BDVY} states that whenever $B\rest_{\rho_0}$ simplified, it can be written as $ \sum_{\alpha\in \B^m} {c_{\alpha} \prod_{i=1}^{m}D_{i,\alpha_i}}$
	where $D_{i,\alpha_i}$ are subprograms of $D_{i}$, and 
	$\sum_{\alpha\in \B^m}{|c_\alpha|} \le m\le n$.
By Claim~\ref{claim:correctness}, the distribution $\D = \GXOR(T,3, \ell_0, \eps/3n)$ fools $\prod_{i=1}^{m} D_{i,\alpha_i}$ with error at most $\eps/3n$, for all $\alpha \in \B^m$.
Thus, it fools $B\rest_{\rho_0}$ with error at most $(\eps/3n) \cdot \sum_{\alpha}{|c_\alpha|} \le \eps/3$.
Overall, we get  
\begin{align*}\Big|\E_{\substack{\rho_0 = (T,y), x\sim \D}}\;[B(\Sel_T(x,y))] - \E[B]\Big|&\le \Bigg|\E_{\substack{\rho_0 = (T,y),\\x\sim U_T}}[B\rest_{\rho_0}(x)] - \E[B]\Bigg| \\
&	\quad+ \Pr[B\rest_{\rho_0}\text{ did not simplify}] + \eps/3 \\
&\le  \eps/3 + \eps/n  + \eps/3 \le \eps.\qedhere\end{align*}
\end{proof}

\subsection*{Acknowledgements}
We would like to thank Oded Goldreich and Salil Vadhan for very helpful comments on an earlier version of this manuscript.

\bibliographystyle{alphaabbr}
{\small
\bibliography{bibs.bib}
}

\appendix
\section{Appendices}
\subsection{Proof of Theorem~\ref{thm:CHRTa}}
\label{app:CHRT}

In this section, we view the Boolean functions computed by branching programs as functions $B: \pmone^n \to \B$. For any set $T \subseteq [n]$, this changes the sum $\sum_{S \subseteq T}{|\hat{B}(S)|}$ by a factor of $2$, which we can afford.

Let $B$ be a ROBP of length $n$ and width $w$. Recall that $V_1, \ldots, V_{n+1}$ denote the layers of vertices in $B$.
For a vertex $v \in V_i$ in the branching program we denote by $B_{\to v}$ the sub-branching program ending in the $i$-th layer and having  $v$ the only accepting state.
We denote by $B_{v \to}$ the sub-branching program starting at $v$ and ending at $V_{n+1}$.
Observe that we may express the function computed by the branching program $B$ as a sum of products of these sub-programs, namely 
\begin{equation}\label{eq:bp-decomposition}
\forall{i\in[n]}: \forall{x\in \pmone^n}: B(x) = \sum_{v\in V_i}{B_{\to v}(x) \cdot B_{v \to}(x)}.	
\end{equation}

The main technical result from \cite{CHRT17} is the following theorem:
\begin{theorem}[\protect{\cite[Thm.~2]{CHRT17}}]\label{thm:main_fourier} Let $B$ be an ordered read-once, oblivious branching program of length $n$ and width $w$. Then, 
	$$\forall{k\in [n]}: \;\;\sum_{s: |s|=k} \abs{\widehat{B}(s)} 
	\le 
	O(\log n)^{wk}\;.$$
	\end{theorem}

We are ready to prove a corollary of this theorem, namely Theorem~\ref{thm:CHRTa}.

\begin{theorem}[Thm.~\ref{thm:CHRTa}, restated]
Let $B$ be a width-$w$ length-$n$ ROBP. Let $\eps>0$, $p \le 1/O(\log n)^w$, $k = O(\log(n/\eps))$, and $\D$ be a $\delta_T$-biased distribution over $[n]$ with marginals $p$, where $\delta_T \le p^{2k}$.
Then,
with probability at least $1-\eps$ over $T\sim \D$, 
\[L_1(\tilde{B}) =  \sum_{S \subseteq T}{ |\hat{B}(S)|} \le O((nw)^3/\eps).\]
\end{theorem}

\begin{claim}\label{claim:midlayers}
For all $\beta>0$, the following holds with probability at least $1-\frac{w^2 \cdot n^3}{\beta}$ over $T$:   for all  $v_0$ and $v$ and $1\le j\leq \min\{2k,n\}$:
\begin{equation}\label{eq:lowerlayers}\sum_{s \subseteq T, |s|=j}{ |\hat{B_{v_0 \to v}}(s)|} \le \frac{\beta}{2^j}.\end{equation}
\end{claim}
\begin{proof}
Fix $v_0$ and $v$. Letting $M$ denote the branching program $B_{v_0 \to v}$ we get
$
\sum_{s:|s|=j} \abs{ \widehat{M}(s) } \leq O(\log n)^{wj}
$ from Theorem~\ref{thm:main_fourier}.
Thus,
\begin{align*}
\E_T\left[ \sum_{s:|s|=j} | \widehat{M}(s)| \cdot \one_{\{s \subseteq T\}}\right] = \sum_{s:|s|=j} 
|\widehat{M}(s)| \cdot \Pr_{T}[s\subseteq T] \leq  O(\log n)^{wj} \cdot (p^j + \delta)\leq \frac{1}{2^j}.
\end{align*}
Finally, we conclude by applying the Markov inequality and a union bound, as there is a total of at most $w^2\cdot n^2$ branching programs $B_{v_0 \to v}$ and at most $n$ choices for $j$.
\end{proof}

Theorem~\ref{thm:CHRTa} follows from the next claim which uses Claim~\ref{claim:midlayers} with $\beta=(nw)^3 /\eps$ and $k =  O(\log (n/\eps))$ that ensure $\frac{w^2 \cdot n^3}{\beta}\leq \eps$ and  $\frac{\beta}{2^k}\leq \frac{\eps}{nw}$.
Indeed, with probability at least $1-\eps$, the spectral-norm of $\tilde{B}$ is at most $1+\sum_{j=1}^{k}{\frac{\beta}{2^j}} + (n-k) \cdot \frac{\eps}{nw} \le 2+\beta$.
\begin{claim}
Suppose that $T$ is such that the events in Claim~\ref{claim:midlayers} hold for $\beta, k$ such that $\beta/2^k \le \eps/(nw)$. 
Then for every $j$ such that $k\le j \le n$, 
\begin{equation}\label{eq:higherlayers}\sum_{s \subseteq T, |s|=j}{ |\hat{B}(s)|} \le \frac{\eps}{nw}.\end{equation}
\end{claim}
\begin{proof}
We  prove by induction on $j$ that Eq.~\eqref{eq:higherlayers} holds for all $B_{\to v}$, for any $\ell\in [n+1]$ and $v\in V_\ell$. Note that $B$ itself is of the form $B_{\to v}$ for $v$ being the accept node in the final layer (w.l.o.g. there exists only one such node).   
The case $k \le j \le 2k$ is handled by Claim~\ref{claim:midlayers}, since $\sum_{s\subseteq T:|s|=j}|\hat{B_{\to v}}(s)|  \le	 \frac{\beta}{2^j} \le \frac{\beta}{2^k} \le \frac{\eps}{(nw)^2}$. 
For $j > 2k$ we have: 
\begin{align*}
\sum_{s\subseteq T: |s|=j} |\widehat{B_{\to v}}(s)| 
&\le \sum_{i\in T \cap [\ell]} \sum_{v_0\in V_i} \;\;\sum_{\substack{s_0\subseteq T \cap \{1,\ldots,i-1\}:\\ |s_0|=j-k}}\;\;\sum_{\substack{s_1\subseteq T \cap \{i,\ldots, \ell \}:\\ |s_1|=k, i\in s_1}}\;\; \vert \widehat{B_{\to v_0}}(s_0)\cdot \widehat{B_{v_0\to v}}(s_1)\vert \tag{by \cref{eq:bp-decomposition}}\\
& \le \sum_{i\in T \cap [\ell]} \sum_{v_0\in V_i} \Big(\sum_{\substack{s_0\subseteq T \cap \{1,\ldots,i-1\}:\\ |s_0|=j-k}}\vert\widehat{B_{\to v_0}}(s_0)\vert\Big)\cdot \Big(\sum_{\substack{s_1\subseteq T \cap \{i,\ldots, \ell\}:\\ |s_1|=k, i\in s_1}} \vert \widehat{ B_{v_0\to v}}(s_1)\vert \Big) \\
&\le \sum_{i\in T \cap [\ell]} \sum_{v_0\in V_i} 
\frac{\eps}{nw}
\cdot 
\frac{\eps}{nw}
\le \frac{\eps}{nw} \tag{induction and Claim~\ref{claim:midlayers}}
\end{align*}
This completes the induction, and hence the claim follows.
\end{proof}


\subsection{Restatement of XOR-lemma for functions fooled by small-biased spaces}
\label{app:GMRTV}
In this section we show how Lemma~\ref{lemma:3.2} is a restatement of Thm~4.1 in \cite{GopalanMRTV12}. We recall the following equivalence between having sandwiching approximations with small spectral-norm and being fooled by every small-biased distribution.
\begin{lemma}[\cite{DETT10}]\label{lem:DETT}
Let $f: \pmone^n \to \R$ be a function. Then, the following hold for every $0 < \eps < \delta$:
\begin{itemize}
	\item If $f$ has $\delta$-sandwiching approximations of spectral-norm at most $\delta/\eps$, then for every $\eps$-biased distribution $D$ on $\pmone^n$, $|\E_{x\sim D}[f(x)] - \E[f]| \le \delta$.
	\item If for every $\eps$-biased distribution $D$ on $\pmone^n$, $|\E_{x\sim D}[f(x)] - \E[f]| \le \delta$, then $f$ has $(2\delta)$-sandwiching approximations of spectral-norm at most $1+\delta/\eps$.
\end{itemize}
\end{lemma}
We recall \cite[Thm.~4.1]{GopalanMRTV12}.
\begin{theorem}[\protect{\cite[Thm.~4.1]{GopalanMRTV12}}]\label{thm:GMRTV:4.1}
 	Let $F_1, \ldots, F_k : \pmone^n \to [0,1]$ be functions on disjoint input variables such that each $F_i$ has $\delta$-sandwiching approximation of spectral-norm at most $t$.
 	Let $H:[0,1]^k \to [0,1]$ be a multilinear function in its inputs.
 	Let $h: \pmone^n \to [0,1]$ be defined as $h(x) = H(F_1(x), \ldots, F_k(x))$. Then $h$ has $(16^k\delta)$-sandwiching approximations of spectral-norm at most $4^k(t+1)^k$. 	
\end{theorem}

We translate the domain $[0,1]$ to $[-1,1]$ to get  a restatement of the previous theorem.
\begin{theorem}[\protect{\cite[Thm.~4.1]{GopalanMRTV12}, $\pm1$-version}]\label{thm:GMRTV:4.1-new-version}
 	Let $F_1, \ldots, F_k : \pmone^n \to [-1,1]$ be functions on disjoint input variables such that each $F_i$ has $\delta$-sandwiching approximation of spectral-norm at most $t$.
 	Let $H:[-1,1]^k \to [-1,1]$ be a multilinear function in its inputs.
 	Let $h: \pmone^n \to [-1,1]$ be defined as $h(x) = H(F_1(x), \ldots, F_k(x))$. Then $h$ has $(16^k\delta)$-sandwiching approximations of spectral-norm at most $2^{k+1}(t+4)^k$.
\end{theorem}
\begin{proof} We take $F'_1, \ldots, F'_k$ to be $\frac{F_1+1}{2}, \ldots, \frac{F_{k}+1}{2}$ respectively.
We get that $F'_i$ has $\delta/2$-sandwiching approximations of spectral-norm at most $(t+1)/2$, for all $i\in \{1,\ldots, k\}$.
We take $H':[0,1]^k\to [0,1]$ to be 
$H'(y_1, \ldots, y_k) = \frac{1+H(2y_1 - 1, \ldots, 2y_k-1)}{2}$.
Since $H$ is multilinear, so is $H'$.
By Theorem~\ref{thm:GMRTV:4.1}, we get that $H'(F'_1, \ldots, F'_k)$ 
has $(16^{k} \cdot \delta/2)$-sandwiching approximations of spectral-norm at most $4^k(\frac{t+1}{2}+1)^k$.
Since $H(F_1, \ldots, F_k) = 2\cdot H'(F'_1, \ldots, F'_k) - 1$
we got that $H$ as a $(16^{k} \cdot \delta)$-sandwiching approximations of spectral-norm at most $1+2\cdot4^k(\frac{t+1}{2}+1)^k = 1+2\cdot 2^{k}(t+3)^k \le 2^{k+1} \cdot (t+4)^k$.
\end{proof} 
	
Finally, we restate Lemma~\ref{lemma:3.2} and prove it.
\begin{lemma}
	Let $0 < \eps< \delta\le 1$.
 	Let $F_1, \ldots, F_k : \pmone^n \to [-1,1]$ be functions on disjoint input variables such that each $F_i$ is $\delta$-fooled by any $\eps$-biased distribution.
 	Let $H:[-1,1]^k \to [-1,1]$ be a multilinear function in its inputs.
 	Then $H(F_1(x), \ldots, F_k(x))$ is $(16^k \cdot 2\delta)$-fooled by any $\eps^k$-biased distribution.
\end{lemma}
\begin{proof}
	Using the second item in Lemma~\ref{lem:DETT}, since $F_1, \ldots, F_k$ are $\delta$-fooled by any $\eps$-biased distribution, we have that there exist $2\delta$-sandwiching approximations of spectral-norm at most $1  + \delta/\eps$.
	Thus by Thm.~\ref{thm:GMRTV:4.1-new-version}, $H(F_1, \ldots, F_k)$ has $ (16^k \cdot 2 \delta)$-sandwiching approximations of spectral-norm at most $2^{k+1} \cdot (\delta/\eps + 5)^k$.
	Set $\delta' := 16^k \cdot 2 \delta$ and $\eps' := \delta'/(2^{k+1} \cdot (\delta/\eps + 5)^k)$. Then, $H(F_1, \ldots, F_k)$ has $ \delta'$-sandwiching approximations of spectral-norm at most $\delta'/\eps'$.
	Using the first item in Lemma~\ref{lem:DETT} (noting that $\eps'<\delta'$), any $\eps'$-biased distribution $\delta'$-fools $H(F_1, \ldots, F_k)$.
	A small calculation shows that $\eps' \ge \eps^k$, hence any $\eps^k$-biased distribution also $\delta'$-fools $H(F_1, \ldots, F_k)$.
	\end{proof}


\subsection{Pseudorandom restrictions for the composition of 3ROBPs}\label{app:composition 3ROBPs}
We restate and prove Lemma~\ref{lemma:fooling H of width-3}.
\begin{lemma}
Let $f_1, \ldots, f_k$ be 3ROBPs on disjoint sets of variables of $[n]$. 
Let $H:\pmone^k \to \pmone$ be any Boolean function.
Then, $f = H(f_1, f_2, \ldots, f_k)$ is $(\delta \cdot (n+1)^{k})$-fooled by the pseudorandom partial assignment in Claim~\ref{claim:assigning} with parameter $\delta$.
\end{lemma}
\begin{proof} 
Claim~\ref{claim:assigning} applies Theorem~\ref{thm:main_two_steps} iteratively $t<n$ times with error parameter $\delta/n$. Thus, it suffices to show that under each application of pseudorandom restriction from Theorem~\ref{thm:main_two_steps}  the acceptance probability of $H(f_1, f_2, \ldots, f_k)$ changes by at most $(\delta/n) \cdot (n+1)^{k}$.

Let $\eps := \delta/n$.
Let $V(f_1), \ldots, V(f_k)$ be the sets of variables on which $f_1, \ldots, f_k$ depend.
We write $H$ in the Fourier basis:
$H(y_1, \ldots, y_k) = \sum_{S\subseteq [k]} \hat{H}(S) \cdot \prod_{i\in S}{y_i}$.
Thus,
$H(f_1(x), \ldots, f_k(x)) = \sum_{S\subseteq [k]} \hat{H}(S) \cdot \prod_{i\in S}{f_i(x)}$.
Recall that the pseudorandom assignment in Theorem~\ref{thm:main_two_steps} is composed of two stages:
Let $\eps_1 = \eps/2$ and $\eps_2 = \eps/2n$.
\begin{enumerate}
	\item Pick $T_0 \subseteq [n]$ using a $(\eps_1/n)^{10}$-biased distribution with marginals $1/2$.
	\item Assign the coordinates in $[n]\setminus T_0$ uniformly at random.
	\item
	\begin{enumerate}
	\item Pick $T\subseteq T_0$ using a $\delta_T$-biased distribution with marginals $p = 1/O(\log \log ( n/\eps_2))^{6}$.
	\item  Assign the coordinates in $T_0\setminus T$ uniformly at random.
	\item Assign the coordinates in $T$ using a $(\eps_2/n)^{O(\log \log (n/\eps_2))}$-biased distribution $\Dx$.
	\end{enumerate}
\end{enumerate}
Recall that for a fixed $T_0$,  the bias-function of any program $f_i$ behaves the same under any relabeling of the layers in $[n]\setminus T_0$. We imagine as if these layers are relabeled so that a collision is possible, and denote this relabeled program by $f_i^{T_0}$.
We have $\Bias_{T_0}(f_i)(x) = \E_{y\sim U_{[n]\setminus T_0}}[(f_i^{T_0})_{T_0|y}(x)]$ and similarly
since the sets $V(f_1), \ldots, V(f_k)$ are disjoint
$\Bias_{T_0}(H(f_1, \ldots, f_k))(x) = \E_{y\sim U_{[n]\setminus T_0}}[H((f_1^{T_0})_{T_0|y}(x), \ldots, (f_k^{T_0})_{T_0|y}(x))]$.
By Theorems~\ref{thm:the-bias-trick} and~\ref{thm:BDVY}, with  probability at least $1-\eps_1\cdot k$ the choice of $T_0$ and $y$, we can write each $(f_i^{T_0})_{T_0|y}(x)$ for $i=1, \ldots, k$ as a linear combination of $\prod_{j \in [m_i]}[f_{i,j}]$ where the sum of coefficients in absolute value is at most the number of variables in $f_i$ (i.e., $|V(f_i)|$), and each $f_{i,j}$ is a ROBP on at most $O(\log(n/\eps))$ bits.
Overall with high probability over $T_0, y$ the product $\prod_{i\in S} (f_i^{T_0})_{T_0|y}$ can be written as a linear combination of the functions $\prod_{i\in S} \prod_{j \in [m_i]}[f_{i,j}]$ where the sum of coefficients in absolute values in the linear combination is at most $\prod_{i\in S}{|V(f_i)|}$.
Thus, $H((f_1^{T_0})_{T_0|y}, \ldots (f_k^{T_0})_{T_0|y})$ can be written as a linear combination of XOR of $O(\log(n/\eps))$-length width-3 ROBPs where the sum of coefficients is a most
$$
\sum_{S\subseteq [k]} |\hat{H}(S)| \cdot \prod_{i\in S}|V(f_i)| 
\le \sum_{S\subseteq [k]} 1 \cdot \prod_{i\in S}|V(f_i)| 
= \prod_{i=1}^{k}(1+|V(f_i)|) 
\le (n+1)^k
$$

By Theorem~\ref{thm:pseudorandom-restriction-XOR-short}, each XOR of $O(\log(n/\eps))$-length width-3 ROBPs is $\eps_2$-fooled by the pseudorandom assignment defined by Step~3 above, thus the overall error is at most $\eps_1\cdot k + \eps_2 \cdot (n+1)^k \le \eps \cdot (n+1)^k = (\delta/n) \cdot (n+1)^k$.
\end{proof}

\end{document}